\documentclass[11pt]{article}

\usepackage{amsmath,amssymb,latexsym}
\usepackage[shortcuts,cyremdash]{extdash}
\usepackage{graphics,ifthen,epsfig}
\usepackage{xspace}
\usepackage{relsize}
\usepackage{afterpage}
\usepackage{setspace}
\usepackage[hang,small,bf]{caption}
\usepackage{natbib}
\usepackage{subfigure}
\usepackage{endnotes,etoolbox}
\usepackage{float}
\usepackage[hang,flushmargin]{footmisc} 
\usepackage[top=1in, bottom=1in, left=1in, right=1in]{geometry}
\usepackage{color}
\usepackage{bbm}
\usepackage{lscape}

\bibpunct{(}{)}{;}{a}{}{,}

\newcommand{\ncm}[2]{\newcommand{#1}{\ensuremath{#2}\xspace}}

\ncm{\fa}{A}

\newtheorem{theorem}{Theorem}
\newtheorem{proposition}{Proposition}
\newtheorem{corollary}{Corollary}
\newtheorem{lemma}{Lemma}

\newlength{\tempind}
\newlength{\tempskip}

\newenvironment{proof}
{ \setlength{\tempind}{\parindent} \setlength{\tempskip}{\parskip}
\setlength{\parindent}{0.25in} \setlength{\parskip}{2 pt} \noindent
\textbf{Proof.} }
{%
\setlength{\parindent}{\tempind} \setlength{\parskip}{\tempskip}
\hfill $\blacksquare$ \vspace{2mm} \\}

\let\OLDthebibliography\thebibliography
\renewcommand\thebibliography[1]{
	\OLDthebibliography{#1}
	\setlength{\parskip}{0pt}
	\setlength{\itemsep}{0pt}
	\vspace*{-12pt}
}

\begin{document}

\begin{center}
{\LARGE A Dynamic Model for Managing Volunteer Engagement}\\[12pt]
{\small Baris Ata, Booth School of Business, University of Chicago,
5807 S. Woodlawn Ave, Chicago, IL 60637, USA, baris.ata@chicagobooth.edu, p:773-834-2344} \\ [3pt]
{\small Mustafa H. Tongarlak, Bo\u{g}azi\c{c}i University, Bebek, 34342 Be\c{s}ikta\c{s}/Istanbul, T\"{u}rkiye, tongarlak@boun.edu.tr, p:+90 212 359 6503} \\ [3pt]
{\small Deishin Lee, Ivey Business School at Western University, 1255 Western Rd, London, ON N6G0N1, Canada, dlee@ivey.ca, p:519-661-3288} \\ [3pt]
{\small Joy Field, Boston College Carroll School of Management, 140 Commonwealth Avenue, Chestnut Hill, MA 02467, USA, joy.field@bc.edu, p:617-552-0442} \\ [12pt]

{\small October 23, 2023} \\
\end{center}

\doublespacing
\begin{abstract}
\noindent Non-profit organizations that provide food, shelter, and other services to people in need, rely on volunteers to deliver their services. Unlike paid labor, non-profit organizations have less control over unpaid volunteers’ schedules, efforts, and reliability. However, these organizations can invest in volunteer engagement activities to ensure a steady and adequate supply of volunteer labor. We study a key operational question of how a non-profit organization can manage its volunteer workforce capacity to ensure consistent provision of services. 
In particular, we formulate a multiclass queueing network model to characterize the optimal engagement activities for the non-profit organization to minimize the costs of enhancing volunteer engagement, while maximizing productive work done by volunteers. 
Because this problem appears intractable, we formulate an approximating Brownian control problem in the heavy traffic limit and study the dynamic control of that system. Our solution is a nested threshold policy with explicit congestion thresholds that indicate when the non-profit should optimally pursue various types of volunteer engagement activities. 
A numerical example calibrated using data from a large food bank shows that our dynamic policy for deploying engagement activities can significantly reduce the food bank's total annual cost of its volunteer operations while still maintaining almost the same level of social impact.
This improvement in performance does not require any additional resources -- it only requires that the food bank strategically deploy its engagement activities based on the number of volunteers signed up to work volunteer shifts.


\end{abstract}

\singlespacing

\newpage
\doublespacing
\abovedisplayskip=9pt
\abovedisplayshortskip=0pt
\belowdisplayskip=9pt
\belowdisplayshortskip=9pt


\section{Introduction} \label{sec:intro}

Non-profit organizations (“non-profits”) and, in particular, public charities such as human services organizations that provide food, shelter, and other services to people in need, rely on volunteers to deliver their services. A 2018 Volunteering in America report by the Corporation for National and Community Service estimates that 77.3 million adults volunteered in 2017, working for a total of 6.9 billion hours, and generating \$167 billion in economic value \citep{americorps-2019}.  Yet, many organizations struggle with how to manage this ``charity workforce" in support of their social mission \citep{simmonds-2014}. Thus, in this paper, we study a key operational question of how a non-profit organization can manage its volunteer workforce capacity to ensure consistent provision of services.

In this study, we focus on non-profit operations that provide ongoing services over the long-term to promote community self-sufficiency and sustainability \citep{berenguer-shen-2020}. For example, food banks provide daily meals or food products to their community. In Las Vegas, Three Square food bank delivers food to a service network of nearly 1,400 organizations, schools, after school, and feeding sites.  In 2019, it distributed more than 50 million pounds of food and grocery products, equivalent to more than 41 million meals \citep{threesquare-2021}.  Each day, 200-300 volunteers are needed to make and pack meals and fill orders for these partner organizations \citep{goheen-2018}. For Three Square and other similar types of human services non-profits, these organizations must staff their operations with the appropriate number of volunteers to get the necessary work done on an on-going basis.

However, managing a volunteer workforce presents challenges distinct from those of paid employees. For example, hiring and paying employees entitle the employer to a predetermined number of hours and schedule of work, and promotes continuity in his/her positions.  In contrast, the non-profit has less control over unpaid volunteers’ schedules, efforts, and reliability \citep{ellis-2010}.  While employees follow a schedule that is determined by the capacity needs of the organization, volunteers decide when and how much they would like to work. Even if a volunteer signs up for a particular shift, the volunteer can cancel at any time. As a result, the amount of volunteer capacity that will actually be available in any given period is often uncertain, yet non-profit organizations need consistent working capacity to meet their commitments. Moreover, individuals volunteer at varying frequencies. Some enthusiastic volunteers may work once or twice a week, whereas others may wait years before returning to volunteer again, or never return. At a national level, more than one out of three individuals who volunteer in one year do not volunteer at all in the next year \citep{eisner-etal-2009}.

For these non-profit operations using a volunteer workforce, a key operational priority then becomes how to reduce or mitigate this uncertainty and inconsistency in workforce supply \citep{berenguer-shen-2020}. Although they cannot hire and pay volunteers, non-profits have some managerial levers for increasing volunteer engagement.  For example, Three Square focuses on ensuring that volunteers understand the organization’s mission and the impact volunteers make, and creating an experience that is meaningful for volunteers.  These are operationalized through recruiting activities (e.g., presentations to civic organizations), tours of the Three Square facility, volunteer orientation and training, group pictures, follow up emails to volunteers quantifying their impact, and continued contact with past volunteers to keep Three Square on their radar.  Across a sample of non-profit agencies, \citet{wisner-etal-2005} supports the importance of these and other specific efforts to promote volunteer satisfaction and further involvement with the organization. These activities can improve volunteer engagement, but they differ in the amount of time, effort, and cost to implement. Of course, the non-profit organization is not always looking to increase volunteer work capacity. The goal, as in for-profit organizations, is to match demand for working capacity with supply.

In this paper, we investigate how a non-profit organization should best deploy its volunteer engagement efforts in order to match working capacity demand with volunteer capacity, taking into account how the propensity to volunteer is affected by these efforts. The non-profit organization has work shifts it must staff. Volunteers can sign up to work these shifts using an online system or by calling the volunteer manager. In order to attract and retain volunteers, the organization uses engagement activities such as speaking engagements at organizations (e.g., companies in the area) or electronic communications (e.g., newsletters). These engagement activities are costly to the non-profit. However, if not enough volunteers sign up, shifts are not filled and the non-profit will not be able to produce the desired output (e.g., the target number of meals per day). The non-profit's problem is to determine how much to spend on engagement activities: spending too much wastes valuable resources if there are already enough volunteers, but spending too little increases the risk of not having enough volunteers to do the work, thereby reducing its social impact.

We model the non-profit's volunteer process as a queueing network with multiple classes of volunteers, distinguished by the frequency in which they volunteer. The system consists of many single-class, infinite server queues corresponding to the types of volunteers and one multiclass, single-server queue corresponding to the volunteers who have signed up to work. We formulate an approximating Brownian control problem in the heavy traffic limit and study the dynamic control of that system. Our solution is a nested threshold policy with explicit congestion thresholds that indicate when the non-profit should optimally pursue various types of volunteer engagement efforts. We perform a numerical study using a simulation model with parameters calibrated using process characteristics of a large food bank volunteer operation in the southwestern United States. Using our dynamic policy for deploying engagement activities, we show that the food bank can significantly reduce the total annual cost of its volunteer operation while still maintaining almost the same level of social impact. This improvement in performance does not require any additional resources -- it only requires that the food bank strategically deploy its engagement activities based on the number of volunteers signed up to work volunteer shifts.

Methodologically, this paper contributes to the literature on drift-rate control problems for diffusion models. First, it approximates a queueing network model of volunteer management by a diffusion model, which involves controlling the drift rate of a reflected Ornstein-Uhlenbeck diffusion process. The corresponding Bellman equation involves boundary conditions at the origin and at infinity, both of which are of Neumann type. Initially, we ignore the boundary condition at infinity and solve for a family of initial value problems starting at the origin parametrically. Studying how their derivatives at infinity vary, we choose the unique member of this family of solutions that satisfies the boundary condition of the Bellman equation at infinity. This is the solution of the Bellman equation.

\paragraph{Literature Review.}

This study is related to the stream of literature that studies volunteer engagement (see \citealt{snyder-omoto-2008} and \citealt{wilson-etal-2015} for reviews). \citet[page~131]{vecina-etal-2012} characterizes volunteer engagement as “an energetic and affective connection with their work.”  The volunteer recruitment process, as well as the design and management of the volunteer work experience, impact volunteer engagement (\citealt{haski-leventhal-etal-2011}; \citealt{brayko-etal-2016}; \citealt{einolf-2018}; \citealt{nesbit-etal-2018}). Volunteer engagement determines, to a large extent, volunteer satisfaction and organizational commitment - outcomes important to both the volunteer and non-profit.  These, in turn, are associated with the volunteer's intention to continue volunteering with the non-profit (\citealt{rehnberg-2009}; \citealt{ellis-2010}; \citealt{gazley-2012}; \citealt{vecina-etal-2012}; \citealt{henderson-sowa-2019}).  Thus, we study how the active management of volunteer processes and targeted activities by non-profits can be used to increase volunteer engagement to aid recruitment and retention.

An example of such a targeted activity is to use appeals that are matched to motives for volunteering, such as expressing humanitarian values or making social or career contacts \citep{clary-etal-1994}. \citet{bussell-forbes-2002} suggests creating “recruitment niches” and targeting recruitment activities based on the differing motives of each niche.  When recruiting volunteers, personal appeals are generally more effective than other channels, although not as scalable as presentations to organizations, direct marketing, and online marketing \citep{wymer-starnes-2001}. The pathway to volunteering also affects volunteer retention rates, with volunteers who were recruited directly by someone in the non-profit having the highest retention rates \citep{foster-bey-etal-2007}.

\citet{wisner-etal-2005} finds that schedule flexibility, orientation and training, empowerment, social interaction, reflection (i.e., the volunteer’s understanding of the organization’s mission and the role the volunteer plays in fulfilling the mission), and recognition and symbolic rewards (e.g., appreciation lunches) are antecedents of satisfaction with the volunteer experience, which is positively associated with retention.  Studies of volunteers in human service organizations \citep{cnaan-cascio-1998} and Wikipedia \citep{gallus-2017} also find strong support for symbolic rewards on volunteer satisfaction and retention.  In a survey study of older adult volunteers in human service and environment non-profits, training and role recognition are the two most important organizational facilitators of volunteer retention \citep{tang-etal-2009}.  \citet[page~35]{eisner-etal-2009} states that volunteer satisfaction and retention is supported “by creating an experience that is meaningful for the volunteer, develops skills, demonstrates impact, and taps into volunteers’ abilities and interests.”  More generally, activities that support volunteers with resources to facilitate the work that they do increase satisfaction and promote retention \citep{rehnberg-2009}.  

While some of these volunteer engagement activities, such as symbolic rewards, are low cost, others can be expensive.  For example, establishing mentoring relationships with volunteers promotes retention but can be costly \citep{mcbride-lee-2012}.  Big Brothers Big Sisters (BBBS) invests an average of \$1,000 to ensure an appropriate match between a youth and mentor \citep{brayko-etal-2016}.  After the first year, costs to sustain a match decrease significantly; thus, losing and replacing trained volunteers imposes a substantial effort and financial burden on BBBS.  Hiring a volunteer coordinator is also an expensive proposition but enables non-profits to perform more volunteer engagement activities \citep{urbaninstitute-2004}.  

We distinguish this study from the work on volunteer engagement found in the discipline-specific literature discussed above. The extant literature, mostly empirical and survey-based, focuses on identifying the  factors that increase volunteer engagement, and their subsequent effects on volunteer recruitment and retention. In our study, we build on this research by developing an operating policy that gives guidance to non-profits on when to pursue targeted volunteer engagement activities for the purpose of managing volunteer capacity to meet demand for services.

There is also a body of work in operations management that studies volunteer operations in non-profit organizations. \citet{ata-etal-2019} and \citet{sampson-2006} study the scheduling of volunteers. \citet{manshadi-rodilitz-2021} develops volunteer notification policies for volunteer-based crowdsourcing platforms that take into account volunteer preferences for tasks and sensitivity to excessive notifications. \citet{hewitt-etal-2015} addresses logistical issues such as approaches for consolidating home meal delivery for Meals on Wheels while minimizing operational disruptions and meeting client needs. \citet{urrea-etal-2019} investigates how volunteer experience and congestion at a charity storehouse impact the time to prepare and delays in completing food orders. In the area of supply chain management, \citet{ataseven-etal-2018} examines the role of intellectual capital in food bank supply chain integration. Our paper contributes to this growing stream of operations management literature.

We draw on and contribute to the literature on the dynamic control of queueing systems. In particular, we use the heavy traffic approximation approach pioneered by \citet{Harrison_88} which approximates the original control problem for a queueing system with a diffusion control problem that is easier to analyze; see \citet{Harrison_Wein_89,Harrison_Wein_90} for early examples of this approach. A series of papers have studied a class of problems related to ours -- heavy traffic approximations resulting in drift rate control problems. \citet{ata-harrison-shepp-2005} considers a drift rate control problem on a bounded interval under a general cost of control but no holding costs. \citet{Ata_Thin_Arrival_Streams} builds on \citet{ata-harrison-shepp-2005} and approximates a multi-class make-to-order production system with a drift rate control problem on a bounded interval that has a piecewise linear convex cost of control. Our paper relates methodologically to \citet{Ata_Thin_Arrival_Streams}, however, it also incorporates abandonment and allows for increasingly costly control interventions on a per class basis. Incorporating these two important model features leads to an analysis that is significantly more complex. 

In a series of papers, \citet{Ghosh_2007, Ghosh_2010} extend \citet{ata-harrison-shepp-2005} by incorporating holding costs and allowing the system manager to choose the bounded interval where the process lives endogenously, and introducing abandonments, respectively. See \citet{Ata_MakeToOrder_2009}, \citet{Ghamami_Ward_2013}, \citet{Ata_Tongarlak_Queueing_2013}, and \citet{Sun_2020} for similar formulations with abandonments. In a related paper, \citet{ata-etal-2019} approximates a gleaning operation using a drift rate control problem. The authors derive a nested threshold policy as the optimal staffing policy. The drift rate of their control problem has a different structure because the paper derives an approximation in the many server asymptotic regime. Thus, their analysis does not apply to our context.

In a recent paper, \citet{ata-barjesteh-2019} considers optimal control of a make-to-stock production system. The authors study dynamic pricing, scheduling, and outsourcing decisions simultaneously. They formulate a drift-rate control problem with quadratic cost of control as an approximation. The associated Bellman equation is a Riccati type differential equation, which admits a closed form solution in terms of Airy function. Leveraging this solution, the authors derive a closed form dynamic pricing policy.  \citet{Budhiraja_2011} studies an admission and service rate control problem for a queueing network and derives asymptotically optimal policies in the heavy traffic limit. Several other authors derive asymptotically optimal policies in other relates settings, see for example, \citet{Bell_Williams_2001, bell-williams-2005}, \citet{ata-kumar-2005}, \citet{Ata_Olsen_2009, Ata_Olsen_Queueing_2013}.

A related stream of literature uses Markov decision process formulations to study  service rate or admission control problems; see for example \citet{Crabill_72, Crabill_74}. \citet{Stidham_Weber_89} studies monotone service and arrival rate control policies for a queueing network. \citet{George_Harrison_2001} considers the dynamic service rate control problem for an M/M/1 queue.  Similarly, \citet{Ata_Shneorson_2006} builds on that and solves a dynamic arrival and service rate control problem.
In the context of queueing systems arising in wireless communications applications, \citet{Ata_2005} and \citet{Ata_Zachariadis_2007} solve related service rate control problems  explicitly; also see \citet{Hasenbein_2010} and \citet{Lewis_2013} for other related research.

More specifically, our paper relates to the asymptotic analysis of closed queueing networks with infinite-server queues, see for example \citet{krichagina-puhalskii-1997}, \citet{kogan-lipster-1993}, \citet{kogan-etal-1986}, \citet{smorodinskii-1986}, and \citet{alwan-ata-2020}. \citet{krichagina-puhalskii-1997} studies a closed queueing model containing a single-server queue and an infinite-server queue that has general service time distributions. \citet{kogan-lipster-1993} considers a closed queueing network with one infinite-server and many single-server queues, where only one single-server queue can be in heavy traffic. \citet{kogan-etal-1986} and \citet{smorodinskii-1986} prove limit theorems for related systems. \citet{alwan-ata-2020} considers a closed queueing network with multiple infinite-server queues and multiple single-server queues to study a ride-hailing system. The authors prove a heavy traffic limit theorem, but do not consider the control of that system. 

In this paper, we consider a queueing network with many single-class, infinite-server queues, each corresponding to a particular type of volunteer in repose, and one multiclass single-server queue, corresponding to the volunteers who have signed up to volunteer and are waiting to do so. We formally derive an approximating Brownian control problem in the heavy traffic limit and focus on the dynamic control of that system,  thus contributing to the body of work mentioned above.

\section{Model} \label{sec:model}

The non-profit organization generates social impact by providing goods and/or services to those in need. The working resources that produce the goods or provide the services are volunteers. 
The non-profit organization is served by two types of volunteers: repeat volunteers and one-time volunteers. Repeat volunteers volunteer again after a period of time whereas one-time volunteers volunteer for one shift and never return. There are $J (\tilde{J})$ classes of repeat (one-time) volunteers. Different classes of volunteers may differ in the types of tasks they can perform as a volunteer, their cancellation rates, and the tools available to the non-profit to enhance their engagement. Different classes of repeat volunteers can also differ in their frequencies of volunteering. Upon volunteering, each repeat volunteer enters a repose state, whose duration is random. After the repose period, the volunteer is available and signs up to volunteer again.

\begin{figure}[htbp]
\centering
\includegraphics[scale=0.40]{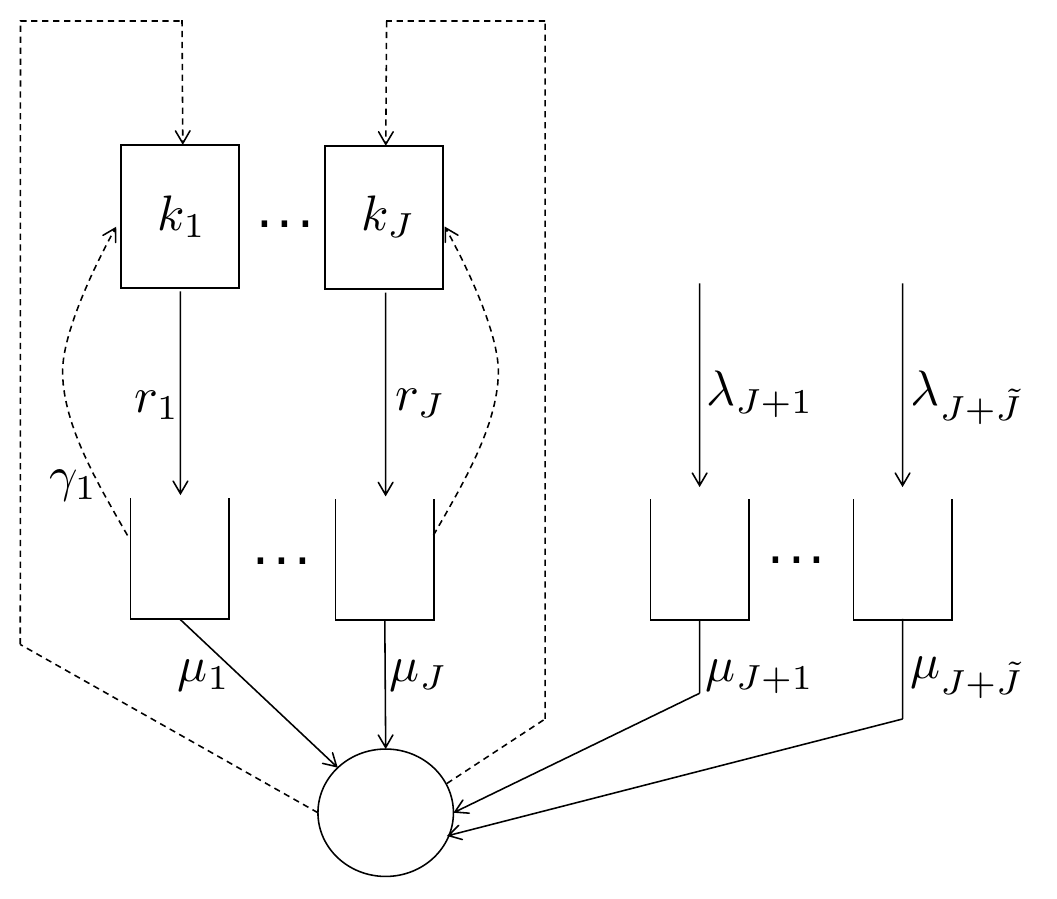}
\caption{A queueing network with $J$ infinite-server nodes, $J$ buffers for repeat volunteer classes, $\tilde{J}$ buffers for one-time volunteer classes, and one multiclass, single-server queue.}
\label{fig:system}
\end{figure}

We model the evolution of the non-profit's volunteer sign-up list using a multiclass queueing network with $J+\tilde{J} $ buffers as depicted in Figure~\ref{fig:system}. The repeat volunteer classes are indexed from 1~through~$J$, and they correspond to $J$ infinite-server nodes, one for each repeat volunteer class, modeling the volunteers of that class who are in repose. There are $k_j$ volunteers in class $j$ ($j = 1, \dots, J$) and $k = \sum_{j =1}^J k_j$ denotes the total number of repeat volunteers. One-time volunteer classes are indexed from $J+1$ through $J+\tilde{J}$, and they each correspond to a buffer in Figure~\ref{fig:system} with the corresponding indices. A class $j$ volunteer spends $m_j$ time units on average during each visit to the non-profit; the reciprocal $\mu_j = 1/m_j$ can be viewed as the non-profit's ``service rate'' of class $j$ volunteers (from the sign-up list). The $m_j$ time units do not include waiting time -- it only includes service time. We assume that the service times are exponentially distributed. 

One-time volunteers of class $j$ arrive according to a Poisson process with rate $\lambda_j$ for $j= J+1, \dots, J+ \tilde{J}$. Similarly, we model the repose time of a repeat volunteer of class $j$ as an exponential random variable with mean $1/r_j$. Crucially, the non-profit can engage in various activities to shorten the repose times (equivalently, increasing the repose exit rates $r_j$ for $j = 1, \dots, J$). These engagement activities can increase the arrival rates of the one-time volunteers, too. To be specific, we assume that the non-profit has $L$ activities (indexed by $l = 1, \dots, L$) to enhance the engagement of volunteers. 
Each activity may involve multiple classes of volunteers. We let $\mathcal{R}_l(\mathcal{S}_l)$ denote the set of repeat (one-time) volunteer classes targeted by activity $l$ for $l = 1, \dots, L$.
Similarly, we let $\mathcal{L}_j$ denote the set of engagement activities that targets class $j$ for $j = 1, \dots, J+\tilde{J}$.
We denote the non-profit's decisions of whether to engage in such efforts by
\begin{align}
\delta_{l} (t) = \Bigg\{ 
     \begin{array}{ll}
          1, & \mbox{if the non-profit engages in activity $l$ at time $t$}, \\
          0, & \mbox{otherwise}.     
     \end{array}
\label{eqn:deltajl}
\end{align}

We assume that engaging in activity $l$ increases the repose exit rate of each repeat volunteer of class $j$ by $\hat{r}_{jl}$ for $j \in \mathcal{R}_l$. Similarly, it increases the arrival rate of one-time volunteers of class $j$ by $\hat{\lambda}_{jl}$ for $j \in \mathcal{S}_l$. The cost rate associated with activity $l$ per unit of time is $F_l$. 
Given the non-profit's engagement decisions $\delta = (\delta_{l})_{l=1}^L$, we model the instantaneous repose exit rate of a repeat volunteer of class $j$, denoted by $R_j (\delta)$, as follows:
\begin{align*}
R_j (\delta) = \Bigg( r_j + \sum_{l \in \mathcal{L}_j} \hat{r}_{jl} \delta_l \Bigg), \,\, j = 1, \dots, J.
\end{align*}
Similarly, given the engagement decisions, the arrival rate of one-time volunteers of class $j$ is given as follows:
\begin{align*}
\Lambda_j (\delta) = \Bigg( \lambda_j + \sum_{l \in \mathcal{L}_j} \hat{\lambda}_{jl} \delta_l \Bigg), \,\, j = J+1, \dots, J+\tilde{J}.
\end{align*}

The volunteers on the sign-up list may cancel. We model this by endowing each class $j$ volunteer on the sign-up list with an exponential random variable with rate $\gamma_j$ which corresponds to their time to cancel, or abandon the sign-up list. We assume that volunteer cancellations are observable to the non-profit, because it follows up with them as their commitment date draws near.
 
When the sign-up list is empty, the non-profit loses potential value that could have been generated. Therefore, the non-profit seeks to keep the sign-up list at a desired length. In particular, if it deems the list to be shorter than ideal, the non-profit can engage in the aforementioned activities to enhance volunteer engagement thereby increasing the repose exit rates or arrival rates of various classes of volunteers. The non-profit also manages the relative magnitudes of the sign-up lists (and hence, also manages the delay experienced by) different volunteer classes by making dynamic scheduling decisions. In what follows, we assume that the non-profit follows a scheduling policy (e.g., First-Come-First-Served (FCFS)) that keeps the sign-up lists for different volunteer classes in fixed proportions over time. Namely, letting $Q_{j}(t)$ denote the number of class $j$ volunteers in the system (waiting or in service), the non-profit strives to make sure
\begin{align}
Q_{j}(t) \approx \frac{x_j}{m_j}\sum_{i=1}^{J+\tilde{J}} m_i Q_i(t), \,\,\, j=1, \dots, J+\tilde{J} \mbox{ and } 
   t \ge 0, \label{eqn:Qj}
\end{align}
where $x_j > 0$ for $j=1, \dots, J+\tilde{J}$ are scheduling policy parameters to be chosen by the non-profit such that $\sum_{j=1}^{J+\tilde{J}}x_j = 1$. To elaborate on the physical meaning of them, note that $\sum_{i=1}^{J+\tilde{J}} m_i Q_i(t)$ denotes the (expected) hours of work for server embodied by the volunteers currently on the sign-up list. Thus, Equation~(\ref{eqn:Qj}) strives to ensure fraction $x_j$ of that work is kept in class $j$. Given $x_j$ for $j = 1, \dots, J+\tilde{J}$, it is straightforward to express each queue length as a function of the total number of volunteers in the system (waiting or in service) as follows: 
\begin{align}
Q_j(t) \approx \frac{x_j \mu_j}{\sum_{i=1}^{J+\tilde{J}} x_i \mu_i} \, \sum_{i=1}^{J+\tilde{J}} Q_i(t), \,\, j = 1, \dots, J+\tilde{J}. \label{eqn:Qj2y}
\end{align}
The class of scheduling policies we consider includes the FCFS service discipline under which we expect to have
\begin{align}
Q_j(t) &\approx \frac{k_j r_j }{\sum_{i=1}^{J} k_i r_i + \sum_{i=J+1}^{J+\tilde{J}} \lambda_i} \, \sum_{i=1}^{J+\tilde{J}} Q_i(t),
     j = 1, \dots, J,  \nonumber \\
Q_j(t) &\approx \frac{\lambda_j }{\sum_{i=1}^{J} k_i r_i + \sum_{i=J+1}^{J+\tilde{J}} \lambda_i} \, \sum_{i=1}^{J+\tilde{J}} Q_i(t), 
     j = J+1, \dots, J+\tilde{J}. \label{eqn:Qj3y}
\end{align}
Thus, setting $x_j = k_j r_j / \mu_j$ for $j=1, \dots, J$ and $x_j = \lambda_j / \mu_j$ for $j = J+1, \dots, J+\tilde{J}$ corresponds to the FCFS discipline.

We model the non-profit's dynamic prioritization decisions of the various classes of volunteers on the sign-up list by the nondecreasing processes $T_j(t)$ for $j=1, \dots, J+\tilde{J}$, where $T_j(t)$ denotes the cumulative amount of time the server spends working on class $j$ volunteers during $[0, t]$. Thus, the cumulative number of class $j$ volunteers served by the non-profit is given by $N_j^s (\mu_j T_j(t))$, where $N_j^s (\cdot)$ is a rate-one Poisson process. The cumulative server idleness process, denoted by $I(\cdot)$, is defined as follows:
\begin{align}
I(t) = t - \sum_{j=1}^{J+\tilde{J}} T_j(t), \,\,\, t \ge 0. \label{eqn:It}
\end{align}

We model the cumulative number of class $j$ (repeat) volunteers signing up to volunteer by time $t$, denoted by $A_j(t)$, as follows: 
\begin{align}
A_j(t) = N_j^a \left(
     \int_0^t \Big( R_j ( \delta(s))     
     (k_j - Q_j(s)) ds\right), \,\, j = 1, \dots, J \mbox{ and } t \ge 0. \label{eqn:Aj}
\end{align}
Similarly, the cumulative number of class $j$ (one-time) volunteers signing up to volunteer by time $t$ is modeled as follows:
\begin{align}
A_j(t) = N_j^a \left(
     \int_0^t \Big( \Lambda_j ( \delta(s)) ds\right), \,\, j = J+1, \dots, J+\tilde{J} \mbox{ and } t \ge 0, \label{eqn:Aj.2}
\end{align}
where $N_j^a$ is a rate-one Poisson process for $j = 1, \dots, J+\tilde{J}$. Recall that class $j$ volunteers on the sign-up list may cancel at rate $\gamma_j$. Thus, we model the cumulative number of cancellations by class $j$ volunteers on the sign-up list by time $t$, denoted by $\Gamma_j(t)$, as follows:
\begin{align}
\Gamma_j(t) = N_j^b \left( \int_0^t \gamma_j Q_j(s) ds \right), \label{eqn:Gammaj}
\end{align}
where $N_j^b$ is a rate-one Poisson process; and $N_j^a$ $N_j^b$, and $N_j^s$ for $j = 1, \dots, J+\tilde{J}$ are mutually independent. When volunteers cancel, idleness could result or the non-profit organization may increase engagement activity, incurring higher cost. These costs are captured in the idleness penalty and increased engagement activity cost. In practice, it seems that volunteers chose to cancel for personal reasons rather than the delay, therefore, we do not include a cancellation cost above the lost throughput and increased engagement costs.

Assuming the sign-up list is empty initially, i.e., all volunteers are in repose, we characterize the evolution of the class $j$ queue length, i.e., the number of class $j$ volunteers on the sign-up list, as follows: For $j=1, \dots, J+\tilde{J}$, and $t \ge 0$,
\begin{align}
Q_j(t) = A_j(t) - N_j^s(\mu_j T_j(t)) - \Gamma_j(t). \label{eqn:Qj2x}
\end{align}

To facilitate the analysis to follow, define the non-profit's cumulative engagement controls as follows: For $l = 1, \dots, L$,
\begin{align}
\Delta_{l}(t) = \int_0^t \delta_{l}(s) ds, \,\,\, t\ge0. \label{eqn:Deltajl}
\end{align}
Letting $\Delta = (\Delta_{l})$ and $T = (T_j)$, we denote the non-profit's policy by $(T, \Delta)$ and call it admissible if it is non-anticipating and the following hold: For $j=1, \dots, J+\tilde{J}$ and $l = 1, \dots, L$,
\begin{align}
&T, \Delta, I \mbox{ are nondecreasing with } T(0) = \Delta(0) = I(0)=0, \label{eqn:TDI} \\
&Q_j(t) \in [0, k_j], \,\, j=1, \dots, J,  \label{eqn:Qj3x} \\
&Q_j(t) \ge 0, \,\, j = J+1, \dots, J+\tilde{J}, \label{eqn:Qj3x.1} \\
& \int_0^t \mathbbm{1}_{\{ \sum_{j=1}^{J+\tilde{J}} Q_j(s) > 0 \} } dI(s) = 0, \label{eqn:idlenessincrease}
\end{align}
where the last constraint indicates the server idleness increases only if the sign-up list is empty. That is, the service policy is work conserving.

Given an admissible policy $(T, \Delta)$, the associated cumulative cost incurred up to time $t$ is given as follows: 
\begin{align}
C(t) = 
     \int_0^t \sum_{l=1}^{L} F_l \delta_l(s) ds    
     +pI(t), \label{eqn:C}
\end{align}
where $p>0$ is the cost rate of idleness. The non-profit's problem is to choose an admissible policy $(T, \Delta)$ so as to
\begin{align}
\mbox{Minimize } \overline{\lim}_{t \rightarrow \infty} \frac{1}{t} \mathbb{E}[ C(t) ] \,\,\mbox{ subject to } (\ref{eqn:Qj}).
\end{align}

In words, the non-profit makes scheduling and engagement effort decisions dynamically to minimize long-run average costs, which include the costs of enhancing volunteer engagement and the cost of idling, i.e., having an empty sign-up list. This problem appears intractable. Therefore, we study this system in an asymptotic regime where the number of (repeat) volunteers and the arrival rates of one-time volunteers get large, and derive a tractable approximation.

\section{An Approximating Brownian Control Problem} \label{sec:approx}

As mentioned earlier, the problem formulated in Section~\ref{sec:model} is not tractable in its exact form. Therefore, by considering a sequence of closely related systems indexed by $n$, we formulate an approximate, yet far more tractable approximation. A superscript $n$ will be attached to the quantities of interest in the $n^{th}$ system. We assume that $k_j^n = \hat{k}_j n$ for $j = 1, \dots, J$.
Thus, $k^n = \sum_{j=1}^{J} \hat{k}_j n$, where $\hat{k}_j$ is a positive constant for $j = 1, \dots, J$. We also assume that for $l=1, \dots, L$ 
\begin{align}
F_l^n = \sqrt{n} F_l, \label{eqn:costF}
\end{align} 
and for $j \in \mathcal{R}_l$ that
\begin{align}
\hat{r}_{jl}^n = \frac{\hat{r}_{jl}}{\sqrt{n}},  \label{eqn:etajln-repeat}
\end{align}
and for $j \in \mathcal{S}_l$ that
\begin{align}
\hat{\lambda}_{jl}^n = \sqrt{n} \, \hat{\lambda}_{jl},  \label{eqn:etajln-onetime}
\end{align}
where $F_l$, $\hat{r}_{jl}$, and $\hat{\lambda}_{jl}$ are positive constants. Then, letting $\alpha_j >0$ denote a constant for $j = 1, \dots, J$, we make the following assumption, which lends itself to the Brownian approximation.

\paragraph{Heavy Traffic Assumption.} For $n \ge 1$, we have that
\begin{align}
   r_j^n &= r_j - \frac{\alpha_j}{\sqrt{n}},  \,\,\, j = 1, \dots, J, \label{eqn:rjn} \\
   \lambda_j^n &= n \lambda_j - \sqrt{n} \alpha_j, \,\,\, j = J+1, \dots, J+ \tilde{J}, \label{eqn:lambdajn} \\
   \mu_j^n &= n \mu_j, \,\,\, j =1, \dots, J+ \tilde{J}. \label{eqn:mujn}
\end{align}
Note that $\alpha_j$ is a measure of excess capacity allocated to class $j$.
Additionally, the following balanced load condition holds:
\begin{align}
\sum_{j=1}^J \frac{r_j \hat{k}_j}{\mu_j} + \sum_{j=J+1}^{J+ \tilde{J}} \frac{\lambda_j}{\mu_j} =1. \label{eqn:balancedload}
\end{align}

Note that letting  
\begin{align}
\theta_0 = -\left( \sum_{j=1}^J \frac{\hat{k}_j \alpha_j}{\mu_j} + \sum_{j=J+1}^{J+\tilde{J}} \frac{ \alpha_j}{\mu_j}  \right) 
   <0 \,\,\, \mbox{and} \,\,\, 
   \rho^n = \sum_{j=1}^J \frac{r_j^n \Halfspace k_j^n}{\mu_j^n}
   +\sum_{j=J+1}^{J+\tilde{J}}  \frac{\lambda_j^n}{\mu_j^n}, \,\, n \ge1, \label{eqn:theta0}
\end{align}
we conclude from the heavy traffic assumption that $(\rho^n -1) \sqrt{n} \rightarrow \theta_0 < 0$ as $n \rightarrow \infty$. As the reader will see below, $\theta_0$ corresponds to the drift parameter of the underlying Brownian motion in our workload formulation (when the non-profit does not engage in any activities), and $\rho^n$ is the traffic intensity of the $n^{th}$ system.

In other words, we study a large balanced-flow system, focusing primarily on the non-profit's sign-up list. 
Under the heavy traffic assumption, the queue lengths at the multiclass, single-server queue, i.e., $Q_j^n(\cdot)$ for $j=1, \dots, J+\tilde{J}$, are expected to be of order $\sqrt{n}$, whereas the number of volunteers in repose is expected to be of order $n$; see for example \citet{kogan-lipster-1993}. Thus, we scale them accordingly. Note that the presence of infinite-server queues leads to the scaling below (\ref{eqn:Zjn}), which differs from the usual scaling for closed networks of singe-server queues in heavy traffic. In that setting, queue lengths are of the same order of magnitude as the total number of jobs in the system; see for example, \citet{harrison-etal-1990}, \citet{harrison-wein-1990}, \citet{chevalier-wein-1993}, and \citet{ata-etal-2020}. Here, we define the scaled queue-length process as follows: 
\begin{align}
Z_j^n(t) = \frac{Q_j^n(t)}{\sqrt{n}} \,\,\, \mbox{for} \,\, j = 1, ..., J \mbox{ and } t \ge 0. \label{eqn:Zjn}
\end{align}
For $j=1, ..., J$, note that $r_j \hat{k}_j/ \mu_j$ denotes the long-run average fraction of time the non-profit should spend serving class $j$ (repeat) volunteers. Similarly, $\lambda_j/\mu_j$ denotes the average fraction of time the non-profit should spend serving class $j$ (one-time) volunteers for $j= J+1, \dots, J+\tilde{J}$.
We then define the centered and the scaled allocation processes $Y_j^n$ and the scaled idleness process $U^n$ as follows: 
\begin{align}
Y_j^n(t) &= \sqrt{n} \left( \frac{r_j \hat{k}_j}{\mu_j} t - T_j(t) \right), \,\,\, j = 1, \dots, J, \label{eqn:Yjn} \\
Y_j^n(t) &= \sqrt{n} \left( \frac{\lambda_j}{\mu_j} t - T_j(t) \right), \,\,\, j = J+1, \dots, J+\tilde{J},\label{eqn:Yjn.2} \\
U^n(t) &=\sqrt{n}\, I(t), \,\, t \ge 0. \label{eqn:Unt}
\end{align}

Recall from the heavy traffic assumption that $\mu_j^n = n \mu_j$ so that one unit of server idleness corresponds to $O(n)$ services that could have been completed. Therefore, we assume $p^n = np$ and define the scaled cost process, denoted by $\xi^n$, as follows:
\begin{align}
\xi^n(t) = \frac{C^n(t)}{\sqrt{n}}, \,\,\, t \ge 0. \label{eqn:xin}
\end{align}
Using the strong approximations \citep{csorgo-horvath-1993} and following \citet{Harrison_88} and \citet{Harrison2000}, the Online Supplement Section~A formally derives the approximating Brownian control problem as $n$ gets large. In the approximating Brownian control problem, the processes $\xi^n, Z^n, Y^n$, and $U^n$ are replaced with their formal limits $\xi, Z, Y$, and $U$, which jointly satisfy the following:

\begin{align}
&\xi(t) = \sum_{l=1}^{L} F_l \Delta_l(t)  
     + p U(t), \,\, t \ge 0, \label{eqn:BCP-xi} \\
&Z_j(t) = X_j(t) + \hat{k}_j  \sum_{l \in \mathcal{L}_j} \hat{r}_{jl} \Delta_{l}(t)  - \hat{k}_j \alpha_j t  
     - \int_0^t (r_j  + \gamma_j) Z_j(s) ds + \mu_j Y_j(t), \,\, j = 1, \dots, J, \,\, t \ge 0, \label{eqn:BCP-Zj1} \\
&Z_j(t) = X_j(t) + \sum_{l \in \mathcal{L}_j} \hat{\lambda}_{jl} \Delta_{l}(t)  \!-\! \alpha_j t  
     - \int_0^t \gamma_j Z_j(s) ds + \mu_j Y_j(t), \,\, j = J\!+\!1, \dots, J+\tilde{J}, \,\, t \ge 0, \label{eqn:BCP-Zj1.2} \\
&Z_j(t) \ge 0, \,\, j=1, \dots , J+\tilde{J}, \,\, t \ge 0, \label{eqn:BCP-Zj2} \\
& Z_j(t) = \frac{x_j}{m_j} \sum_{l=1}^{J} m_l Z_l(t), \,\, j=1, \dots, J+\tilde{J}, \,\, t\ge 0, \label{eqn:BCP-Zj3} \\
& \delta_{l}(t) \in [0,1] \,\, \mbox{ and } \,\, \Delta_{l}(t) = \int_{0}^{t} \delta_{l}(s) ds
     \,\, \mbox{ for } \,\,  l=1, \dots, L, \,\,\, t \ge 0, \label{eqn:BCP-Zj4} \\
&U(t) = \sum_{j=1}^{J+\tilde{J}} Y_j(t), \,\, t \ge0, \label{eqn:BCP-L1} \\
&U \mbox{ is nondecreasing with } U(0) = 0, \label{eqn:BCP-L2} \\
& \int_0^{\infty} 1_{ \left\{ \sum_{j=1}^{J+\tilde{J}} Z_j(t) > 0 \right\} } dU(t) =0, \label{eqn:BCP-workconserving}
\end{align}
where the last constraint expresses the work-conserving property of the non-profit's service policy. That is, the idleness does not increase unless the sign-up list is empty. Moreover, $X_j(t)$ is a $(0, \sigma_j^2)$ Brownian motion where $\sigma_j^2 = 2 r_j \hat{k}_j$ for $j=1, \dots, J$ and $\sigma_j^2 = 2 \lambda_j$  for $j=J+1, \dots, J+\tilde{J}$.

In the approximating Brownian control problem (BCP), the non-profit strives to
\begin{align}
\mbox{Minimize } \overline{\lim}_{t \rightarrow \infty} \frac{1}{t} \mathbb{E} [ \xi(t) ] \,\,\, \mbox{subject to} \,\,\,
     (\ref{eqn:BCP-xi})-(\ref{eqn:BCP-workconserving}). \label{eqn:BCPobj}
\end{align}

\noindent \textbf{Remark.} Equation~(\ref{eqn:BCP-Zj4}) allows $\delta(\cdot)$ to take fractional values. This is done for technical convenience. However, as the reader will see below, the optimal policy we ultimately propose sets $\delta(t) \in \{ 0, 1 \}$ for all $j, t$, cf. Equation~(\ref{eqn:thetas}).

In what follows, we simplify the Brownian control problem and arrive at the so-called equivalent workload formulation. To this end, define the workload process, denoted by $\{ W(t), t \ge0 \}$ as follows:
\begin{align}
W(t) = \sum_{j=1}^{J+\tilde{J}} m_j Z_j(t), \,\,\, t\ge0, \label{eqn:W}
\end{align}
which can be interpreted as the expected total hours of work for the server at time $t \ge0$.
Also, note that Equation~(\ref{eqn:BCP-Zj3}) corresponds to the following equation:
\begin{align}
Z_j(t) = x_j \mu_j W(t), \,\, j =1, \dots, {J+\tilde{J}}, \,\, t\ge0. \label{eqn:Zjx}
\end{align}

To facilitate the analysis, we define
\begin{align}
     \kappa &= \sum_{j=1}^J (r_j + \gamma_j ) x_j + \sum_{j=J+1}^{J+\tilde{J}} \gamma_j x_j, \label{eqn:handk} \\
     \eta_l &= \sum_{j \in \mathcal{R}_l} \frac{\hat{k}_j}{\mu_j} \hat{r}_{jl} + \sum_{j \in \mathcal{S}_l} \frac{\hat{\lambda}_{jl}}{\mu_j},
        \,\, l = 1, \dots, L. \label{eqn:preserve-eta}
\end{align}
Then, substituting (\ref{eqn:BCP-Zj1})-(\ref{eqn:BCP-Zj1.2}) and (\ref{eqn:Zjx}) into (\ref{eqn:W}), and using (\ref{eqn:BCP-Zj4}), (\ref{eqn:BCP-L1}), and (\ref{eqn:handk}), we arrive at the following equivalent workload formulation: Choose the non-anticipating control $\delta(\cdot) = (\delta_{l}(\cdot))$ to
\begin{align}
&\mbox{Minimize } \overline{\lim}_{t \rightarrow \infty} \frac{1}{t} \mathbb{E} \left[
     \int_0^t \sum_{l=1}^L F_l \delta_l(s)  ds + p U(t) \right]
     \label{eqn:workload-obj} \\
&\,\,\, \mbox{subject to} \nonumber \\
&\,\,\, W(t) = X(t) + \theta_0 t + \sum_{l=1}^L 
     \int_0^t \eta_l \delta_{l}(s) ds - \int_0^t \kappa W(s) ds + U(t), \,\, t \ge0, \label{eqn:workload-W1} \\
&\,\,\, W(t) \ge 0, \,\, t \ge0, \label{eqn:workload-W2} \\
&\,\,\, \delta_{l}(t) \in [0,1] \,\, \mbox{ for } \,\, l=1, \dots, L \,\, \mbox{ and }
     \,\, t \ge 0, \label{eqn:workload-W3} \\
&\,\,\, U \mbox{ is nondecreasing with } U(0)=0, \label{eqn:workload-L1} \\
&\,\,\, \int_0^t 1_{\{W(t) > 0 \}} dU(t)=0, \label{eqn:workload-L2} 
\end{align}
where $X(t)$ is a $(0, \sigma_w^2)$ Brownian motion with $\sigma_w^2 = \sum_{j=1}^{J} 2 r_j \hat{k}_j/\mu_j^2 + \sum_{j=J+1}^{J+\tilde{J}} 2 \lambda_j / \mu_j^2$.

The following proposition establishes the equivalence of the Brownian control problem and the equivalent workload problem. It is proved in the Online Supplement Section~B.
\begin{proposition} \label{prop:brownianworkload}
The formulation (\ref{eqn:workload-obj})-(\ref{eqn:workload-L2}) is equivalent to the BCP~(\ref{eqn:BCPobj}) in the following sense: A feasible policy $\delta(\cdot)$ of the BCP~(\ref{eqn:BCPobj}) is also feasible for the workload formulation~(\ref{eqn:workload-obj})-(\ref{eqn:workload-L2}) with the same cost. Similarly, a feasible policy $\delta(\cdot)$ of the formulation~(\ref{eqn:workload-obj})-(\ref{eqn:workload-L2}) is also feasible for the BCP~(\ref{eqn:BCPobj}) and it has the same cost in both formulations. 
\end{proposition}

To further simplify the analysis, we define
\begin{align}
c_l = \frac{F_l }{\eta_l},
     \,\, l = 1, \dots, L. \label{eqn:cjmlm}
\end{align}
By relabeling the engagement activities if needed, we assume that $c_1 < c_2 < \dots < c_L < p$.

To facilitate the analysis to follow, let
\begin{align}
\theta_l = \theta_0 + \sum_{i=1}^l \eta_i, \,\, l = 1, ..., L. \label{eqn:thetam}
\end{align}
Then let $A = [\theta_0, \theta_L]$ and define the piecewise linear, convex increasing cost function $c(\cdot): A \rightarrow \mathbb{R}_+$ as follows: 
\begin{align}
c(x) = \left\{ \begin{array}{ll}
		c_1 (x - \theta_{0}), & \mbox{if }\,\,\theta_{0} \le x \le \theta_1 \\
		\sum_{i=1}^{l-1} c_i \eta_i + c_l (x - \theta_{l-1}), & \mbox{if }\,\, \theta_{l-1} < x \le \theta_l, \,\, l = 2, ..., L.
     \end{array} \right.  \label{eqn:cx}
\end{align}

Figure~\ref{fig:cx} displays an illustrative $c(\cdot)$ function with $L=4$.

\begin{figure}[htbp]
\centering
\includegraphics[scale=0.65]{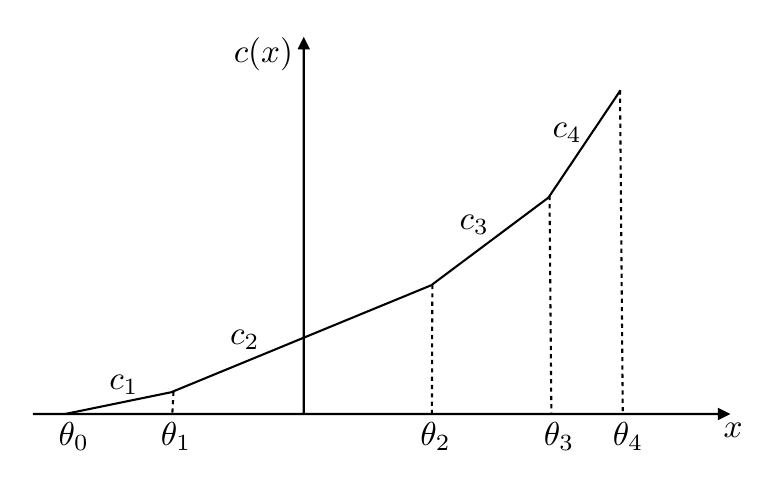}
\caption{An illustrative $c(\cdot)$ function with $L=4$.}
\label{fig:cx}
\end{figure}

Then we consider the following drift-rate control problem: Choose $\theta(\cdot) : [0, \infty) \rightarrow A$ so as to
\begin{align}
&\mbox{Minimize } \overline{\lim}_{t \rightarrow \infty} \frac{1}{t} \mathbb{E} \left[  
   \int_0^t c(\theta(s)) ds + p U(t) \right] \label{eqn:driftrate-obj} \\
& \,\,\, \mbox{subject to} \nonumber \\
& \,\,\, W(t) = X(t) + \int_0^t \theta(s) ds - \int_0^t \kappa W(s) ds + U(t), \,\, t \ge 0, \label{eqn:driftrate-W1} \\
& \,\,\, W(t) \ge 0, \,\, t \ge0, \label{eqn:driftrate-W2} \\
& \,\,\, U(t) \mbox{ is nondecreasing with } U(0) = 0,  \label{eqn:driftrate-L1} \\
& \,\,\, \int_0^t 1_{ \{ W(s) > 0 \}} dU(s) = 0. \label{eqn:driftrate-L2}
\end{align}
Proving the existence and uniqueness of the workload process rigorously involves studying an underlying Skorohod map. Such analysis is undertaken in the literature for similar processes, see for example \citet{reed-etal-2013}. 

The next proposition establishes the equivalence of the drift rate control problem~(\ref{eqn:driftrate-obj})-(\ref{eqn:driftrate-L2}) and the workload formulation~(\ref{eqn:workload-obj})-(\ref{eqn:workload-L2}); see the Online Supplement Section~B for its proof.

\begin{proposition} \label{prop:driftrate} 
The drift-rate control problem~(\ref{eqn:driftrate-obj})-(\ref{eqn:driftrate-L2}) and the workload formulation~(\ref{eqn:workload-obj})-(\ref{eqn:workload-L2}) are equivalent in the following sense: Given a feasible policy $\delta(\cdot)$ for the workload formulation, we set 
$\theta(t) = \theta_0 + \sum_{l=1}^{L} \sum_{j \in \mathcal{R}_l} (\hat{k}_j  \hat{r}_{jl} / \mu_{j}) \delta_{l} (t) + \sum_{l=1}^{L} \sum_{j \in \mathcal{S}_l}  (\hat{\lambda}_{jl} / \mu_{j}) \delta_{l} (t)$ 
for $t \ge 0$. This constitutes a feasible policy for the drift-rate control problem and its cost is less than or equal to that of $\delta(\cdot)$ in the workload formulation. Similarly, given a feasible policy $\theta(\cdot)$ for the drift-rate control problem, we define $\delta(\cdot) = (\delta_{l}(\cdot))$ as follows. If $\theta(t) = \theta_0$, then we set $\delta_{l}(t) = 0$ for all $l$. Otherwise, we have $\theta(t) \in (\theta_{m-1}, \theta_m]$ for $m \in \{ 1, \dots, M\}$, and we set 
\begin{align}
\delta_{l}(t) = \left\{
     \begin{array}{ll}
          1, & \mbox{if } l \le m-1, \\
          \frac{\theta(s) - \theta_{m-1}}{\eta_m}, & \mbox{if } l = m, \\
          0, & \mbox{if } l \ge m+1,
     \end{array} \right. \label{eqn:driftrate}
\end{align}
for $l=1, \dots, L$. This policy is feasible for the workload problem and has the same cost as policy $\theta(\cdot)$ in the drift-rate control problem.
\end{proposition}
In light of Proposition~\ref{prop:driftrate}, we will refer to both formulations as the workload problem below.

\section{Solution to the Equivalent Workload Formulation} \label{sec:solutionworkload}

This section solves the equivalent workload problem. In doing so, attention is restricted to the stationary Markov policies to minimize technical complexity. As a preliminary to our analysis, we next consider the Bellman equation for the formulation~(\ref{eqn:driftrate-obj})-(\ref{eqn:driftrate-L2}): Find a twice-continuously differentiable function $f$ and the constant $\beta>0$ that jointly satisfy the following differential equation:
\begin{align}
& \beta = \min_{x \in A} \left\{ \frac{1}{2} \sigma^2 f''(w) + (x - \kappa w) f'(w) + c(x)  \right\}, \label{eqn:f1} \\
& \mbox{such that } f'(0) = -p \mbox{ and } f' \mbox{ is increasing with } \lim_{w \rightarrow \infty} f'(w) = 0. \label{eqn:f2}
\end{align}
In dynamic programming, the unknown function $f$ is often interpreted as the relative value function and $\beta$ is a guess at the optimal average cost.  Presumably our analysis can be extended to include more general history-dependent and non-stationary policies, but that avenue will not be explored here, see Section 3.3.4 of \citet{ata-2003} for a sketch of such an extension for a related setting. 

Next, for technical convenience, we simplify the Bellman equation~(\ref{eqn:f1})-(\ref{eqn:f2}) in three steps. First, we define \vspace{-10pt}
\begin{align}
g(w) = f(w) + pw, \,\, w \ge 0, \label{eqn:gw}
\end{align} 
and express the Bellman equation in terms of this function by substituting~(\ref{eqn:gw}) into (\ref{eqn:f1})-(\ref{eqn:f2}): Find a twice-continuously differentiable, convex increasing function $g$ and the constant $\beta > 0$ that jointly satisfy the following:
\begin{align}
& \beta = \min_{x \in A} \left\{
   \frac{1}{2} \sigma^2 g''(w) + (x - \kappa w ) [ g'(w) -p] + c(x) \right\} \label{eqn:g1} \\
& \mbox{such that } g'(0)=0, \,\, g' \mbox{ is increasing with } \lim_{w \rightarrow \infty}
   g'(w) = p. \label{eqn:g2} 
\end{align}

Second, we note that~(\ref{eqn:g1})-(\ref{eqn:g2}) is really a first-order differential equation because it does not include the unknown function $g$ itself. Thus, we let $v(w) = g'(w)$ and rewrite the Bellman equation as follows: Find a nondecreasing, continuously differentiable function $v$ and a constant $\beta >0$ that jointly satisfy the following:
\begin{align}
& \beta = \min_{x \in A} \left\{
   \frac{1}{2} \sigma^2 v'(w) + (x - \kappa w ) [ v(w) -p] + c(x) \right\} \label{eqn:v1} \\
& \mbox{such that } v(0)=0, \,\, v \mbox{ is increasing with } \lim_{w \rightarrow \infty}
   v(w) =  p. \label{eqn:v2} 
\end{align}

Lastly, in order to further streamline the analysis to follow, we define the convex conjugate of the cost function $c(\cdot)$ as follows:
\begin{align}
\phi(y) = \sup_{x \in A} \{ yx - c(x) \}, \,\, y \in \mathbb{R}. \label{eqn:phi}
\end{align}
Although there may be multiple maximizers, there exists a minimal one (see, for example, \citealt{ata-harrison-shepp-2005}) which we denote by $\psi(y)$, where
\begin{align}
\psi(y) = \inf \arg \max_{x \in A} \{ yx - c(x) \}, \,\, y \in \mathbb{R}. \label{eqn:psi}
\end{align}
The function $\psi$ efficiently captures all relevant aspects of the cost function $c(\cdot)$. It is straightforward to show that $\psi$ is a left-continuous step function, whereas $\phi$ is a piecewise linear, convex function (see the Online Supplement Section~E for further details). 

Substituting~(\ref{eqn:phi}) in (\ref{eqn:v1}) and rearranging the terms gives rise to the final form of the Bellman equation: Find a continuously differentiable function $v$ and a constant $\beta > 0$ that jointly satisfy the following:
\begin{align}
& \beta = \frac{1}{2} \sigma^2 v'(x) + x \kappa(p - v(x)) - \phi(p - v(x)), \,\, x \ge 0, \label{eqn:bellfinal1} \\
& \mbox{such that } v(0) = 0 \mbox{ and } v \mbox{ is increasing with}
   \lim_{x \rightarrow \infty} v(x) = p. \label{eqn:bellfinal2}
\end{align}

The following result establishes the existence of the solution to the Bellman equation. Its proof involves several steps, involving various technical subtleties. It is provided in the Online Supplement Section~F.
\begin{theorem} \label{thm:bellman}
The Bellman equation~(\ref{eqn:bellfinal1})-(\ref{eqn:bellfinal2}) has a solution $(\beta^*, v)$  where $\beta^* >0$ and $v$ is increasing and continuously differentiable.
\end{theorem}
Now, the function $\psi$ naturally suggests a candidate policy. Thus, our proposed policy for the equivalent workload formulation sets
\begin{align}
\theta^*(z) =\psi(p - v(z)), \,\,\, z \ge 0. \label{eqn:thetascand}
\end{align}
In addition, in order to provide a solution of the Bellman equation (\ref{eqn:f1})-(\ref{eqn:f2}), we define
\begin{align*}
f(z) = \int_0^z v(s) ds - pz, \,\, z \ge 0.
\end{align*}
\begin{corollary} \label{cor:bell}
The pair $(\beta^*, f)$ solve the Bellman equation (\ref{eqn:f1})-(\ref{eqn:f2}).
\end{corollary}

The next subsection characterizes the proposed policy explicitly.

\subsection{An Explicit Characterization of the Proposed Policy} \label{sub:threshtau}

This subsection characterizes the optimal policy as a nested threshold policy. It also characterizes the corresponding thresholds explicitly. Recall from Theorem \ref{thm:bellman} that the function $v$ is increasing. Therefore, it is invertible. We denote that inverse by $v^{-1}$ and define the thresholds:
\begin{align}
\tau_l = v^{-1}(p - \hat{c}_l), \,\, l = 1, ..., L. \label{eqn:taum}
\end{align}
We let $\hat{c}_{L+1} = p$ and $\tau_{L+1} = 0$ for notational convenience. It follows from the monotonicity of $v$ that $\tau_1 > \tau_2 > ... > \tau_L > \tau_{L+1} = 0$.
Then follows from the definition of $\psi$ (see Equation~(\ref{eqn:psi}), also see the Online Supplement Section~E) that the candidate policy is the nested threshold policy: 
\begin{align}
\theta^*(z) = \left\{
   \begin{array}{ll}
      \theta_L, & \mbox{if } z < \tau_L^*, \\
      \theta_l, & \mbox{if } \tau_{l+1}^* \le z < \tau_l^*, \,\, l=1, ..., L-1, \\
      \theta_0, & \mbox{if } \tau_1^* < z.
   \end{array}
\right. \label{eqn:thetas}
\end{align}

To facilitate the proof of the optimality of the candidate policy, we next introduce the class of admissible policies. Our definition amounts to ensuring that the system is stable under an admissible policy. To this end, we let
\begin{align*}
\bar{\theta} = \max \{ \theta_l : \theta_l <0, l = 0,1, \dots, L \}.
\end{align*}

\paragraph{Definition (Admissible Policies).} A policy $\theta(\cdot)$ is admissible if it is stationary Markov and if there exists a threshold $\bar{z}$ such that $\theta(z) \le \bar{\theta}  < 0$ for $z \ge \bar{z}$.

To repeat, this condition ensures that under an admissible policy the system is stable. This is a natural requirement because an unstable system results in infinite cost, making its policy not worthy of consideration. The next result establishes that the candidate policy given in~(\ref{eqn:thetas}) is indeed optimal; see the Online Supplement Section~C for its proof. 
\begin{theorem} \label{thm:policy}
The candidate policy $\theta^*(\cdot)$ given in~(\ref{eqn:thetas}) is optimal for the workload problem, and its long-run average cost is $\beta^*$.
\end{theorem}

To translate the optimal policy of the equivalent workload formulation to the original formulation, we first define the (unscaled) workload process:
\begin{align}
\mathcal{W}^n(t) = \sum_{j=1}^{J+\tilde{J}} m_j Q_j^n(t), \,\, t \ge0. \label{eqn:Wnt}
\end{align}
We also define the workload thresholds $w_1^* > \cdots > w_L^* >0$ as follows:
\begin{align}
w_l^* = \sqrt{n} \tau_l, \,\, l = 1, \dots , L. \label{eqn:wms}
\end{align}
The proposed policy engages in activities $1, \dots, l$ whenever $\theta_l$ is chosen in the workload problem. In other words, the natural interpretation of the optimal policy (given in~(\ref{eqn:thetas})) for the original problem is that the non-profit should engage in these activities whenever 
\begin{align}
\mathcal{W}^n(t) < w_l^* \mbox{ for } l=1, \dots, L, \label{eqn:Wmt}
\end{align}
and should do nothing otherwise, i.e., $\mathcal{W}^n(t) \ge w_1^*$ (see Figure~\ref{fig:proposedpolicyorig2}).
\begin{figure}[htbp]
\centering
\includegraphics[scale=0.50]{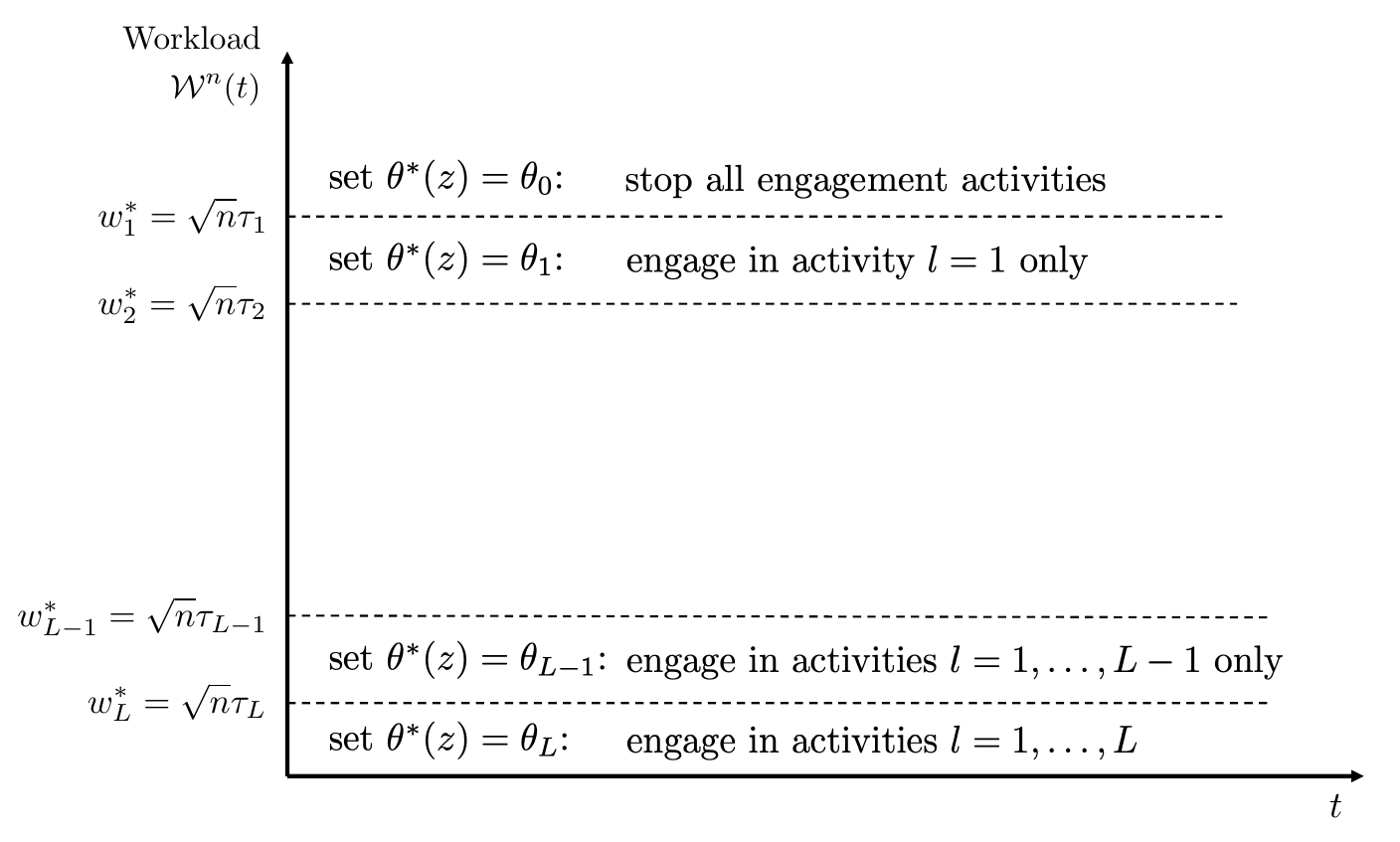}
\caption{The proposed policy for the original system and its relationship to the proposed policy of the scaled workload problem.}
\label{fig:proposedpolicyorig2}
\end{figure}

An alternative way to express the proposed policy is to specify when each volunteer engagement activity will be used. That is, we set
\begin{align}
\delta_{l}(t) = \left\{
   \begin{array}{ll}
      1, & \mbox{if } \mathcal{W}^n(t) < w_{l}^*, \\
      0, & \mbox{otherwise}.
   \end{array}
\right. \label{eqn:deltajl1}
\end{align}

Next, we specify the server's prioritization policy. To ensure Equation~(\ref{eqn:Qj}) is satisfied, for each class $j$, we calculate the penalty function
\begin{align*}
p_j(t) = Q_j^n(t) - \frac{x_j}{m_j} \mathcal{W}^n(t), \,\, t \ge 0.
\end{align*}
Then, the server works on the volunteer class $j$ for which $p_j(t)$ is highest. In case of a tie, the class with the lowest index (among tied ones) is prioritized. As an aside, the non-profit can also ensure Equation~(\ref{eqn:Qj}) is satisfied using a discrete-review policy, which makes resource allocation decisions at predetermined review times. See \citet{Harrison_96, Harrison_98}, \citet{Maglaras_1999, Maglaras_2000}, \citet{ata-kumar-2005}, and \citet{Ata_Olsen_2009, Ata_Olsen_Queueing_2013} for such policies and their analysis under fluid and diffusion approximations.

Lastly, we characterize thresholds $\tau_1, \ldots, \tau_L$ explicitly. To do so, we next provide a closed-form formula for $v(\cdot)$. And,	inverting it yields the thresholds, cf. Equation (\ref{eqn:taum}). 

For notational brevity, we define $u_l(\cdot)$ for all $x$ and for $l=1,\ldots, L+1$ as follows:
\begin{align}
 u_l(x) &=  p  -  \hat{c}_l  \exp \left\{ \frac{- 2 \theta_{l-1} (x-\tau_l) + \kappa (x^2-\tau_l^2)}{\sigma^2} \right\}  \nonumber \\ 
& +  \frac{2(\beta - c(\theta_{l-1}))\sqrt{\pi}}{\sigma \sqrt{\kappa}}   
\exp \left\{ \frac{(\kappa x - \theta_{l-1})^2}{\kappa\sigma^2} \right\} \left[ \Phi \left( \frac{\kappa x - \theta_{l-1}}{\sigma \sqrt{\kappa/2}} \right) - \Phi \left( \frac{\kappa \tau_l- \theta_{l-1}}{\sigma \sqrt{\kappa/2}} \right) \right], \nonumber 
\end{align}
where $\Phi$ denotes the cumulative distribution function for a standard normal random variable. The following proposition characterizes the thresholds $\tau_1, \dots, \tau_L$; see the Online Supplement Section~D for its proof.
%
\begin{proposition} \label{prop:valuefunc}
For $x \in [\tau_l, \tau_{l-1})$, $l=1, \ldots, L+1$, we have that 
$v(x) = u_{l}(x)$ and $\tau_l = v^{-1}(p-\hat{c}_l)$.
\end{proposition} 
The Online Supplement Section~F shows the solution to the Bellman equation (\ref{eqn:bellfinal1})-(\ref{eqn:bellfinal2}).

\section{Numerical/Simulation Study} \label{sec:numerical}

To investigate the effectiveness of our dynamic policy, we conducted a numerical study. This study was based primarily on the volunteer operations of a food bank in the Southwestern United States which we will call ``Food Bank~A''. Additionally, we also contacted a sample of 155 food banks in the U.S. using the GuideStar database on nonprofit organizations and searched on “Food Banks.” Out of these, 66 food banks agreed to answer a Qualtrics survey about how they engage volunteers in their organization. We matched the responses from the survey with information from the Form 990 filed by the organization. The respondents included food banks from 36 states. Gross receipts, including financial donations, in-kind donations, grants, and fees for services, ranged from \$54K to \$105M, with a median of \$3.4M. Information collected in the survey included the types of activities volunteers perform, the criticalness of these activities to food bank performance, top cost categories associated with volunteer labor, total volunteer hours per year, full-time equivalent (FTE) employees, probability that a volunteer also donates financially, and, if so, average volunteer financial donation amount.
Most germane to this study, we note that 24\% of the respondents listed volunteer recruiting as one of their top three volunteer-related costs and 55\% mentioned managing and training volunteers. As organizations that provide continuous services, the following were mentioned as ongoing volunteer activities: sorting food donations (80\%), packing and bagging food for clients (61\%), distribution and delivery (48\%), stocking shelves (29\%), and clerical tasks (21\%).

Our simulation process logic and the parameters used to calculate the dynamic policy thresholds are based on the operations of Food Bank~A. Food Bank~A runs a large, volunteer-staffed meal preparation operation that delivers food to a service network of nearly 1{,}400 organizations, schools, after-school, and feeding sites. We collected operational information from Food Bank A through an on-site visit, three interviews (the first interview was conducted with three managers, the Director of Volunteer and Agency Services, Business Support Manager, and Volunteer Engagement Manager; the second and third interviews were conducted with the Volunteer Engagement Manager), and follow-up email correspondences. Using the simulation, we compare the performance of our dynamic policy to all possible static policies. In particular, we assess and compare the total annual cost, abandonment percentage, activity cost, and lost throughput resulting from the respective policies.

\subsection{Simulation Process Logic}

Our simulation process logic closely follows the actual process at Food Bank~A (i.e., without making the simplifying assumptions used to derive the dynamic policy). In particular, there are two simplifying assumptions made in the analytical model (Section~\ref{sec:model}) that we do not incorporate into the simulation. 
First, although the model is in continuous time, we build a discrete time simulation where each period represents one day. Moreover, volunteers that work the same shift are synchronized in their arrival to the food bank. 
Second, in the simulation, we set distinct times for possible engagement activities based on an approximate schedule from Food Bank~A. If the food bank chooses to execute the activity, it will result in more volunteer signups from the affected class(es) during the impact period of the engagement activity. Also, instead of continuously adjusting engagement activities, we will check system congestion before each engagement activity to determine whether to conduct that activity.

Each period of the simulation is equivalent to one day. The food bank runs the volunteer operation 6 days per week. The following describes the simulation process logic from the volunteer's perspective:
\begin{enumerate}

\item A volunteer arrives to sign up for a volunteer shift. The volunteer enters a queue for a volunteer shift (i.e., waits to be processed by the food bank, where ``processed'' means to have the volunteer work a shift). The arrival rate of each class of volunteer depends on the inherent characteristics of the volunteer class and the engagement activities performed by the food bank.

\item The volunteer advances in the first-come-first-served queue until the day they can work a volunteer shift. Any time before the volunteer works the shift, the volunteer can abandon the queue.

\item On the day of the volunteer shift, the volunteer works simultaneously with all other volunteers who are working the shift on that day.

\item After the day of the volunteer shift, the volunteer enters repose if they are in a class of repeat volunteers, or exits the system entirely if they are in a class of one-time volunteers.

\end{enumerate}
The following describes the simulation process logic from the food bank's perspective:
\begin{enumerate}

\item  For every working day of the volunteer operation, the food bank processes the first 250 volunteers in the queue. If there are fewer than 250 volunteers in the queue, some volunteer slots for that day are left unfilled and the food bank loses the corresponding amount of output.

\item The food bank has engagement opportunities on specific days of the year. Each type of engagement opportunity occurs on a periodic basis, e.g., once per day, once per week, etc. Using a static policy, the food bank performs the engagement activity according to the schedule. Using the dynamic policy, before each engagement activity opportunity, the food bank checks the number of volunteers in the queue, and determines whether to perform the engagement activity based on the policy thresholds. If the food bank chooses to perform the engagement activity, the arrival rates of the affected volunteer classes are boosted until the next opportunity for the same engagement activity. If the food bank does not perform the engagement activity, the arrival rates of the volunteer classes remain at the base arrival rate.

\end{enumerate}

\subsection{Simulation Parameters}

Most of the parameters of the simulation were calibrated using operational information from Food Bank~A. However, the abandonment rate, $\gamma_j$, was estimated using qualitative descriptions from the managers during our interviews and therefore, we perform sensitivity analysis on this parameter. 
Using these parameters, we derived the thresholds of the dynamic policy using the approximating Brownian control problem described in Section~\ref{sec:approx}.

The following simulation parameters are based on the meal preparation operation of Food Bank~A. (See the Online Supplement Section~H for detailed descriptions and calculations of parameter values).

\paragraph{Volunteers.} There are $3$ classes of volunteers: corporate volunteers (class $j=1$), individuals (class $j=2$), and social groups such as school groups or soccer clubs (class $j=3$).  Based on information from Food Bank~A, corporate volunteers and individuals are more likely to be repeat volunteers, whereas social groups tend to come and go. Therefore, we consider volunteers in classes $j=1,2$ to be repeat volunteers and volunteers in class $j=3$ to be one-time volunteers, resulting in $J=2$ and $\tilde{J} = 1$. Estimates from Food Bank~A indicate that there are approximately $k_1^n = 11{,}200$ corporate volunteers and $k_2^n = 16{,}800$ individual volunteers.

There are two types of corporate volunteers ($j=1$): regulars who volunteer on average once every 8 months and sporadic volunteers who volunteer on average once every 4 years. Of the volunteer slots filled by corporate volunteers, 70\% are filled by regulars and 30\% are filled by the sporadic type. Using these approximate volunteering frequencies from Food Bank~A, we derive the resulting arrival rate of class $j=1$ volunteers to be $r_1^n k_1^n = 6720$ arrivals per year. 
Individual volunteers ($j=2$) consist of 4 types of volunteers. In decreasing order of volunteering frequency: 0.5\% volunteer on average once per week, 17.5\% volunteer on average once per month, 32\% volunteer on average twice per year, and 50\% on average volunteer once every 4 years. Thus, we derive the arrival rate of class $j=2$ volunteers to be $r_2^n k_2^n = 52{,}500$ arrivals per year. 
Food Bank~A estimates that approximately 25\% of volunteer arrivals are social group volunteers ($j=3$), therefore, we estimate that the arrival rate of class $j=3$ is $\lambda_3^n = 19{,}540$ arrivals per year.

Combining the arrival rates of all three classes, the total annual arrival rate is $ r_1^n k_1^n + r_2^n k_2^n + \lambda_3^n = 78{,}760$ arrivals per year. See the Online Supplement Section~H.1 for detailed descriptions and calculations related to volunteer arrival rates.

\paragraph{Service rate.} 

The food bank is able to accommodate 250 volunteers per day, i.e., it strives to fill 250 volunteer slots per day. Meal preparation operates 6 days per week for 52 weeks per year. Any class of volunteer can fill a volunteer slot. Therefore, $\mu_1^n = \mu_2^n = \mu_3^n =  78{,}000$ slots per year.

\paragraph{Engagement activities.} 

The food bank can engage in $L=4$ different volunteer engagement activities. Each engagement activity follows a periodic schedule. In the simulation, at each point when an engagement activity is scheduled, the activation of the activity will depend on the congestion of the system and the thresholds of the dynamic policy. If the congestion is low enough (i.e., not enough volunteers), the engagement activity will be activated and during the period of time until the next scheduled occurrence of the activity, the arrivals of the volunteers increase accordingly. Otherwise, the food bank will not activate the engagement activity and the arrival rates of volunteers will remain unchanged. Following are brief descriptions of the engagement activities.

\begin{itemize}
   \item \textbf{Orientation during volunteer shift ($l=1$)}: A brief orientation at the beginning and closing remarks at the end of each volunteer shift explain the social impact that the volunteer work has on the population in need. This activity affects the arrival rate of volunteer classes $\mathcal{R}_{1} \cup \mathcal{S}_{1} = \{ 1,2,3\}$. On average, Food Bank~A estimates that the orientation activity increases volunteer signups by 10-15 volunteers per week. The typical frequency of this activity is once per volunteer shift (i.e., 6 times per week), thus, we calculate that this activity increases volunteer arrivals by 2 volunteers per engagement activity. This activity takes approximately 10 minutes of employee time.
   
   \item \textbf{Electronic communication ($l=2$)}: The food bank sends targeted emails and other forms of electronic communication to volunteers to notify them of volunteer activities. This activity affects the arrival rate of volunteer classes $\mathcal{R}_{2} \cup \mathcal{S}_{2} = \{ 1,2,3\}$. On average, Food Bank~A estimates that electronic communication increases volunteer signups by 15 volunteers per week. The typical frequency of this activity is once per week, thus, we calculate that this activity increases volunteer arrivals by 15 volunteers per engagement activity. It takes an employee approximately one hour to formulate an e-communication (e.g., newsletter).

   \item \textbf{Speaking engagement ($l=3$)}: Food bank staff can make presentations at organizations in the region to raise awareness. This activity affects the arrival rate of volunteer classes $\mathcal{R}_{3} = \{ 3\}$. On average, Food Bank~A estimates that speaking engagements increase volunteer signups by 20 volunteers per month. The typical frequency of this activity is once per month, thus, we calculate that this activity increases volunteer arrivals by 20 volunteers per engagement activity. Food bank employees must travel to and spend time speaking with organizations in the region.

   \item \textbf{Tabling at a fair ($l=4$)}: 
Food bank staff can set up an information table at fairs throughout the year. This activity affects the arrival rate of volunteer classes $\mathcal{R}_{4} = \{ 1,3\}$. On average, Food Bank~A estimates that tabling at fairs increases volunteer signups by 30 volunteers per month. The typical frequency of this activity is once per month, thus, we calculate that this activity increases volunteer arrivals by 30 volunteers per engagement activity. Food bank employees must travel to and spend time staffing the table at the fair.

\end{itemize}

Parameters related to engagement activities are shown in Table~\ref{tab:engage-params}. See the Online Supplement Section~H.2 for detailed descriptions of and calculations for these engagement activities. 

\begin{table} [tbp]
\centering
\begin{tabular}{ccccc}
& & & Increase in & Fixed cost \\
& Engagement   & Volunteer  &  annual arrivals & of activity, $F_l^n$ \\  
& activity  &  classes affected &  [arrivals/year]    &  [\$/year] \\ \hline

$l=1$ & Orientation & $j=1,2,3$    &   624 & 936  \\ \hline

$l=2$ & E-communication & $j=1,2,3$    &  780 &  1820   \\ \hline
$l=3$ & Speaking & $j=3$   &  240 &  720   \\ \hline
$l=4$ & Tabling & $j=1,3$  &  360 &  1800   \\ \hline
\end{tabular}
\caption{Engagement activity parameters.}
\label{tab:engage-params}
\end{table}

\paragraph{Other parameters.}

We assume that the cost rate of idleness (throughput loss) in the system is equal to the opportunity cost of forgone social impact, i.e., it is equivalent to the value of the meals produced by volunteers per year. Based on Food Bank~A's cost of the meal, we calculate that $p = \$50$ (see the Online Supplement Section~H.3 for details).
Discussions with Food Bank A revealed that there are always some volunteers that abandon the queue but the number is typically low. We correspondingly set the cancellation/abandonment rate to be low: $\gamma_j = 0.01$, $j=1,2,3$. However, we perform sensitivity analysis on $\gamma_j$ to determine the effect of varying $\gamma_j$. 
Lastly, we assume a first-come-first-served service discipline, therefore, $x_j =  k_j r_j / \mu_j$ for $j = 1, 2$ and $x_j = \lambda_j / \mu_j$ for $j=3$.
Hence, $x_1 = 0.085, x_2 = 0.067, x_3 = 0.248$. See the Online Supplement Section~H.3 for detailed descriptions and calculations related to the parameters described above. 

\subsection{Optimal policy thresholds} \label{sub:numericalthresh}

To calculate the threshold values of the optimal policy, we first scale the parameters of the $n^{th}$ system to parameters for the limit system. We assume $n$ is large and roughly equal to the number of potential volunteers Food Bank~A can reach with its engagement activities, thus we set $n=56{,}000$. The calculations for parameter values $\lambda_j$, $\mu_j$, $\hat{r}_{l}$, $\hat{\lambda}_l$, $F_l$, $p$, $\kappa$, and $\sigma_w^2$ are given in the Online Supplement Section~H.4. 

Using Equation~(\ref{eqn:cjmlm}), we calculate $\hat{c}_l$ for $l=1, \dots, 4$ and note that $\hat{c}_1 < \dots < \hat{c}_4$. We use Equation~(\ref{eqn:theta0}) to find $\theta_0$, and then use Equations~(\ref{eqn:thetam}) and (\ref{eqn:cx}) to calculate $\theta_l$ and $c(\theta_l)$, respectively. Since the dynamic policy is derived using an approximate system, we further tune the policy thresholds by adjusting $\theta_0$ (which results in adjustments to the thresholds). Because the dynamic policy is nested, adjusting $\theta_0$ gives us a principled way to change the thresholds (i.e., preserving their relationship to each other). For example, in our base case simulation scenario, Equation~(\ref{eqn:theta0}) gives an initial value of $\theta_0 = -2.3$. We performed a search over a range of values around that value of $\theta_0$ and found that $\theta_0 = -1.4$ gave the thresholds with the best performance.

\begin{figure}[htbp]
\centering
\includegraphics[scale=0.55]{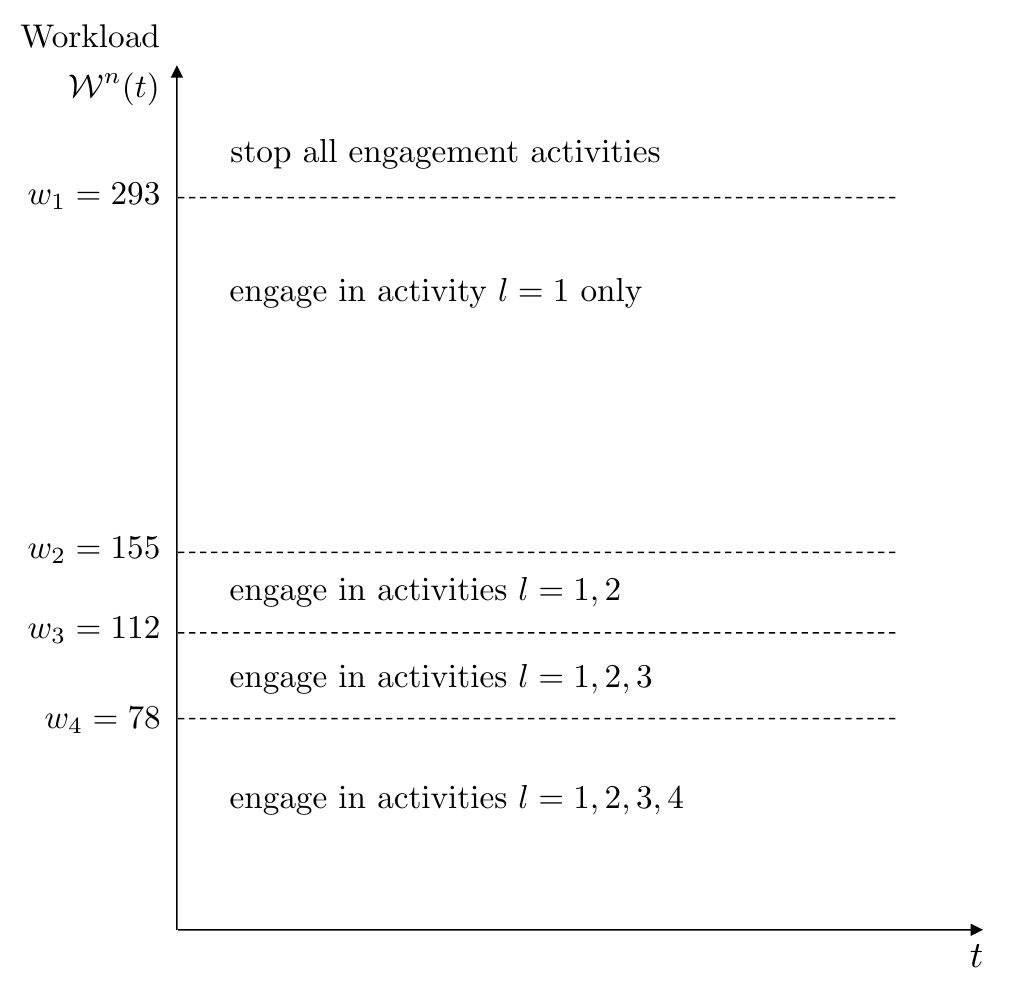}
\caption{The optimal dynamic policy. }
\label{fig:optpolicy}
\end{figure}

The Online Supplement Section~H.5 shows the equations used to determine optimal thresholds, $\tau_1, \dots, \tau_4$. Then, using Equation~(\ref{eqn:wms}), we determine the optimal workload thresholds for the optimal dynamic policy.
Figure~\ref{fig:optpolicy} depicts the optimal dynamic policy indicating when to deploy specific engagement activities.

\paragraph{Performance of optimal policy.} Each replication of the simulation runs for 25 years with a warm up period of 20 years. To calculate the annual performance, the last 5 years of each simulation run were averaged. Each scenario was simulated using 20 replications. We determined the best static policy among all possible static policies (i.e., all possible combinations of activities $l=1,2,3,4$ and no activities). We then compared the performance of our dynamic policy to the best static policy. 

\begin{figure}[htbp]
\centering
\includegraphics[scale=0.5]{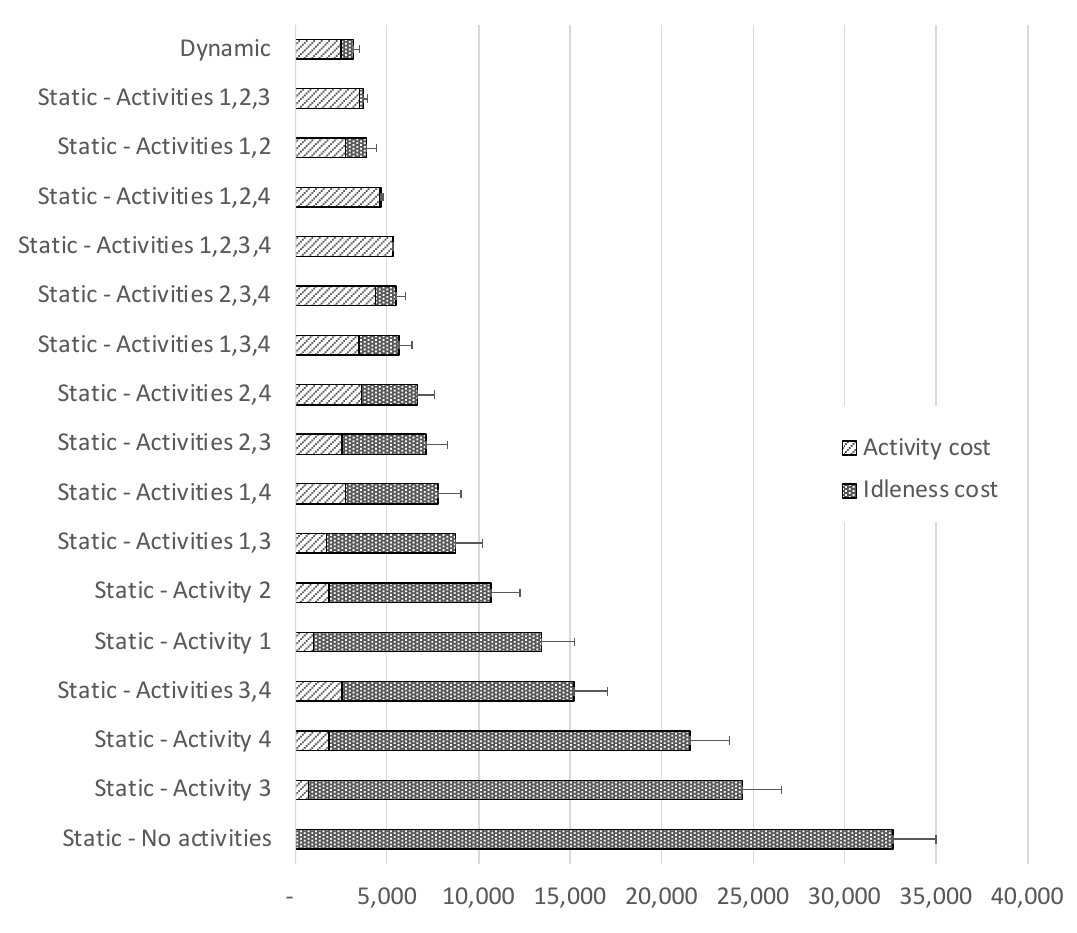}
\caption{Total annual cost of the dynamic policy compared to static policies.}
\label{fig:numresult1}
\end{figure}

The objective of the food bank as formulated in Section~\ref{sec:model} is to minimize total annual cost (see Equation~(\ref{eqn:C})).
The total annual cost includes the costs of activities deployed and the penalty cost of idleness or throughput loss (which captures the loss of social impact when there are not enough volunteers to work). 
Figure~\ref{fig:numresult1} shows the cost performance of the dynamic policy compared to the static policies. The best performing static policy uses activities $l=1,2, 3$ all the time, resulting in a total annual cost of $\$3728 \pm 200$ (95\% CI).  However, the total annual cost of the dynamic policy is $\$3123 \pm 383$ (95\% CI), an improvement (reduction in cost) of $\$605 \pm 407$ (95\% CI), which corresponds to an expected cost improvement of $16.2\%$. Under the dynamic policy, activities $1, 2,3$, and $4$ are used 99\%, 60\%, 31\%, and 14\% of the time, respectively.

The best static policy accumulated \$3476 and \$252 in activity and idleness costs, respectively, whereas the dynamic policy accumulated \$2494 and \$629 in activity and idleness costs. The dynamic policy is able to adjust to congestion levels and therefore only engages activities if the expected throughput loss exceeds the cost of the activity. Moreover, the dynamic policy anticipates expected abandonments and therefore reduces activities when congestion is high. This is reflected in the lower abandonment rate under the dynamic policy (1.07\%) compared to the best static policy (1.38\%).

We performed a sensitivity analysis on the parameter $\gamma_j$, as this parameter was estimated from qualitative descriptions from Food Bank~A. Table~\ref{tab:sensitivity} shows the results of the sensitivity analysis. Within a reasonable range of the parameter value, we found that the dynamic policy performed better than the best (lowest cost) static policy consistently. The dynamic policy did better relative to the best static policy when $\gamma_j$ is low. Intuitively, when abandonments are high (i.e., high $\gamma_j$), it is advantageous to deploy engagement activities more often to attract volunteers -- this is exactly what the best static policy does. Therefore, when $\gamma_j$ is high, the static policy performs well and the dynamic policy ends up behaving similar to the static policy. Therefore, the performance of the two policies are similar. When $\gamma_j$ is low, the flexibility of the dynamic policy does much better relative to the static policy because it can adjust the activity level according to the number of volunteers already signed up.

The Online Supplement Section~H.6 gives details of the simulation results. Also shown in the Online Supplement Section~H.6 are simulation results when all classes are repeat volunteers and all classes are one-time volunteers. Simulation robustness checks are given in the Online Supplement Section~H.7.

\begin{table} [tbp]
\centering
\begin{tabular}{lc}
$\gamma_j$   &  Expected \% Cost Improvement \\ \hline
0.005    &  $30.6\% $  \\ 
$0.010^*$&  $16.2\% $  \\ 
0.015    &  $ 8.6\% $  \\ 
0.020    &  $ 3.8\% $  \\ 
0.025    &  $ 0.3\% $ \\ 
\end{tabular}
\caption{Sensitivity analysis on parameter $\gamma_j$. $^*\mbox{Base case is } \gamma_j = 0.010$. }
\label{tab:sensitivity}
\end{table}

\section{Concluding remarks} \label{sec:conc}

As non-profit organizations step in to provide much needed services for vulnerable populations, the importance of their role in society grows. Much of the work done by non-profit organizations are performed by volunteers. In non-profit operations that provide ongoing services over the long term, ensuring a stable flow of volunteer labor is essential to serving their constituents. In this paper, we developed a model to capture how volunteer engagement activities affect the supply of volunteers to a non-profit organization. We incorporated the tradeoff between not having enough volunteers to do the work and encouraging too many volunteers to sign up, resulting in unnecessary engagement activity costs and increased volunteer abandonment. Using this model, we derived the optimal dynamic policy for deploying volunteer engagement activities. 
Our solution is a nested threshold policy with explicit congestion thresholds that indicate when the non-profit should deploy each type of engagement activity. We performed a numerical study using a simulation based on the volunteer operations of a large food bank. For a base case, the study showed that the dynamic policy significantly reduced the cost of the volunteer operation compared to the best static policy. Methodologically, this paper also contributes to the literature on drift-rate control problems for diffusion models.

We did not pursue a rigorous proof of asymptotic optimality of the proposed policies in the heavy traffic limit. However, we conjecture that the proposed policies will be asymptotically optimal in the diffusion scale. That is, we expect their optimality gap to be $o(\sqrt{n})$ as $n \rightarrow \infty$. Another possible future research question is to include the dynamic scheduling capability, thereby relaxing the requirement in Equation~(\ref{eqn:Qj}). In that case, the system manager decides which class to serve next dynamically so as to optimize the system performance.

\newpage

\newpage

\appendix

\setcounter{page}{1}
\begin{center}
{\Large \textbf{A Dynamic Model for Managing Volunteer Engagement\\Online Supplement: Proofs and Numerical Study Details} }\\[12pt]
\end{center}

\section{Brownian Control Problem} \label{app:BCP}

\setcounter{equation}{0}
\renewcommand{\theequation}{\thesection\arabic{equation}}

Note by the strong approximations \citep{csorgo-horvath-1993} that for a rate-one Poisson process, denoted by $N(t)$, we have that as $n \rightarrow \infty$, 
\begin{align}
N(nt) = nt + \sqrt{n} B(t) + o(\sqrt{n}), \label{app:Nnt}
\end{align}
where $B(\cdot)$ is a standard Brownian motion.
In what follows, we will use (\ref{app:Nnt}) repeatedly to derive the approximating Brownian control problem formally. To this end, recall that
\begin{align}
Z_j^n(t) &= \frac{Q_j^n(t)}{\sqrt{n}}, \,\,\, j = 1, ..., J, \,\, t \ge 0, \label{app:Zjn} \\
Y_j^n(t) &= \sqrt{n} \left( \frac{r_j \hat{k}_j}{\mu_j} t - T_j(t) \right),  \,\,\, j=1, \dots, J, \,\, t \ge 0, \label{app:Yjn} \\
Y_j^n(t) &= \sqrt{n} \left( \frac{\lambda_j }{\mu_j} t - T_j(t) \right),  \,\,\, j=J+1, \dots, J+\tilde{J}, 
     \,\, t \ge 0, \label{app:Yjn.2} \\
U^n(t) &= \sqrt{n} \, I(t), \,\,\, t \ge 0. \label{app:Ln}
\end{align}
Next, we consider the abandonment, service completion, and arrival processes $\{\Gamma_j^n(t), N_j^s(\mu_j^n T_j(t)), A_j^n(t) \}$ for class $j$ volunteers, $j=1, \dots, J+\tilde{J}$. Using Equation~(\ref{app:Nnt}), note that
\begin{align}
\Gamma_j^n(t) = N_j^b \left( \int_0^t \gamma_j Q_j^n(s) ds \right)
   = N_j^b \left( \sqrt{n} \int_0^t \gamma_j Z_j(s) ds \right)
   \approx \sqrt{n} \int_0^t \gamma_j Z_j(s) ds + o(\sqrt{n}). \label{app:Gammajn}
\end{align}
Similarly, for the service completion process, we first consider $j=1, \dots, J$ and write
\begin{align}
N_j^s (\mu_j^n T_j(t)) &= N_j^s( n \mu_j T_j(t)) \nonumber \\
&= N_j^s \left( n \mu_j \left( \frac{r_j \hat{k}_j}{\mu_j} t - \frac{Y_j^n(t)}{\sqrt{n}} \right) \right) \nonumber \\
&= n r_j \hat{k}_j t - \sqrt{n} \mu_j Y_j^n(t) + \sqrt{n r_j \hat{k}_j} \, B_j^s(t) + o(\sqrt{n}), \label{app:Njs}
\end{align}
where the second equality follows from Equation~(\ref{eqn:Yjn}) and the last equality follows from Equation~(\ref{app:Nnt}). Next, we consider $j=J+1, \dots, J+\tilde{J}$ and write
\begin{align}
N_j^s (\mu_j^n T_j(t)) &= N_j^s( n \mu_j T_j(t)) \nonumber \\
&= N_j^s \left( n \mu_j \left( \frac{\lambda_j}{\mu_j} t - \frac{Y_j^n(t)}{\sqrt{n}} \right) \right) \nonumber \\
&= n \lambda_j t - \sqrt{n} \mu_j Y_j^n(t) + \sqrt{n \lambda_j} \, B_j^s(t) + o(\sqrt{n}), \label{app:Njs.2}
\end{align}
where the second equality follows from Equation~(\ref{eqn:Yjn.2}) and the last equality follows from Equation~(\ref{app:Nnt}). Also, $B_j^s$ is a standard Brownian motion for $j=1, \dots, J+\tilde{J}$.

Lastly, we consider the arrival process $A_j^n$. In doing so, we first consider the repeat volunteers, i.e., $j=1, \dots, J$:
\begin{align}
A_j^n(t) &= N_j^a \left( \int_0^t R_j (\delta(s)) (k_j^n - Q_j^n(s))ds \right) \nonumber \\
&= N_j^a \left( \int_0^t \Big( r_j^n + \sum_{l \in \mathcal{L}_j} \frac{\hat{r}_{jl} \delta_{l}(s) }{\sqrt{n}} \Big)
   (k_j^n - Q_j^n(s)) ds \right) \nonumber \\
&\approx N_j^a \left( \int_0^t \Big( r_j  - \frac{\alpha_j}{\sqrt{n}} 
   + \sum_{l \in \mathcal{L}_j} \frac{\hat{r}_{jl}\delta_{l}(s) }{\sqrt{n}} \Big)
   ( n \hat{k}_j - \sqrt{n} Z_j(s)) ds \right) \nonumber \\
&= N_j^a \left( n \int_0^t \Big( r_j - \frac{\alpha_j}{\sqrt{n}} 
   +  \sum_{l \in \mathcal{L}_j} \frac{\hat{r}_{jl} \delta_{l}(s) }{\sqrt{n}} \Big)
   \left(\hat{k}_j - \frac{ Z_j(s)}{\sqrt{n}} \right) ds \right),     \label{app:Ajn}
\end{align}
where we used $r_j^n = r_j - \alpha_j/\sqrt{n}$, which holds by the heavy traffic assumption. The next two steps in the derivation of Equation~(\ref{app:Ajn})  follow from Equation~(\ref{eqn:Zjn}) and by rearranging the terms. 

To facilitate the approximation of $A_j^n(t)$ further, consider the following term in Equation~(\ref{app:Ajn}): 
\begin{align}
&n \int_0^t \left( r_j - \frac{\alpha_j}{\sqrt{n}} 
   +  \sum_{l \in \mathcal{L}_j} \frac{\delta_{l}(s) \hat{r}_{jl}}{\sqrt{n}} \right)
   \left(\hat{k}_j - \frac{ Z_j(s)}{\sqrt{n}} \right) ds \nonumber \\
&= n r_j \hat{k}_j t + \sqrt{n} \left[ \hat{k}_j \sum_{l \in \mathcal{L}_j} \Delta_l(t) \hat{r}_{jl} - \hat{k}_j \alpha_j t 
   - \int_0^t r_j  Z_j(s) ds \right] +O(1). \label{app:approxAjn}
\end{align}
Substituting this into Equation~(\ref{app:Ajn})  and using~(\ref{app:Nnt}) give
\begin{align}
A_j^n(t) &= n r_j \hat{k}_j t + \sqrt{n} \left( \hat{k}_j \sum_{l \in \mathcal{L}_j}  \Delta_{l}(t) \hat{r}_{jl} -\hat{k}_j \alpha_j t
   - \int_0^t r_j  Z_j(s) ds \right) 
  + \sqrt{n \hat{k}_j r_j} \,\,B_j^a(t) + o(\sqrt{n}), \label{app:Ajn2}
\end{align}
where $B_j^a$ is a standard Brownian motion.

Now we consider the arrival process of the one-time volunteers, i.e., $A_j^n$ for $j=J+1, \dots, J+\tilde{J}$:
\begin{align}
A_j^n(t) &= N_j^a \left( \int_0^t \Lambda_j (\delta(s)) ds \right) \nonumber \\
&= N_j^a \left( \int_0^t \Big( \lambda_j^n + \sum_{l \in \mathcal{L}_j} \hat{\lambda}_{jl} \delta_{l}(s) \sqrt{n} \Big)
   ds \right) \nonumber \\
&\approx N_j^a \left( \int_0^t \Big( n \lambda_j  -\sqrt{n} \alpha_j
   + \sqrt{n} \sum_{l \in \mathcal{L}_j} \delta_l(s) \hat{\lambda}_{jl} \Big) ds \right) \nonumber \\
&= N_j^a \left( n \lambda_j t + \sqrt{n} \left( \sum_{l \in \mathcal{L}_j} \hat{\lambda}_{jl} \delta_l(s) 
   - \alpha_j \right) \right) \nonumber \\
&= N_j^a \left( n \left(\lambda_j t + \frac{1}{\sqrt{n}} \left( \sum_{l \in \mathcal{L}_j} \hat{\lambda}_{jl} \delta_l(s) 
   - \alpha_j \right) \right) \right) \nonumber \\
&= n \lambda_j t + \sqrt{n} \left(  \sum_{l \in \mathcal{L}_j} \delta_l(s) \hat{\lambda}_{jl}  - \alpha_j \right)
   + \sqrt{n \lambda_j} B_j^a(t) + o(\sqrt{n}),
    \label{app:Ajn2.2}
\end{align}
where we used Equation~(\ref{app:Nnt}) in the last step.

Substituting the approximations of $\Gamma_j^n(t)$, $N_j^s(\mu_j^n T_j^n(t))$, and $A_j^n(t)$ derived above (see Equations~(\ref{app:Gammajn}), (\ref{app:Njs})-(\ref{app:Njs.2}), and (\ref{app:Ajn2})-(\ref{app:Ajn2.2})) into Equation~(\ref{app:approxAjn}) yields an approximation for the dynamics of the volunteer queues. To derive that approximation, we first consider the repeat volunteers, i.e., classes $j=1, \dots, J$:
\begin{align}
Q_j^n(t) =&\, A_j^n(t) - N_j^n( \mu_j^n T_j(t)) - \Gamma_j^n(t) \nonumber \\
=&\, n r_j \hat{k}_j t + \sqrt{n} \Big( \hat{k}_j \sum_{l \in \mathcal{L}_j}  \Delta_{l}(t) \hat{r}_{jl} - \hat{k}_j \alpha_j t
   - \int_0^t r_j Z_j(s) ds \Big) \nonumber \\
& + \sqrt{n r_j \hat{k}_j} \,\, B_j^n(t) - n r_j \hat{k}_j t + \sqrt{n} \mu_j Y_j(t) 
   -\sqrt{n r_j \hat{k}_j} \,\, B_j^s(t) \nonumber \\
& - \sqrt{n} \int_0^t \gamma_j Z_j(s) ds+ o(\sqrt{n}). \label{app:Qjn}
\end{align}
Scaling both sides by $\sqrt{n}$ and passing to the limit formally as $n \rightarrow \infty$ gives the following:
\begin{align}
Z_j(t) =& \, X_j(t) + \hat{k}_j \sum_{l \in \mathcal{L}_j} \hat{r}_{jl} \Delta_{l}(t)  - \hat{k}_j \alpha_j t 
  - \int_0^t (r_j + \gamma_j) Z_j(s) ds + \mu_j Y_j(t),  \label{app:Zj-repeat}
\end{align}
where $X_j(t) = \sqrt{r_j \hat{k}_j} (B_j^n(t) - B_j^s(t))$ is a $(0, \sigma_j^2)$ Brownian motion with $\sigma_j^2 = 2 r_j \hat{k}_j$. This leads to Equation~(\ref{eqn:BCP-Zj1}).

Next we consider the one-time volunteers, i.e., classes $j=J+1, \dots, J+\tilde{J}$:
\begin{align}
Q_j^n(t) =&\, A_j^n(t) - N_j^n( \mu_j^n T_j(t)) - \Gamma_j^n(t) \nonumber \\
=&\, n \lambda_j t + \sqrt{n} \Big( \sum_{l \in \mathcal{L}_j}  \int_0^t \hat{\lambda}_{jl} \delta_l(s) ds - \alpha_j t \Big)
   + \sqrt{n \lambda_j} B_j^n(t) \nonumber \\
&-n \lambda_j t + \sqrt{n} \mu_j Y_j^n(t) - \sqrt{n \lambda_j} B_j^s(t)
   - \sqrt{n} \int_0^t \gamma_j Z_j(s) ds + o(\sqrt{n}). \label{app:Qjn.2}
\end{align}
Scaling both sides by $\sqrt{n}$ and passing to the limit formally as $n \rightarrow \infty$ gives the following
\begin{align}
Z_j(t) =& \, X_j(t) + \sum_{l \in \mathcal{L}_j} \hat{\lambda}_{jl} \Delta_{l}(t)  - \alpha_j t 
  - \int_0^t \gamma_j Z_j(s) ds + \mu_j Y_j(t), \,\,\, t\ge0, \label{app:Zj-onetime}
\end{align}
where $X_j(t) = \sqrt{\lambda_j} (B_j^n(t) - B_j^s(t))$ is a $(0, \sigma_j^2)$ Brownian motion with $\sigma_j^2 = 2 \lambda_j$. This leads to Equation~(\ref{eqn:BCP-Zj1.2}).

Clearly, we must have $Z_j(t) \ge 0$ for $t\ge0$ and $j=1, \dots, J+\tilde{J}$. Also note by the heavy traffic assumption that
\begin{align}
\sum_{j=1}^{J+\tilde{J}} Y_j(t) &= \sqrt{n} \left( \Big(\sum_{j=1}^J \frac{r_j \hat{k}_j}{\mu_j}  
   + \sum_{j=J+1}^{J+\tilde{J}} \frac{\lambda_j}{\mu_j}\Big) t
   - \sum_{j=1}^{J+\tilde{J}} T_j(t)  \right) \nonumber \\
&= \sqrt{n} \left( t - \sum_{j=1}^{J+\tilde{J}} T_j(t)  \right) 
   = \sqrt{n} I^n(t) = U(t), \,\,\, t\ge0, \label{app:sumYj}
\end{align}
which yields Equation~(\ref{eqn:BCP-L1}). Also, note that $U$,  $Z_j$ for $j=1, ..., J+\tilde{J}$ inherits properties (\ref{eqn:BCP-L2})-(\ref{eqn:BCP-workconserving}) from corresponding properties of the idleness process $I^n$ and the queue length processes $Q_j^n$ for $j=1, \dots, J+\tilde{J}$. Moreover, Equation~(\ref{eqn:BCP-Zj3}) follows from Equation~(\ref{eqn:Qj}) under diffusion scaling. Equation~(\ref{eqn:BCP-Zj4}) follows directly from Equation~(\ref{eqn:Deltajl}). Lastly, we verify Equation~(\ref{eqn:BCP-xi}). To that end, recall that
\begin{align}
C^n(t) &= \sum_{l=1}^L \sqrt{n} F_l \Delta(t) + p^n I^n(t) \nonumber \\
&=\sqrt{n} \Big( \sum_{l=1}^L F_l \Delta_l(t) + p \sqrt{n} I^n(t) \Big). \label{eqn:Cn1}
\end{align}
Because $U(t) \approx \sqrt{n} I^n(t)$, we have that
\begin{align*}
\frac{C^n(t)}{\sqrt{n} } \approx \sum_{l=1}^L F_l \Delta(t) + pU(t), \,\, t \ge0.
\end{align*}
Then, passing to the limit formally as $n \rightarrow \infty$ yields the approximating Brownian control problem (\ref{eqn:BCP-xi})-(\ref{eqn:BCPobj}).
\hfill $\blacksquare$

\section{Proofs of the Results in Section~\ref{sec:approx}} \label{app:proofs}


\paragraph{Proof of Proposition~\ref{prop:brownianworkload}.}

Let $\delta(\cdot)$ be a feasible policy for the BCP~(\ref{eqn:BCPobj}) and let $Z$, $Y$, $U$, $\Delta$, $\xi$ denote the corresponding processes in the BCP. Then we define $W(t)$ as in Equation~(\ref{eqn:W}). The following two identities facilitate the proof:
\begin{align}
\int_0^t \sum_{j=1}^J \sum_{l \in \mathcal{L}_j} \frac{\hat{k}_j}{\mu_j} \hat{r}_{jl} \delta_l(s) ds 
   &= \int_0^t \sum_{l=1}^L \Big( \sum_{j \in \mathcal{R}_l} \frac{\hat{k}_j}{\mu_j} \hat{r}_{jl} \Big) \delta_l(s) ds \label{eqn:identity} \\
\int_0^t \sum_{j=J+1}^{J+\tilde{J}} \sum_{l \in \mathcal{L}_j} \frac{\hat{\lambda}_{jl}}{\mu_j} \delta_l(s) ds 
   &= \int_0^t \sum_{l=1}^L \Big( \sum_{j \in \mathcal{S}_l} \frac{\hat{\lambda}_{jl}}{\mu_j}  \Big) \delta_l(s) ds \label{eqn:identity.2} 
\end{align}
Now, using the definition of $\kappa$ (see Equation~(\ref{eqn:handk})), substituting Equation~(\ref{eqn:BCP-Zj1}) into Equation~(\ref{eqn:W}) and using Equations~(\ref{eqn:BCP-Zj3})-(\ref{eqn:BCP-L1}) and (\ref{eqn:identity})-(\ref{eqn:identity.2}), it is easy to verify that $W$, $U$, and $\delta$ satisfy (\ref{eqn:workload-W1})-(\ref{eqn:workload-L1}). Also, we deduce from Equation~(\ref{eqn:BCP-Zj3}) that 
\begin{align*}
Z_j(t) = \frac{x_j}{m_j} W(t), \,\,\, j = 1, \dots, J, \,\, t \ge 0.
\end{align*}
Substituting this and Equation~(\ref{eqn:BCP-Zj4}) into Equation (\ref{eqn:BCP-xi}) yields
\begin{align*}
\xi(t) =  \int_0^t \sum_{l=1}^L  F_l  \delta_l(s) ds
     + p U(t), \,\,\, t \ge 0,
\end{align*}
proving the first claim.

To conclude the proof, let $\delta(\cdot)$ be a feasible policy for the workload formulation~(\ref{eqn:workload-obj})-(\ref{eqn:workload-L2}). Then define $Z_j(t)$ as in Equation~(\ref{eqn:Zjx}) and $\Delta_{l}(t)$ as in Equation~(\ref{eqn:BCP-Zj4}), and let for $j=1, \dots, J$ and $t \ge 0$, 
\begin{align*}
Y_j(t) = m_j \Big( Z_j(t) - X_j(t) - \hat{k}_j \sum_{l \in \mathcal{L}_j} \Delta_{l}(t) \hat{r}_{jl} 
     - \hat{k}_j \alpha_j t - \int_0^t (r_j + \gamma_j) Z_j(s) ds \Big).
\end{align*}
Similarly, for $j=J+1, \dots, J+\tilde{J}$ and $t \ge 0$, we let
\begin{align*}
Y_j(t) = m_j \Big( Z_j(t) - X_j(t) - \sum_{l \in \mathcal{L}_j} \Delta_{l}(t) \hat{\lambda}_{jl} 
     - \alpha_j t - \int_0^t \gamma_j Z_j(s) ds \Big).
\end{align*}
Now it is straightforward to verify that $Z$, $Y$, $\Delta$, $U$ satisfy~(\ref{eqn:BCP-Zj1})-(\ref{eqn:BCP-workconserving}). Similarly, substituting Equation~(\ref{eqn:BCP-Zj4}) into the following yields 
\begin{align*}
\int_0^t \sum_{l=1}^L F_l \delta_l(s) ds + pU(t)
   = \sum_{l=1}^L \Delta_l(t)  F_l + p U(t)
   = \xi(t).
\end{align*}
\hfill $\blacksquare$


\paragraph{Proof of Proposition~\ref{prop:driftrate}.}

Let $\delta(\cdot)$ be a feasible policy for the workload problem, and let $\theta(\cdot)$ be as defined in the statement of the proposition. It is easy to verify $\theta(t) \in [\theta_0, \theta_M]$ for $t \ge 0$. Thus, $\theta(\cdot)$ is feasible for the drift-rate control problem. Moreover, the workload (and idleness) processes in the formulations are identical. Also, recall that $c_1 < c_2 < \dots < c_M$. Thus, it follows from Equations~(\ref{eqn:thetam}) and~(\ref{eqn:cx}) that 
\begin{align*}
c(\theta(t)) \le \sum_{l=1}^L c_l \eta_l \delta_l(t) 
   = \sum_{l=1}^L  F_l  \delta_l(t) \,\, \mbox{ for } \,\, t\ge0,
\end{align*}
proving the first statement. To conclude the proof, let $\theta(\cdot)$ be a feasible policy for the drift-rate control problem, and define $\delta(\cdot)$ as in Equation~(\ref{eqn:driftrate}). Note that
\begin{align*}
\theta_0 + \sum_{l=1}^L \eta_l \delta_l(t) 
   = \theta_0 + \sum_{l=1}^L \sum_{j \in \mathcal{R}_l} (\hat{k}_j \hat{r}_{jl} / \mu_j) \delta_l(t)
   + \sum_{l=1}^L \sum_{j \in \mathcal{S}_l} (\hat{\lambda}_{jl} / \mu_j) \delta_l(t)
   = \theta(t), \,\, t\ge 0.
\end{align*}
Thus, the evolutions of the workload and the idleness processes are identical in the two formulations, and $\delta(\cdot)$ is feasible for the workload process. Then comparing Equations~(\ref{eqn:cjmlm})-(\ref{eqn:thetam})  and~(\ref{eqn:driftrate}) to Equation~(\ref{eqn:workload-obj}) reveals that
\begin{align*}
c(\theta(t)) =  \sum_{l=1}^L c_l \eta_l \delta_l(t) 
   = \sum_{l=1}^L  F_l \delta_l(t), \,\, t \ge 0.
\end{align*}
Consequently, the two policies have the same cost.
\hfill $\blacksquare$

\section{Proof of Theorem \ref{thm:policy}} \label{app:proofthm2}

\setcounter{equation}{0}
\renewcommand{\theequation}{\thesection\arabic{equation}}

To facilitate the proof of Theorem~\ref{thm:policy}, we first establish the following results.


\begin{lemma} \label{lem:thm2}
Under any admissible policy, we have that
\begin{align*}
& \mbox{(i) } \mathbb{E} \Big[ \int_0^t f'(Z(s)) dB(s) \Big] = 0, \,\, t\ge0. \\
& \mbox{(ii) } \lim_{t \rightarrow \infty} \frac{\mathbb{E} [ f(Z(t)) ]}{t} = 0.
\end{align*}
\end{lemma}

\begin{proof}
To establish part (i), it suffices to show that 
\begin{align}
\mathbb{E} \Big[ \int_0^t f'(Z(s))^2 ds \Big] < \infty \,\, \mbox{for} \,\, t > 0, \label{eqn:lem:thm2.1}
\end{align}
see \citet{harrison-2013}. Also recall from Equations~(\ref{eqn:f1})-(\ref{eqn:v2}) that $f'(z) = v(z) - p$ and $v(z) \in[0, p]$ for $z \ge 0$. Thus, $[f'(z)]^2 \le p^2$ for $z \ge 0$, proving (\ref{eqn:lem:thm2.1}).

To prove part (ii), first note that under any admissible policy $\mathbb{E} Z(t)/t \rightarrow 0$ as $t \rightarrow \infty$. Also note that 
\begin{align*}
\mathbb{E}f(Z(t)) &= \mathbb{E} \int_0^{Z(t)} (v(Z(s)) - p) ds + \mathbb{E} f(Z(0)) \\
|\mathbb{E}f(Z(t))| &\le \mathbb{E} \int_0^{Z(t)} |v(Z(s)) - p| ds + \mathbb{E} f(Z(0)) \\
&\le 2p \, \mathbb{E} Z(t) + \mathbb{E} f(Z(0)),
\end{align*}
where the last inequality follows because $v(z) \in [0, p ]$ for $ z \ge 0$. Thus, we conclude
\begin{align*}
\frac{|\mathbb{E}f(Z(t))|}{t} \le  2p \, \frac{\mathbb{E}Z(t)}{t}
   +  \frac{\mathbb{E}f(Z(0))}{t},
\end{align*}
which yields part (ii).
\end{proof}
\vspace{-25pt}


\paragraph{Proof of Theorem~\ref{thm:policy}.}

Note from Equations~(\ref{eqn:f1})-(\ref{eqn:f2}) and~(\ref{eqn:thetascand}) that the candidate policy satisfies the following:
\begin{align}
\beta^* = \frac{1}{2}\sigma^2 f''(z) + (\theta^*(z) - \kappa z) f'(z) + c(\theta^*(z)). \label{eqn:proofthm2.1}
\end{align}
Similarly, for an arbitrary admissible policy $\theta(\cdot)$ we have
\begin{align}
\beta^* \le \frac{1}{2}\sigma^2 f''(z) + (\theta(z) - \kappa z) f'(z) + c(\theta(z)). \label{eqn:proofthm2.2}
\end{align}
Moreover, for any admissible policy $\theta(\cdot)$ , applying Ito's lemma to $f(Z(t))$ gives
\begin{align}
f(Z(t)) - f(Z(0)) =& \int_0^t \Big[ (\theta(Z(s)) - \kappa Z(s)) f'(Z(s)) 
   + \frac{\sigma^2}{2} f''(Z(s)) \Big] ds \nonumber \\
&+ \int_0^t \sigma f'(Z(s)) dB(s) + f'(0) U(t), \label{eqn:proofthm2.3}
\end{align}
see Chapters 4 and 6 of \citet{harrison-2013}. By taking the expectations of both sides of Equation~(\ref{eqn:proofthm2.3}), it follows from Lemma~\ref{lem:thm2} that
\begin{align}
\mathbb{E}[f(Z(t))] - \mathbb{E}[f(Z(0))] =& \mathbb{E} \int_0^t [ (\theta(Z(s)) - \kappa Z(s)) 
   + \frac{\sigma^2}{2} f''(Z(s))] ds + f'(0) \mathbb{E}[U(t)]. \label{eqn:proofthm2.4}
\end{align}

For any admissible policy $\theta(\cdot)$, combining Equations~(\ref{eqn:proofthm2.2}) and~(\ref{eqn:proofthm2.4}) gives
\begin{align*}
\mathbb{E}[f(Z(t))] - \mathbb{E}[f(Z(0))] \ge
    \mathbb{E} \Big[ \int_0^t (\beta^* - c(\theta(Z(s)))) ds \Big] - p  \mathbb{E} U(t).
\end{align*}
Rearranging the terms gives
\begin{align*}
\mathbb{E} \Big[ \int_0^t c(\theta(Z(s))) ds + p U(t) \Big]
   \ge \beta^* t + \mathbb{E}[f(Z(0))] - \mathbb{E}[f(Z(t))].
\end{align*}
Combining this with Lemma~\ref{lem:thm2} gives
\begin{align*}
\underline{\lim}_{t \rightarrow \infty} \frac{1}{t}
   \mathbb{E} \Big[ \int_0^t c(\theta(Z(s))) ds + p U(t) \Big] \ge \beta^*.
\end{align*}
Similarly, for candidate policy $\theta^*(\cdot)$, combining Equations~(\ref{eqn:proofthm2.1}) and~(\ref{eqn:proofthm2.4}) gives 
\begin{align*}
\mathbb{E} \Big[ \int_0^t c(\theta(Z(s))) ds + p U(t) \Big]
   = \beta^*t + \mathbb{E}[f(Z(0))] - \mathbb{E}[f(Z(t))].
\end{align*}
It follows from Lemma~\ref{lem:thm2} that
\begin{align}
\lim_{t \rightarrow \infty} \frac{1}{t} \mathbb{E} \Big[ \int_0^t c(\theta(Z(s))) ds
   + p U(t) \Big] = \beta^*.
\end{align}
Therefore, the candidate policy is optimal and its long-run average cost is $\beta^*$.
\hfill $\blacksquare$.

\section{Proof of Proposition \ref{prop:valuefunc}} \label{app:proofprop3}

\setcounter{equation}{0}
\renewcommand{\theequation}{\thesection\arabic{equation}}

To facilitate the proof of Proposition~\ref{prop:valuefunc}, we first establish the following results.

\begin{lemma} \label{lem:integral-exp}
We have that 
\begin{align}
& \mbox{(i) } \int_{\tau}^x \exp \left\{ \frac{2 \theta z - \kappa z^2}{\sigma^2} \right\} dz = \frac{\sigma \sqrt{\pi}}{\sqrt{\kappa}} 
\exp \left\{ \frac{\theta^2}{\kappa\sigma^2} \right\} \left[ \Phi \left( \frac{x - \theta/\kappa }{\sigma / \sqrt{2\kappa}} \right) - \Phi \left( \frac{\tau- \theta/\kappa}{\sigma / \sqrt{2\kappa}} \right) \right], \label{eqn:lem:integrate-exp}  \\
& \mbox{(ii) }\int_{\tau}^x z  \exp \left\{ \frac{2 \theta z - \kappa z^2}{\sigma^2} \right\} dz  =  \frac{\theta \Halfspace \sigma \sqrt{\pi}}{\kappa\sqrt{\kappa}} \exp \left\{ \frac{\theta^2}{\kappa\sigma^2} \right\} 
\left[ \Phi \left( \frac{x - \theta/\kappa }{\sigma / \sqrt{2\kappa}} \right) - \Phi \left( \frac{\tau- \theta/\kappa}{\sigma / \sqrt{2\kappa}} \right) \right] \nonumber \\
& \ \ \ \ \ \ + \frac{\sigma^2}{2\kappa} \left[ \exp \left\{ \frac{2 \theta \tau - \kappa \tau^2}{\sigma^2} \right\} - \exp \left\{ \frac{2 \theta x - \kappa x^2}{\sigma^2} \right\} \right]. \label{eqn:lem:integrate-exp2}
\end{align}
\end{lemma}

\begin{proof}
To establish part (i), note that 
\begin{align}
&\exp  \left\{ \frac{2 \theta z - \kappa z^2}{\sigma^2} \right\} = \exp \left\{ \frac{2 \theta z - \kappa z^2 - \theta^2/\kappa + \theta^2/\kappa}{\sigma^2}  \right\}, \nonumber \\
& \ \ \ \ = \exp \left\{ -\left(  \frac{ \sqrt{\kappa} z - \theta / \sqrt{\kappa} }{\sigma}  \right)^2 +  \frac{ \theta^2}{\kappa\sigma^2}  \right\} 
= \exp \left\{ - \frac{1}{2} \left(  \frac{ z - \theta / \kappa }{\sigma/\sqrt{2\kappa}}  \right)^2\right\} \exp \left\{\frac{ \theta^2}{\kappa\sigma^2}  \right\}. \label{eqn:lemma-result}
\end{align}
Using (\ref{eqn:lemma-result}), we have that
\begin{align*}
& \int_{\tau}^x \exp \left\{ \frac{2 \theta z - \kappa z^2}{\sigma^2} \right\} dz = \exp \left\{\frac{ \theta^2}{\kappa\sigma^2}  \right\}  \int_{\tau}^x  \exp \left\{ - \frac{1}{2} \left(  \frac{ z - \theta / \kappa }{\sigma/\sqrt{2\kappa}}  \right)^2\right\} dz,  \\
& \ \ \ \ = \exp \left\{\frac{ \theta^2}{\kappa\sigma^2}  \right\} \frac{\sigma/\sqrt{\pi}}{\sqrt{\kappa}}
 \int_{\tau}^x  \frac{1}{(\sigma/\sqrt{2\kappa})\sqrt{2\pi}} \exp \left\{ - \frac{1}{2} \left(  \frac{ z - \theta / \kappa }{\sigma/\sqrt{2\kappa}}  \right)^2\right\} dz.
\end{align*}
The integral on the right hand side can be represented as the difference of two normal cdfs with mean $\theta / \kappa$ and standard deviation $\sigma/\sqrt{2\kappa}$, therefore proving (\ref{eqn:lem:integrate-exp}).

To prove part (ii), using (\ref{eqn:lemma-result}) we have that
\begin{align*}
\int_{\tau}^x z  \exp \left\{ \frac{2 \theta z - \kappa z^2}{\sigma^2} \right\} dz & = \exp \left\{\frac{ \theta^2}{\kappa\sigma^2}  \right\}   \int_{\tau}^x  z \exp \left\{ - \frac{1}{2} \left(  \frac{ z - \theta / \kappa }{\sigma/\sqrt{2\kappa}}  \right)^2\right\} dz. 
\end{align*}
Using the change of variable $u = \frac{ z - \theta / \kappa }{\sigma/\sqrt{2\kappa}}$ ($z = u \sigma / \sqrt{2\kappa} + \theta/\kappa$ and $dz = \sigma / \sqrt{2\kappa} du$) we get		
\begin{align}
& \int_{\tau}^x z   \exp \left\{ \frac{2 \theta z - \kappa z^2}{\sigma^2} \right\} dz  = \exp \left\{\frac{ \theta^2}{\kappa\sigma^2}  \right\} 
\int_{\frac{ \tau- \theta / \kappa }{\sigma/\sqrt{2\kappa}} }^\frac{ x - \theta / \kappa }{\sigma/\sqrt{2\kappa}} 
\left(u \frac{\sigma}{\sqrt{2\kappa}} + \frac{\theta}{\kappa}\right) \exp \left\{ - \frac{ u^2}{2} \right\} \frac{\sigma}{\sqrt{2\kappa}} du, \nonumber \\
& \ \ = \exp \left\{\frac{ \theta^2}{\kappa\sigma^2}  \right\} \frac{\sigma^2}{2\kappa}
\int_{\frac{ \tau- \theta / \kappa }{\sigma/\sqrt{2\kappa}} }^\frac{ x - \theta / \kappa }{\sigma/\sqrt{2\kappa}} 
u  \exp \left\{ - \frac{ u^2}{2} \right\}  du + \exp \left\{\frac{ \theta^2}{\kappa\sigma^2}  \right\} \frac{\sigma \theta}{\kappa\sqrt{2\kappa}} 
\int_{\frac{ \tau- \theta / \kappa }{\sigma/\sqrt{2\kappa}} }^\frac{ x - \theta / \kappa }{\sigma/\sqrt{2\kappa}} 
  \exp \left\{ - \frac{ u^2}{2} \right\}  du. \label{eqn:reference}
\end{align}
To complete the proof of part (ii), we rewrite the second term in (\ref{eqn:reference}) as
\begin{align*}
& \exp \left\{\frac{ \theta^2}{\kappa\sigma^2}  \right\} \frac{\sigma \theta \sqrt{2\pi}}{\kappa\sqrt{2\kappa}} 
\int_{\frac{ \tau- \theta / \kappa }{\sigma/\sqrt{2\kappa}} }^\frac{ x - \theta / \kappa }{\sigma/\sqrt{2\kappa}} 
  \frac{1}{\sqrt{2\pi}} \exp \left\{ - \frac{ u^2}{2} \right\}  du.  
\end{align*}	
Next, the integral can be stated as the difference of two standard normal cdfs as follows: 
\begin{align}
\exp \left\{\frac{ \theta^2}{\kappa\sigma^2}  \right\} \frac{\sigma \sqrt{\pi} \theta}{\kappa\sqrt{\kappa}} \left[ \Phi \left( \frac{x - \theta/\kappa }{\sigma / \sqrt{2\kappa}} \right) - \Phi \left( \frac{\tau- \theta/\kappa}{\sigma / \sqrt{2\kappa}} \right) \right] \label{eqn:reference1}.
\end{align}	

Moreover, using a second change of variable $t = u^2/2$, the first term can be written as
\begin{align}
& \exp \left\{\frac{ \theta^2}{\kappa\sigma^2}  \right\} \frac{\sigma^2}{2\kappa}  \int_{\frac{1}{2}\left(\frac{ \tau- \theta / \kappa }{\sigma/\sqrt{2\kappa}}\right)^2 }^{ \frac{1}{2}\left(\frac{x - \theta / \kappa }{\sigma/\sqrt{2\kappa}}\right)^2} \exp \left\{ - t \right\}  dt  \nonumber \\
& \ \ \ \ = \exp \left\{\frac{ \theta^2}{\kappa\sigma^2}  \right\} \frac{\sigma^2}{2\kappa} 
\left[    -\exp \left\{-\frac{1}{2}\left(\frac{ x - \theta / \kappa }{\sigma/\sqrt{2\kappa}}\right)^2  \right\} +\exp \left\{-\frac{1}{2}\left(\frac{ \tau- \theta / \kappa }{\sigma/\sqrt{2\kappa}}\right)^2  \right\} \right] \nonumber  \\
& \ \ \ \ = \frac{\sigma^2}{2\kappa} \left[ \exp \left\{ \frac{2 \theta \tau - \kappa \tau^2}{\sigma^2} \right\} - \exp \left\{ \frac{2 \theta x - \kappa x^2}{\sigma^2} \right\} \right]. \label{eqn:reference2}
\end{align}
Combining (\ref{eqn:reference1}) and (\ref{eqn:reference2}) completes the proof of part (ii).
\end{proof}
\vspace{-30pt}


\paragraph{Proof of Proposition~\ref{prop:valuefunc}.}
For notational convenience, we define $\tau_0 = \infty$. 
Also for $l=L+1$, recall that $\tau_{L+1} = 0$ and $\hat{c}_{L+1}=p$. 
For $l=1, \ldots, L+1$, we consider $x \in [\tau_l, \tau_{l-1})$ and observe that $\phi(p - v(x)) = \theta_{l-1} (p - v(x)) - c(\theta_{l-1})$  
(see Equation (\ref{eqn:phiapp}) in Appendix \ref{app:equivworkload}). Substituting this into the Bellman Equation in (\ref{eqn:bellfinal1})-(\ref{eqn:bellfinal2}), we arrive at the following:  
\begin{align}
\beta = \frac{1}{2} \sigma^2 v'(x) + \kappa x (p-v(x)) - [\theta_{l-1} (p-v(x)) - c(\theta_{l-1})] \label{eqn:ivp-proof2},
\end{align}
such that $v(\tau_l) = p - \hat{c}_l$. 

Rearranging the terms in (\ref{eqn:ivp-proof2}) gives
\begin{align*}
v'(x) + \left(\frac{2\theta_{l-1} - 2 \kappa x}{\sigma^2}\right) v(x) =& \frac{2}{\sigma^2} [\beta + p \theta_{l-1} - c(\theta_{l-1})] - \frac{2}{\sigma^2} \kappa p x.
\end{align*}
Multiplying both sides with the integrating factor $\exp \left\{ (2 \theta_{l-1} x - \kappa x^2 ) / \sigma^2 \right\}$ yields:
\begin{align*}
\left[ \exp \left\{ \frac{2 \theta_{l-1} x - \kappa x^2}{\sigma^2} \right\} v(x) \right]'
     =& \frac{2(\beta + p\theta_{l-1} - c(\theta_{l-1}))}{\sigma^2} \exp \left\{ \frac{2 \theta_{l-1} x - \kappa x^2}{\sigma^2} \right\} \\
& \ \ \ \ \		- \frac{2 \kappa p x}{\sigma^2}\exp \left\{ \frac{2 \theta_{l-1} x - \kappa x^2}{\sigma^2} \right\}   
\end{align*}

Integrating both sides of the equation over $[\tau_l, x]$  and using $v(\tau_l) = p-\hat{c}_l$ yields 
\begin{align*}
& \exp \left\{ \frac{2 \theta_{l-1} x - \kappa x^2}{\sigma^2} \right\} v(x) -  \exp \left\{ \frac{2 \theta_{l-1} \tau_l - \kappa \tau_l^2}{\sigma^2} \right\} (p-\hat{c}_l) \\
& \ \ \ \ \ = \frac{2(\beta + p\theta_{l-1} - c(\theta_{l-1}))}{\sigma^2} \int_{\tau_l}^x \exp \left\{ \frac{2 \theta_{l-1} z - \kappa z^2}{\sigma^2} \right\} dz 
 - \frac{2\kappa p }{\sigma^2} \int_{\tau_l}^x  z \exp \left\{ \frac{2 \theta_{l-1} z - \kappa z^2}{\sigma^2} \right\} dz 
\end{align*}

Using results from Lemma \ref{lem:integral-exp} for $\tau=\tau_l$ to replace integrals on the right hand side gives
\begin{align*}
& \exp \left\{ \frac{2 \theta_{l-1} x - \kappa x^2}{\sigma^2} \right\} v(x) -  \exp \left\{ \frac{2 \theta_{l-1} \tau_l - \kappa \tau_l^2}{\sigma^2} \right\} (p-\hat{c}_l) \\
&  	\ \ \ \ \ = \frac{2(\beta + p\theta_{l-1} - c(\theta_{l-1}))\sqrt{\pi}}{\sigma \sqrt{\kappa}} 
\exp \left\{ \frac{\theta_{l-1}^2}{\kappa\sigma^2} \right\} \left[ \Phi \left( \frac{x - \theta/\kappa }{\sigma / \sqrt{2\kappa}} \right) - \Phi \left( \frac{\tau- \theta/\kappa}{\sigma / \sqrt{2\kappa}} \right) \right]\\
&  	\ \ \ \ \ - \frac{2p \Halfspace \theta_{l-1} \Halfspace \sqrt{\pi}}{\sigma\Halfspace\sqrt{\kappa}} 
\exp \left\{ \frac{\theta_{l-1}^2}{\kappa\sigma^2} \right\} 
\left[ \Phi \left( \frac{x - \theta/\kappa }{\sigma / \sqrt{2\kappa}} \right) - \Phi \left( \frac{\tau- \theta/\kappa}{\sigma / \sqrt{2\kappa}} \right) \right]
 \\
&   \ \ \ \ \ - p \left[ \exp \left\{ \frac{2 \theta_{l-1} \tau_l - \kappa \tau_l^2}{\sigma^2} \right\} - \exp \left\{ \frac{2 \theta_{l-1} x - \kappa x^2}{\sigma^2} \right\} \right]
\end{align*}

Rearranging terms on the right hand side and yields the following:
\begin{align*}
& \exp \left\{ \frac{2 \theta_{l-1} x - \kappa x^2}{\sigma^2} \right\} v(x) -  \exp \left\{ \frac{2 \theta_{l-1} \tau_l - \kappa \tau_l^2}{\sigma^2} \right\} (p-\hat{c}_l) \\
&   	\ \ \ \ \ = \frac{2(\beta  - c(\theta_{l-1}))\sqrt{\pi}}{\sigma \sqrt{\kappa}}  
\exp \left\{ \frac{\theta_{l-1}^2}{\kappa\sigma^2} \right\} \left[ \Phi \left( \frac{x - \theta_{l-1}/\kappa }{\sigma / \sqrt{2\kappa}} \right) - \Phi \left( \frac{- \theta_{l-1}/\kappa}{\sigma / \sqrt{2\kappa}} \right) \right] \\
&   	\ \ \ \ \ - p \left[ \exp \left\{ \frac{2 \theta_{l-1} \tau_l - \kappa \tau_l^2}{\sigma^2} \right\} - \exp \left\{ \frac{2 \theta_{l-1} x - \kappa x^2}{\sigma^2} \right\} \right].
\end{align*}

Dividing both sides of the equation with $\exp \left\{ (2 \theta_{l-1} x - \kappa x^2) / \sigma^2 \right\}$ and solving for $v(x)$ yields
\begin{align*}
& v(x) =  \exp \left\{ \frac{2 \theta_{l-1} (\tau_l - x_l) - \kappa (\tau_l^2-x^2)}{\sigma^2} \right\} (p-\hat{c}_l) \\
&   	\ \ \ \ \ + \frac{2(\beta  - c(\theta_{l-1}))\sqrt{\pi}}{\sigma \sqrt{\kappa}}  \exp \left\{ \frac{(\kappa x- \theta_{l-1})^2}{\kappa\sigma^2} \right\} \left[ \Phi \left( \frac{x - \theta_{l-1}/\kappa }{\sigma / \sqrt{2\kappa}} \right) - \Phi \left( \frac{- \theta_{l-1}/\kappa}{\sigma / \sqrt{2\kappa}} \right) \right] \\
&   	\ \ \ \ \ - p \left[ \exp \left\{ \frac{2 \theta_{l-1} (\tau_l - x_l) - \kappa (\tau_l^2-x^2)}{\sigma^2} \right\} - 1 \right].
\end{align*}

Rearranging gives the desired result: $v(x) = u_l(x)$, and $\tau_{l-1} = v^{-1}(p-\hat{c}_{l-1})$.
\hfill $\blacksquare$

\section{Auxiliary Functions $\phi$ and $\psi$} \label{app:equivworkload}

This section further characterizes function $\phi$ and $\psi$. To this end, recall that $c(\theta_0) = 0$ and 
\begin{align*}
c(x) = \sum_{i=1}^{m-1} \hat{c}_i \hat{\eta}_i + \hat{c}_m (x - \theta_{m-1}), \,\,
     \theta_{m-1} < x \le \theta_m, \,\, m = 1, \dots, M.
\end{align*}
Also recall that for $y \in \mathbb{R}$, $\phi(y) = \sup_{x \in A} \{ yx - c(x) \}$ and $\psi(y) = \inf \arg\max_{x \in A} \{ yx - c(x) \}$.
It is straightforward to show that
\begin{align}
\psi(y) &= \left\{ 
   \begin{array}{ll}
      \theta_0, & \mbox{if } \,\, y \le \hat{c}_1, \\
      \theta_{m-1}, & \mbox{if } \,\, \hat{c}_{m-1} < y \le \hat{c}_m, \,\, m = 2, \dots, M, \\ 
      \theta_M, & \mbox{if } \,\, y > \hat{c}_M, 
   \end{array}
\right. \label{eqn:psiapp}
\end{align}
and
\begin{align}
\phi(y) &= \left\{ 
   \begin{array}{ll}
      \theta_0 y , & \mbox{if } \,\, y \le \hat{c}_1, \\
      \theta_{m-1} y - c(\theta_{m-1}), & \mbox{if } \,\, \hat{c}_{m-1} < y \le \hat{c}_m, \,\, m = 2, \dots, M, \\ 
      \theta_M y - c(\theta_M), & \mbox{if } \,\, y > \hat{c}_M, 
   \end{array}
\right. \label{eqn:phiapp}
\end{align}
Figures~\ref{fig:psiphi}(a) and~\ref{fig:psiphi}(b) illustrate functions $\psi(\cdot)$ and $\phi(\cdot)$, respectively, for the cost function displayed in Figure~\ref{fig:cx}. Also, it follows from Equations~(\ref{eqn:psiapp})-(\ref{eqn:phiapp}) that
\begin{align*}
\phi(y) = \int_0^y \psi(u) du, \,\, y \in \mathbb{R}.
\end{align*}
\begin{figure}[htbp]
\centering
\subfigure[An illustrative $\psi(\cdot)$ function with $M=4$.]{\label{fig:psi}\includegraphics[scale=0.6]{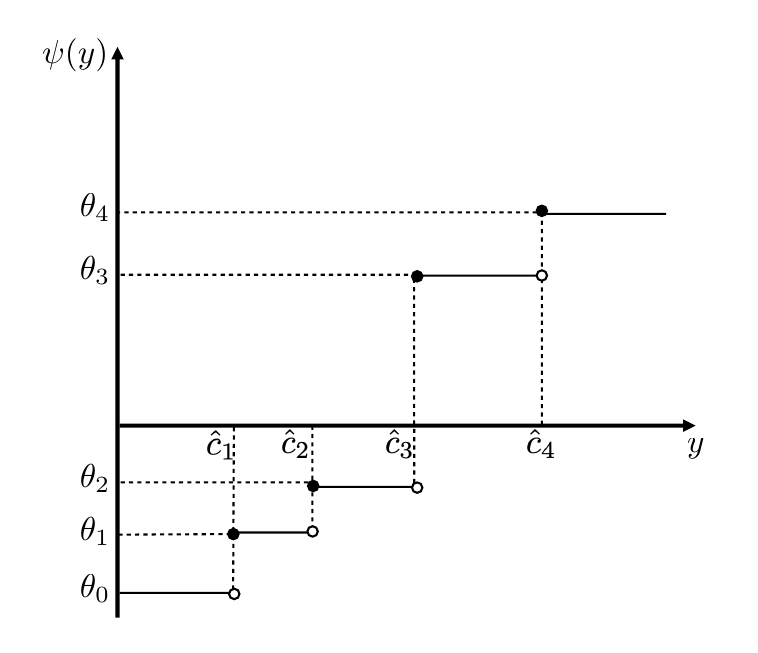}} 
\subfigure[An illustrative $\phi(\cdot)$ function with $M=4$.]{\label{fig:phi}\includegraphics[scale=0.6]{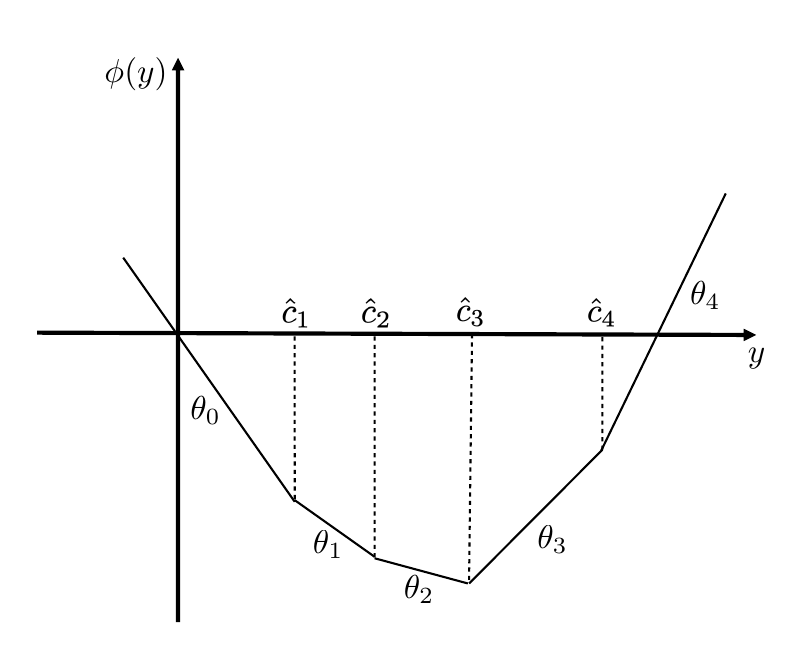}} 
\caption{Illustrative $\psi(\cdot)$ and $\phi(\cdot)$ functions.}
\label{fig:psiphi}
\end{figure}

\section{Solution to the Bellman Equation} \label{sec:solutionbellman}

To solve the Bellman equation (\ref{eqn:bellfinal1})-(\ref{eqn:bellfinal2}), we proceed as follows: First, we consider an initial value problem parameterized by $\beta$, denoted by IVP($\beta$), that is closely related to the Bellman equation. Denoting its solution by $v_{\beta}(\cdot)$, we then study its properties as $\beta$ varies. Ultimately, we show that there exists a unique $\beta^*$ such that $(\beta^*, v_{\beta}^*)$ solve the Bellman equation.

To this end, letting $\underline{\beta} = - \inf \phi(y) > 0$, consider the following initial value problem, denoted by IVP($\beta$), for $\beta > \underline{\beta}$:
\begin{align}
&\beta = \frac{1}{2} \sigma^2 v'(x) + x  \kappa (p - v(x)) - \phi (p - v(x)), \,\, x \ge 0, \label{eqn:ivp1} \\
&\mbox{subject to } v(0) =0. \label{eqn:ivp2}
\end{align}

Lemmas~\ref{lem:lipschitz}, \ref{lem:ivpcont}, and \ref{lem:gronwall} provide auxilliary results. The proof of Lemma~\ref{lem:lipschitz} is straightforward, and the proof of Lemma~\ref{lem:ivpcont} is standard\footnote{See for example Theorem 1.1.1 of \citet{Keller_2018} for a similar result on a bounded interval, whose extension to the positive real line is straightforward.} in the literature. Hence, they are omitted. For a proof of Lemma~\ref{lem:gronwall}, see, for example \citet{ata-etal-2019}. 

%
%
\begin{lemma} \label{lem:lipschitz}
The function $\phi(\cdot)$ is Lipschitz continuous with Lipschitz constant $C_L = \max \{ | \theta_k | : k = 0, 1, ..., K \}$.
\end{lemma}

%
%
\begin{lemma} \label{lem:ivpcont}
The initial value problem IVP($\beta$), has a unique continuously differentiable solution, denoted by $v_{\beta}(x)$, for $\beta > \underline{\beta}$.
\end{lemma}

%
%
\begin{lemma} \label{lem:gronwall}
Let $h$ be a nonnegative function such that 
\begin{align}
h(x) \le C + A \int_0^x h(y) dy \mbox{ for } a \le x \le b \label{eqn:h1}
\end{align}
for some constants $A$ and $C$. Then
\begin{align}
h(x) \le C e^{A(x-a)}, \,\, x \in [a, b]. \label{eqn:h2}
\end{align}
\end{lemma}

%
%
The next lemma studies how $v_{\beta}(\cdot)$ varies with $\beta$, and is proved in Appendix~\ref{app:bellmanproofs}.
\begin{lemma} \label{lem:vincbeta}
For $\beta > \underline{\beta}$, $v_{\beta}(x)$ is increasing and continuous in $\beta$.
\end{lemma}

To facilitate the analysis, define the sets $\mathcal{I}$ and $\mathcal{D}$ as follows:
\begin{align*}
\mathcal{I} &= \{ \beta > \underline{\beta} : v_{\beta} \mbox{ is increasing} \} \\
\mathcal{D} &= \{ \beta > \underline{\beta} : \exists x_{\beta} \mbox{ such that } v_{\beta} 
   \mbox{ is nondecreasing for } x \in (0, x_{\beta}) \mbox{ and } \\
& \quad \quad \quad \quad \quad \quad \mbox{nonincreasing on } (x_{\beta}, \infty) \}
\end{align*}

%
%
\begin{lemma} \label{lem:hk}
For $\beta > \underline{\beta}$, if $v_\beta^{\prime}(x^*) =0$ for some $x^* >0$, then $\kappa (p - v_{\beta}(x^*)) > 0$.
\end{lemma}
\begin{proof}
Note from IVP($\beta$) (see Equations~(\ref{eqn:ivp1})-(\ref{eqn:ivp2})) that
\begin{align}
\kappa (p - v_{\beta}(x^*)) = \frac{\beta + \phi( p - v_{\beta}(x^*))}{x^*}. \label{eqn:hke}
\end{align}
Also note from the definition $\underline{\beta} = - \inf_{y \in \Re} \phi(y)$ that $\beta + \phi(y) > 0$ for $\beta > \underline{\beta}$ and $y \in \Re$. Thus, the result follows from Equation~(\ref{eqn:hke}).
\end{proof} \vspace{-15pt}

%
%
\begin{lemma} \label{lem:vinc}
For $\beta > \underline{\beta}$, $v_{\beta}(\cdot)$ increases to its supremum. 
\end{lemma}
\begin{proof}
Suppose not. Then by continuity of $v_{\beta}$ and its derivative, there exists $x_2 > x_1 > 0$ such that
\begin{align}
0 = v_{\beta}^{\prime} (x_1) &\le v_{\beta}^{\prime} (x_2)  \label{eqn:vprime12.1} \\
v_{\beta} (x_1) &= v_{\beta} (x_2) \label{eqn:vprime12.2} 
\end{align}
By the IVP($\beta$) (see Equations~(\ref{eqn:ivp1})-(\ref{eqn:ivp2})), we write
\begin{align}
\beta &= \frac{1}{2} \sigma^2 v_{\beta}^{\prime}(x_1) 
   + x_1\kappa (p-v_{\beta}(x_1)) - \phi(p-v_{\beta}(x_1)), \label{eqn:betax1} \\
\beta &= \frac{1}{2} \sigma^2 v_{\beta}^{\prime}(x_2) 
   + x_2\kappa (p-v_{\beta}(x_2)) - \phi(p-v_{\beta}(x_2)), \label{eqn:betax2}
\end{align}
Subtracting~(\ref{eqn:betax1}) from (\ref{eqn:betax2}) while using Equations~(\ref{eqn:vprime12.1})-(\ref{eqn:vprime12.2}) gives
\begin{align}
0 = \frac{1}{2} \sigma^2 v_{\beta}^{\prime} (x_2) + (x_2 - x_1) \kappa(p- v_{\beta}(x_1)) > 0, \label{eqn:contradiction}
\end{align}
a contradiction, where the inequality follows because $\kappa(p- v_{\beta}(x_1)) >0$ by Lemma~\ref{lem:hk}.
\end{proof} \vspace{-15pt}

%
%
\begin{lemma} \label{lem:noconstantv}
For $\beta > \underline{\beta}$, there is no interval on $(0, \infty)$ on which $v_{\beta}(x)$ is constant. In particular, for $\beta > \underline{\beta}$, the set $\{ x \ge 0 : v_{\beta}^{\prime} (x) = 0 \}$ has Lebesgue measure zero.
\end{lemma}
\begin{proof}
Suppose not. Then there exists $(x_1, x_2) \subset (0, \infty)$ such that $v_{\beta}(x) =c$ for all $x \in [x_1, x_2]$. In particular, $v_{\beta}^{\prime} (x) =0$ on $(x_1, x_2)$. Substituting this into IVP($\beta$) (see Equations~(\ref{eqn:ivp1})-(\ref{eqn:ivp2})) gives the following
\begin{align}
\beta = x  \kappa (p-c) - \phi(p-c), \,\, x \in (x_1, x_2),
\end{align}
which clearly is a contradiction. Next, we prove that the set $\{ x \ge 0 : v_{\beta}^{\prime}(x) =0 \}$ has zero measure. Suppose not, then it must include an interval $(x_1, x_2)$ (see \citealt{royden-1998}), which is a contradiction by the first half of the lemma.
\end{proof} \vspace{-15pt}

%
%
\begin{corollary} \label{cor:vbetaincdec} We have that: 

\noindent i) If $v_{\beta}(x)$ is nondecreasing on $(a, b) \subset (0, \infty)$, it is increasing on $(a, b)$.

\noindent ii) If $v_{\beta}(x)$ is nonincreasing on $(a, b) \subset (0, \infty)$, it is decreasing on $(a, b)$.
\end{corollary}
\begin{proof}
\emph{Part i)} Let $a < x_1 < x_2 < b$. By Lemma~\ref{lem:noconstantv}, $v_{\beta}^{\prime}(x) > 0$  a.e.  on $(a, b)$. Thus,
\begin{align}
v_{\beta}(x_2) = v_{\beta}(x_1) + \int_{x_1}^{x_2} v_{\beta}^{\prime}(y) dy > v_{\beta}(x_1),
\end{align}
because $v_{\beta}^{\prime}(y) > 0$ for a.e. $y \in[a, b]$ by Lemma~\ref{lem:noconstantv}, which proves Part i). The proof of Part ii) is virtually identical, and hence, it is omitted.
\end{proof} \vspace{-15pt}

To shed further light on the structure of sets $\mathcal{I}$, $\mathcal{D}$, define the sets $\mathcal{\tilde{I}}$, $\mathcal{\tilde{D}}$ as follows:
\begin{align}
\mathcal{\tilde{I}} &= \{ \beta > \underline{\beta} : v_{\beta} \mbox{ is nondecreasing}\}, \label{eqn:tildeI} \\
\mathcal{\tilde{D}} &= \{ \beta > \underline{\beta} : \exists x_{\beta} \mbox{ such that } v_{\beta} 
   \mbox{ is increasing on } (0, x_{\beta}) \mbox{ and decreasing on } (x_{\beta}, \infty) \} 
   \label{eqn:tildeD}
\end{align}
%

%
%
The following corollary is immediate from Corollary~\ref{cor:vbetaincdec}. Lemma~\ref{lem:betainD} is proved in Appendix~\ref{app:bellmanproofs}; and it facilitates Lemma~\ref{lem:IDpartition}.
\begin{corollary} \label{cor:ID}
We have that $\mathcal{I} = \mathcal{\tilde{I}}$ and $\mathcal{D} = \mathcal{\tilde{D}}$.
\end{corollary}

%
%
\begin{lemma} \label{lem:betainD}
For $\beta > \underline{\beta}$, $\beta \in \mathcal{D}$ if and only if there exists $x_0 > 0$ such that $v_{\beta}^{\prime}(x_0) < 0$.
\end{lemma}
\begin{lemma} \label{lem:IDpartition}
The sets $\mathcal{I}$  and  $\mathcal{D}$ partition $(\underline{\beta}, \infty)$, i.e., $\mathcal{I} \cup \mathcal{D} = (\underline{\beta}, \infty)$ and $\mathcal{I} \cap \mathcal{D} = \emptyset$.
\end{lemma}
\begin{proof}
Fix $\beta > \underline{\beta}$. If $v_{\beta}^{\prime} (x) \ge 0$ for all $x > 0$, then $v_{\beta}$ is nondecreasing. By Corollary~\ref{cor:vbetaincdec}, it is increasing too, so $\beta \in \mathcal{I}$. Otherwise, there exists $x>0$ such that $v_{\beta}^{\prime} (x) < 0$ and $\beta \in \mathcal{D}$ by Lemma~\ref{lem:betainD}. Because one of these will be true for all $\beta > \underline{\beta}$, the result follows.
\end{proof} \vspace{-20pt}

Lemmas~\ref{lem:vbetaneginf}, \ref{lem:IDnotempty}, and \ref{lem:betaI} are proved in Appendix~\ref{app:bellmanproofs}.

%
%
\begin{lemma} \label{lem:vbetaneginf}
If $\beta > \underline{\beta}$ and $\beta \in \mathcal{D}$, then $\lim_{x \rightarrow \infty} v_{\beta}(x) = - \infty$.
\end{lemma}
%

%
%
\begin{lemma} \label{lem:IDnotempty}
We have that $\mathcal{I} \not= \emptyset$ and  $\mathcal{D} \not= \emptyset$.
\end{lemma}
%

%
%
\begin{lemma} \label{lem:betaI}
For $\beta > \underline{\beta}$, we have that $\beta \in \mathcal{I}$ and $v_{\beta}$ is unbounded if and only if there exists $x_0 > 0$ such that $v_{\beta}(x_0) \ge p $.
\end{lemma}
%

%
%
\begin{lemma} \label{lem:equiv}
Let $\beta > \underline{\beta}$. The following are equivalent:

i) $\beta \in \mathcal{D}$.

ii) There exists $x > 0$ such that $v_{\beta}^{\prime}(x) < 0$.

iii) There exists $x > 0$ such that $v_{\beta}(x) < 0$.

iv) $\lim_{x \rightarrow \infty} v_{\beta}(x) = - \infty$.
\end{lemma}
\begin{proof}
Statements i) and ii) are equivalent by Lemma~\ref{lem:betainD}. Statement i) implies iv) by Lemma~\ref{lem:vbetaneginf}. Clearly, iv) implies iii). To conclude the proof, we argue iii) implies i) by contradiction. Suppose iii) holds but $\beta \in \mathcal{I}$, but the latter implies $v_{\beta}(x)$ is nondecreasing in $x$. Because $v_{\beta}(0) = 0$, it implies $v_{\beta}(x) \ge 0$ for all $x > 0$, a contradiction.
\end{proof} \vspace{-20pt}

%
%
\begin{lemma} \label{lem:14}
Let $\beta_2 > \beta_1 > \underline{\beta}$. If $\beta_2 \in \mathcal{D}$, then $\beta_1 \in \mathcal{D}$.
\end{lemma}
\begin{proof}
If $\beta_2 \in \mathcal{D}$, then by Lemma~\ref{lem:equiv}, there exists $x_0$ such that $v_{\beta_2}(x_0) < 0$. Then by Lemma~\ref{lem:vincbeta}, $ v_{\beta_1}(x_0) < v_{\beta_2}(x_0) < 0$. Thus, by Lemma~\ref{lem:equiv}, $\beta_1 \in \mathcal{D}$.
\end{proof} \vspace{-20pt}

%
%
Next, we define $\beta^* = \inf \mathcal{I}$. Ultimately, we prove that this is the optimal long-run average cost for the workload problem.
\begin{lemma} \label{lem:15}
We have that $\beta^* > \underline{\beta}$.
\end{lemma}
\begin{proof}
Because $\mathcal{D} \not= \emptyset$ by Lemma~\ref{lem:IDnotempty}, there exists $\tilde{\beta} > \underline{\beta}$ such that $\tilde{\beta} \in \mathcal{D} $. Thus, by Lemma~\ref{lem:14}, $(\underline{\beta}, \tilde{\beta}) \subset \mathcal{D}$. Moreover, for all $\beta \in \mathcal{I}$, we must have $\beta \ge \tilde{\beta} $. Otherwise, $\beta \in \mathcal{D}$, a contradiction. This implies then that
$\beta^* = \inf \mathcal{I} \ge \tilde{\beta} > \underline{\beta}$.
\end{proof} \vspace{-20pt}

%
%
\begin{lemma} \label{lem:16}
We have that $\beta^* \in \mathcal{I}$ and $v_{\beta^*}$ is bounded.
\end{lemma}
\begin{proof}
We argue by contradiction that $\beta^* \in \mathcal{I}$. Suppose  $\beta^* \notin \mathcal{I}$, then by Lemma~\ref{lem:IDpartition}, $\beta^* \in \mathcal{D}$ because $\beta^* > \underline{\beta}$ by Lemma~\ref{lem:14}. Moreover, by Lemma~\ref{lem:equiv}, there exists $x_1 > 0$ such that $v_{\beta^*}(x_1) < 0$. Because $v_{\beta}(x_1)$ is continuous in $\beta$ by Lemma~\ref{lem:vincbeta}, there exists $\delta > 0$ such that $v_{\beta}(x_1) < 0$ for $\beta \in (\beta^* - \delta, \beta^* + \delta)$. However, by definition of $\beta^*$, there exists $\tilde{\beta} \in (\beta^*, \beta^* + \delta)$ such that $\tilde{\beta} \in \mathcal{I}$, in particular $v_{\tilde{\beta}}(x) \ge 0$ for all $x > 0$ (by definition of $\mathcal{I}$ and $v_{\beta}(0) = 0$), which implies $v_{\tilde{\beta}}(x_1) \ge 0$, a contradiction. Thus, $\beta^* \in \mathcal{I}$.

We prove that $v_{\beta^*}$ is bounded by contradiction.
Suppose not. Then there exists $x_0$ such that $v_{\beta^*}(x_0) > 2  p$. Moreover, by Lemma~\ref{lem:vincbeta}, there exists $\varepsilon > 0$ such that $v_{\beta^* - \varepsilon}(x_0) >  p$, which implies $\beta^* - \varepsilon \in \mathcal{I}$ by Lemma~\ref{lem:betaI}. That $\beta^* - \varepsilon \in \mathcal{I}$ , however, contradicts the definition of $\beta^*$.
\end{proof} \vspace{-20pt}

%
%
\begin{lemma} \label{lem:17}
We have that $\mathcal{D} = (\underline{\beta}, \beta^*)$ and $\mathcal{I} = [\beta^*, \infty)$.
\end{lemma}
\begin{proof}
It follows from the definition of $\beta^*$ and that $\beta^* \in \mathcal{I}$ that $(\underline{\beta}, \beta^*) \subset \mathcal{D}$. Next, we argue that no $\beta$ in $[\beta^*, \infty)$ belongs to $\mathcal{D}$. Then the result follows from Lemma~\ref{lem:IDpartition}. To this end, suppose there exists $\beta \in [ \beta^*, \infty) \cap \mathcal{D}$. Then, by Lemma~\ref{lem:14}, $\beta^* \in \mathcal{D}$, a contradiction. Hence, the result follows.

\end{proof} \vspace{-25pt}

%
%
\begin{lemma} \label{lem:18}
We have that $\lim_{x \rightarrow \infty} v_{\beta^*}(x) =  p$.
\end{lemma}
\begin{proof}
Because $v_{\beta^*}$ is bounded, it follows from Lemma~\ref{lem:betaI} that $v_{\beta^*}(x) < p$ for $x > 0$. Moreover, because $v_{\beta^*}$ in nondecreasing (since $\beta^* \in \mathcal{I}$), 
\begin{align}
\lim_{x \rightarrow \infty} v_{\beta^*}(x) \le  p. \label{eqn:A28}
\end{align} 

Let $K = \lim_{x \rightarrow \infty} v_{\beta^*}(x)$ and suppose $K < p$. Rearranging the terms in the IVP($\beta$) (see Equations~(\ref{eqn:ivp1})-(\ref{eqn:ivp2})) gives
\begin{align*}
\frac{1}{2}\sigma^2 v_{\beta^*}^{\prime}(x) = \beta^* 
   +x \kappa \left[ v_{\beta}(x) - p \right] - \phi(p - v_{\beta}(x))
\end{align*}
We conclude by continuity of $\phi$ that
\begin{align*}
\frac{\sigma^2}{2}  \lim_{x \rightarrow \infty} v_{\beta^*}^{\prime}(x) 
   &= \beta^* - \phi(p-K) + \kappa \lim_{x \rightarrow \infty} x \left[ K - p\right] 
   = - \infty.   
\end{align*}
Therefore, by the above equation, there exists $x_0 > 0$ such that $v_{\beta^*}^{\prime}(x_0) < 0$. Then Lemma~\ref{lem:equiv} implies $\beta^* \in \mathcal{D}$, a contradiction. Thus, $K \ge p$. Combining this with~(\ref{eqn:A28}) completes the proof.
\end{proof} \vspace{-20pt}

%
%
\begin{corollary} \label{cor:solvebellman}
We have that $(\beta^*, v_{\beta^*})$ solve the Bellman equation (\ref{eqn:bellfinal1})-(\ref{eqn:bellfinal2}).
\end{corollary}
\begin{proof}
Note by construction $v_{\beta^*}$ satisfies~(\ref{eqn:bellfinal1}) and $v_{\beta^*}(0) = 0$ (see definition of IVP($\beta$) in Equations~(\ref{eqn:ivp1})-(\ref{eqn:ivp2})). Moreover, since $\beta^* \in \mathcal{I}$ (by Lemma~\ref{lem:16}), we have that $v_{\beta^*}$ is increasing. Lastly, Lemma~\ref{lem:18} shows
$\lim_{x \rightarrow \infty} v_{\beta^*}(x) =  p$.
Thus, $(\beta^*, v_{\beta^*})$ solve the Bellman equation~(\ref{eqn:bellfinal1})-(\ref{eqn:bellfinal2}).

\end{proof}\vspace{-15pt}

\section{Proofs of the Technical Results} \label{app:bellmanproofs}


\paragraph{Proof of Lemma~\ref{lem:vincbeta}.}

First, we prove that $v_{\beta_2} > v_{\beta_1}$ for $x > 0$ and $\beta_2 > \beta_1 > \underline{\beta}$. Fix $\beta_2 > \beta_1 > \underline{\beta}$. Then define
$\tilde{\phi}(y) = \phi(y) - \theta_M y \mbox{ for } y \in \mathbb{R}$,
and note that $\tilde{\phi}$ is a decreasing function. Then substituting $\phi(y) = \tilde{\phi}(y) +\theta_M y$ in Equations~(\ref{eqn:ivp1})-(\ref{eqn:ivp2}) and rearranging the terms give
\begin{align}
v'_{\beta_i}(x) - \frac{2 \kappa}{\sigma^2} \left( x - \frac{\theta_M}{\kappa} \right) v_{\beta_i}(x)
   = \frac{2}{\sigma^2} (\beta_i + \theta_M p) 
   + \frac{2}{\sigma^2} \tilde{\phi}(p - v_{\beta_i}(x))
   - \frac{2}{\sigma^2}  \kappa p x, \,\, x\ge0. \label{eqn:lem:vincbeta1}
\end{align}
We argue by contradiction. Suppose $v_{\beta_1}(x) \ge v_{\beta_2}(x)$ for some $x>0$. Let 
\begin{align*}
\hat{x} = \inf \{ x \ge 0 : v_{\beta_1}(x) \ge v_{\beta_2}(x) \}.
\end{align*}

We proceed with the following two cases: Case (i) $\hat{x} > 0$, Case (ii) $\hat{x} = 0$.

\noindent \emph{Case (i):} $\hat{x} > 0$. Then by continuity of $v_{\beta_i}(\cdot)$, we conclude that 
\begin{align}
v_{\beta_1}(\hat{x}) = v_{\beta_2}(\hat{x})
   \,\, \mbox{and} \,\, v_{\beta_1}(x) < v_{\beta_2}(x) \,\, \mbox{on} \,\, [0, x^*). \label{eqn:lem:vincbeta2}
\end{align}
Note that multiplying Equation~(\ref{eqn:lem:vincbeta1}) by the integrating factor $\exp\big\{ \frac{\kappa x^2 - 2 \theta_M x}{\sigma^2} \big\}$ yields the following: 
\begin{align*}
\Bigg[ v_{\beta_i}(x) \exp \Big\{ \frac{\kappa x^2 - 2 \theta_M x}{\sigma^2} \Big\} \Bigg]'
   =& \exp \Big\{ \frac{\kappa x^2 - 2 \theta_M x}{\sigma^2} \Big\} \frac{2}{\sigma^2} 
   (\beta_i + \theta_M p) \\
& -\frac{2}{\sigma^2} \kappa p x \exp \Big\{ \frac{\kappa x^2 - 2 \theta_M x}{\sigma^2} \Big\} \\
& +\frac{2}{\sigma^2} \exp \Big\{ \frac{\kappa x^2 - 2 \theta_M x}{\sigma^2} \Big\}
   \tilde{\phi} (p - v_{\beta_i}(x)), \,\, x \ge0.
\end{align*}
Integrating both sides of this over $[0, \hat{x}]$ gives the following: 
\begin{align}
\exp \Big\{ \frac{\kappa \hat{x}^2 - 2 \theta_M \hat{x}}{\sigma^2} \Big\} v_{\beta_i}(\hat{x})
   =& \int_0^{\kappa \hat{x}} \exp \Big\{ \frac{\kappa x^2 - 2 \theta_M x}{\sigma^2} \Big\} 
    \frac{2}{\sigma^2} (\beta_i + \theta_M p) dx \nonumber \\
& -\frac{2}{\sigma^2}  \int_0^{\hat{x}} \kappa p x \exp \Big\{ \frac{\kappa x^2 - 2 \theta_M x}{\sigma^2} \Big\}  dx 
   \nonumber\\
& +\frac{2}{\sigma^2} \int_0^{\hat{x}}  \exp \Big\{ \frac{\kappa x^2 - 2 \theta_M x}{\sigma^2} \Big\}
   \tilde{\phi} (p - v_{\beta_i}(x)) dx. \label{eqn:lem:vincbeta3}
\end{align}
Considering Equation~(\ref{eqn:lem:vincbeta3}) for $i=1,2$ and taking the difference gives
\begin{align*}
0 =& \int_0^{\hat{x}}  \frac{2( \beta_2 - \beta_1)}{\sigma^2} 
   \exp \Big\{ \frac{\kappa x^2 - 2 \theta_M x}{\sigma^2} \Big\} dx \\
   & + \frac{2}{\sigma^2} \int_0^{\hat{x}}  \exp \Big\{ \frac{\kappa x^2 - 2 \theta_M x}{\sigma^2} \Big\} 
   [\tilde{\phi} (p - v_{\beta_2}(x)) - \tilde{\phi} (p - v_{\beta_1}(x)) ] dx > 0,
\end{align*}
where the inequality follows from Equation~(\ref{eqn:lem:vincbeta2}) and the monotonicity of $\tilde{\phi}$. This yields a contradiction in Case (i).

\noindent \emph{Case (ii):} $\hat{x} =0$. In this case, there exists a sequence $\{ x_n \}$ such that $x_n \downarrow 0$ and $v_{\beta_1}(x_n) \ge v_{\beta_2}(x_n)$. In particular,
\begin{align}
\frac{v_{\beta_1}(x_n)}{x_n} \ge \frac{v_{\beta_2}(x_n)}{x_n} \,\, \mbox{for} \,\, n \ge1. \label{eqn:lem:vincbeta4}
\end{align}
Because $v_{\beta_1}(0) = v_{\beta_2}(0)$, taking the limit in Equation~(\ref{eqn:lem:vincbeta4}) as $n \rightarrow \infty$ gives $v'_{\beta_2}(0) \le v'_{\beta_1}(0)$. Combining this with Equations~(\ref{eqn:ivp1})-(\ref{eqn:ivp2}) gives 
\begin{align*}
v'_{\beta_2}(0) = \frac{2 \beta_2}{\sigma^2} + \frac{2}{\sigma^2} \phi(p)
   \le v'_{\beta_1}(0) = \frac{2 \beta_1}{\sigma^2} + \frac{2}{\sigma^2} \phi(p).
\end{align*}
Or, $\beta_2 \le \beta_1$, which is a contradiction, too.

Combining the two cases proves that $v_{\beta_2}(x) > v_{\beta_1}(x)$ for $x > 0$.

To conclude the proof, we next prove that $v_{\beta}(x)$ is continuous in $\beta$ on $(\underline{\beta}, \infty)$. To this end, fix $\beta_2 > \beta_1 > \underline{\beta}$ and note from Equation~(\ref{eqn:ivp1}) that
\begin{align}
\frac{1}{2} \sigma^2 v'_{\beta_i}(x) = \beta_i - x  \kappa (p - v_{\beta_i}(x))
   - \phi(p - v_{\beta_i}(x)), \,\, x \ge 0, \,\, i = 1,2. \label{eqn:lem:vincbeta5}
\end{align}
Integrating both sides of Equation~(\ref{eqn:lem:vincbeta5}) on $[0,y]$, we conclude that for $y> 0$
\begin{align*}
v_{\beta_2}(y) = \frac{2}{\sigma^2} \beta_2 y - \frac{\kappa p}{\sigma^2} y^2
   -\frac{2\kappa}{\sigma^2} \int_0^y s v_{\beta_2}(s) ds
   -\frac{2}{\sigma^2} \int_0^y \phi(p - v_{\beta_2}(s)) ds, \\
v_{\beta_1}(y) = \frac{2}{\sigma^2} \beta_1 y - \frac{\kappa p}{\sigma^2} y^2
   -\frac{2\kappa}{\sigma^2} \int_0^y s v_{\beta_1}(s) ds
   -\frac{2}{\sigma^2} \int_0^y \phi(p - v_{\beta_1}(s)) ds.
\end{align*}
Fix $\bar{x} > 0$. Then taking the difference of the preceding two equations for $y \in [0, \bar{x}]$ and using the Lipschitz continuity of $\phi$ (see Lemma~\ref{lem:lipschitz}) yield the following:
\begin{align*}
| v_{\beta_2}(y) - v_{\beta_1}(y)| &\le \frac{2}{\sigma^2} | \beta_2- \beta_1 | y
   + \frac{2\kappa}{\sigma^2} y \int_0^y | v_{\beta_2}(s) - v_{\beta_1}(s)| ds
   + \frac{2L}{\sigma^2} \int_0^y  | v_{\beta_2}(s) - v_{\beta_1}(s)| ds \\
&\le \frac{2}{\sigma^2} | \beta_2- \beta_1 | \bar{x} + \frac{2\kappa}{\sigma^2} (\bar{x} + L)
   \int_0^y | v_{\beta_2}(s) - v_{\beta_1}(s)| ds.
\end{align*}
Note that letting $F(y) = | v_{\beta_2}(y) - v_{\beta_1}(y)|$, we have for $y \in [0, \bar{x}]$ that
\begin{align*}
F(y) \le \frac{2}{\sigma^2} | \beta_2 - \beta_1| \bar{x} + \frac{2(\kappa \bar{x} + L)}{\sigma^2}
   \int_0^y F(s) ds.
\end{align*}
Thus, by Gronwall's inequality (see Lemma~\ref{lem:gronwall}), we write
\begin{align*}
| v_{\beta_2}(y) - v_{\beta_1}(y)| \le \frac{2 \bar{x}}{\sigma^2} | \beta_2 - \beta_1| 
   \exp \Big\{ \frac{2(\kappa \bar{x} + L)}{\sigma^2} y \Big\}, \,\, \forall y \in [0, \bar{x}).
\end{align*}
In particular, $v_{\beta}(y)$ is continuous in $\beta$.
\hfill $\blacksquare$

\paragraph{Proof of Lemma~\ref{lem:betainD}.}
First, we prove that if $\beta \in \mathcal{D}$, then there exists $x_0 > 0$ such that $v_{\beta}^{\prime}( x_0) < 0$. Suppose not. Then $v_{\beta}^{\prime}( x_0) \ge 0$ for all $x \ge 0$, and $v_{\beta}$ is nondecreasing. Then, by Part i) of Corollary~\ref{cor:vbetaincdec}, $v_{\beta}$ is increasing, and $\beta \in \mathcal{I}$, a contradiction.

Next, we prove that if there exists $x_0 > 0$ such that $v_{\beta}^{\prime}( x_0) < 0$, then $\beta \in \mathcal{D}$. To this end, note by Lemma~\ref{lem:vinc} that because $v_{\beta}$ increases strictly to its supremum, it must be that $v_{\beta}$ achieves its maximum at some $x^*< x_0$ and that $v_{\beta}^{\prime}( x^*) = 0$. Then it follows from Lemma~\ref{lem:hk} that $\kappa (p - v_{\beta}(x^*)) > 0$. 
We establish that $\beta \in \mathcal{D}$ by contradiction. Suppose not. Then there exists $x_2 > x_1 > x^*$ such that
\begin{align}
v_{\beta}^{\prime}( x_1) &\le 0 < v_{\beta}^{\prime}( x_2 ) \label{eqn:betainD1} \\
v_{\beta}( x_1) &= v_{\beta}( x_2) \label{eqn:betainD2} 
\end{align}
It also follows from the IVP (see Equations~(\ref{eqn:ivp1})-(\ref{eqn:ivp2})) that
\begin{align}
\beta &= \frac{1}{2} \sigma^2 v_{\beta}^{\prime}(x_1) 
   + x_1 \kappa (p-v_{\beta}(x_1)) - \phi(p-v_{\beta}(x_1)), \label{eqn:betainD3} \\
\beta &= \frac{1}{2} \sigma^2 v_{\beta}^{\prime}(x_2) 
   + x_2 \kappa (p-v_{\beta}(x_2)) - \phi(p-v_{\beta}(x_2)), \label{eqn:betainD4}
\end{align}
Also note that because $v_{\beta}$ is maximized at $x^*$, we have $v_{\beta}(x_1) \le v_{\beta}(x^*)$ and thus
\begin{align}
\kappa (p - v_{\beta}(x_1)) \ge \kappa (p - v_{\beta}(x^*)) > 0. \label{eqn:betainD5}
\end{align}
Then subtracting (\ref{eqn:betainD3}) from (\ref{eqn:betainD4}) while using (\ref{eqn:betainD1})-(\ref{eqn:betainD2}) and (\ref{eqn:betainD5}) gives
\begin{align}
0 = \frac{1}{2} \sigma^2 [ v_{\beta}^{\prime}(x_2) - v_{\beta}^{\prime}(x_1)]
   + (x_2 - x_1) \kappa (p - v_{\beta}(x_1)) > 0,
\end{align}
a contradiction. Thus, $\beta \in \mathcal{D}$.
\hfill $\blacksquare$


\paragraph{Proof of Lemma~\ref{lem:vbetaneginf}.}
It follows from Lemma~\ref{lem:betainD} that there exists $x_0 > 0$ such that $v_{\beta}^{\prime} (x_0) < 0$. Then by Lemma~\ref{lem:vinc}, $v_{\beta}$ achieves its maximum at $x^* < x_0$. In particular, $v_{\beta}^{\prime} (x^*) = 0$, which in turn implies by Lemma~\ref{lem:hk} that $\kappa (p - v_{\beta}(x^*)) > 0$. Moreover, because $x^*$ is the maximizer, we conclude that
\begin{align}
\kappa (p - v_{\beta}(x)) \ge \kappa (p - v_{\beta}(x^*)) > 0 \,\, \mbox{for } x>0. \label{eqn:vbetaneginf1}
\end{align}

To facilitate the analysis, define $\tilde{\phi}(y) = \phi(y) - \theta_M y$ for $y \in \Re$ and note that $\tilde{\phi}$ is a decreasing function. Also, note from the IVP($\beta$) (see Equations~(\ref{eqn:ivp1})-(\ref{eqn:ivp2})) that
\begin{align*}
\frac{1}{2} \sigma^2 v_{\beta}^{\prime}(x) &= \beta + \phi (p - v_{\beta}(x)) 
   -x \kappa (p-v_{\beta}(x)) \\
&= \beta + \tilde{\phi} (p - v_{\beta}(x)) + \theta_M ( p - v_{\beta}(x))
   -x\kappa (p-v_{\beta}(x^*)) \\
& \le \beta + \tilde{\phi} (p - v_{\beta}(x^* )) + \theta_M ( p - v_{\beta}(x))
   -x\kappa (p-v_{\beta}(x^*)).
\end{align*}
That is, 
\begin{align}
\frac{1}{2} \sigma^2 v_{\beta}^{\prime}(x) \le  \beta + \tilde{\phi} (p - v_{\beta}(x^* )) 
   + \theta_M  p - \theta_M v_{\beta}(x) -x \kappa (p-v_{\beta}(x^*)). \label{eqn:vbetaneginf2}
\end{align}

We have two cases to consider:

\noindent \emph{Case 1:} $\theta_M < 0$. In this case, we conclude from Equation~(\ref{eqn:vbetaneginf2}) that 
\begin{align*}
\frac{1}{2} \sigma^2 v_{\beta}^{\prime}(x) \le
   \beta + \tilde{\phi} (p - v_{\beta}(x^*)) + \theta_M  p - \theta_M v_{\beta}(x^*)
   -x \kappa (p-v_{\beta}(x^*)).
\end{align*}
Integrating both sides and using $v_{\beta}(0)=0$ give
\begin{align*}
\frac{1}{2} \sigma^2 v_{\beta}(x) \le
   x[ \beta + \tilde{\phi} (p - v_{\beta}(x^*)) + \theta_M  p - \theta_M v_{\beta}(x^*)]
   -\frac{x^2}{2} \kappa (p-v_{\beta}(x^*)),
\end{align*}
which proves that $v_{\beta}(x) \rightarrow - \infty$ as $x \rightarrow \infty$.

\noindent \emph{Case 2:} $\theta_M > 0$. In this case, the proof proceeds by contradiction. Suppose that $v_{\beta}(x)$ is bounded from below, i.e., there exists $0 < K < \infty$ such that
\begin{align*}
v_{\beta}(x) \ge - K \mbox{ for } x > 0.
\end{align*}
In this case, $-\theta_M v_{\beta}(x) \le \theta_M K$. Substituting this into Equation~(\ref{eqn:vbetaneginf2}) gives 
\begin{align*}
\frac{1}{2} \sigma^2 v_{\beta}^{\prime}(x) \le
   [\beta + \tilde{\phi} (p - v_{\beta}(x^*)) + \theta_M  p + \theta_M K]
   -x k(p-v_{\beta}(x^*)).
\end{align*}
Integrating both sides and using $v_{\beta}(0)=0$ give
\begin{align*}
\frac{1}{2} \sigma^2 v_{\beta}(x) \le
   x[ \beta + \tilde{\phi} (p - v_{\beta}(x^*)) + \theta_M  p +\theta_M K]
   -\frac{x^2}{2}\kappa (p-v_{\beta}(x^*)),
\end{align*}
where the right-hand side tends to $-\infty$ as $x \rightarrow \infty$, implying $v_{\beta}(x) \rightarrow - \infty$, contradicting $v_{\beta}(x)$ is bounded from below.
\hfill $\blacksquare$


\paragraph{Proof of Lemma~\ref{lem:IDnotempty}.}
We first prove that $\mathcal{I} \not= \emptyset$. Note from Lemmas~\ref{lem:IDpartition} and~\ref{lem:vbetaneginf} that it suffices to show that there exists $\beta > \underline{\beta}$ such that $\lim_{x \rightarrow \infty} v_{\beta}(x) = \infty$. To this end, recall from IVP($\beta$) that
\begin{align*}
\frac{1}{2} \sigma^2 v_{\beta}^{\prime}(x) &=
   \beta + \phi(p - v_{\beta}(x)) - x \kappa (p-v_{\beta}(x)) \\
& \ge \beta - \underline{\beta} - x\kappa p + \kappa x  v_{\beta}(x),
\end{align*}
where the inequality holds because $\underline{\beta} = -\inf \phi(y)$. By rearranging the terms, we write
\begin{align*}
v_{\beta}^{\prime}(x) - \frac{2\kappa }{\sigma^2} x v_{\beta}(x) \ge \frac{2}{\sigma^2} ( \beta - \underline{\beta})
   - \frac{2 \kappa p}{\sigma^2} x.
\end{align*}
Multiplying both sides with the integrating factor $\exp \left\{- \frac{\kappa }{\sigma^2} x^2 \right\}$ gives:
\begin{align*}
\left[ \exp \left\{ -\frac{\kappa }{\sigma^2} x^2 \right\} v_{\beta}(x) \right]^{\prime}
   \ge \frac{2}{\sigma^2} (\beta - \underline{\beta}) \exp \left\{ -\frac{\kappa }{\sigma^2} x^2 \right\} 
   - \frac{2 \kappa p}{\sigma^2}x \exp \left\{ -\frac{\kappa }{\sigma^2} x^2 \right\}
\end{align*}
Integrating both sides on $[0, x]$ and using the boundary condition $v_{\beta}(0) = 0$ give
\begin{align*}
\exp \left\{ -\frac{\kappa }{\sigma^2} x^2 \right\} v_{\beta}(x) 
   &\ge \frac{2}{\sigma^2} (\beta - \underline{\beta}) \int_0^x\exp \left\{ -\frac{\kappa }{\sigma^2} s^2 \right\} ds
   + p \int_0^x \left( -\frac{2\kappa }{\sigma^2}s \right)
   \exp \left\{ -\frac{\kappa }{\sigma^2} s^2 \right\} ds \\
&\ge \frac{2}{\sigma^2} (\beta - \underline{\beta})  \int_0^x \exp \left\{ -\frac{\kappa }{\sigma^2} s^2 \right\}  ds
   + p \exp \left\{ -\frac{\kappa }{\sigma^2} s^2 \right\} \Bigg|_0^x \\
&\ge \frac{2}{\sigma^2} (\beta - \underline{\beta})  \int_0^x \exp \left\{ -\frac{\kappa }{\sigma^2} s^2 \right\}  ds
   + p \exp \left\{ -\frac{\kappa }{\sigma^2} x^2 \right\} - p.
\end{align*}
Multiplying both sides with $\exp \left\{ \frac{\kappa }{\sigma^2} x^2 \right\}$ gives
\begin{align}
v_{\beta}(x) \ge \left( \frac{2}{\sigma^2} (\beta - \underline{\beta}) 
   \int_0^x \exp \left\{ -\frac{\kappa }{\sigma^2} s^2 \right\} ds - p \right) 
   \exp \left\{ \frac{\kappa  x^2}{\sigma^2} \right\} + p.
\end{align}

Let $\alpha = \int_0^1 \exp \left\{ -\frac{\kappa }{\sigma^2} s^2 \right\} ds > 0$, and note that for $x> 1$, 
\begin{align}
v_{\beta}(x) \ge \left( \frac{2 \alpha}{\sigma^2} (\beta - \underline{\beta}) 
   - p \right) \exp \left\{ \frac{\kappa  x^2}{\sigma^2} \right\}  + p. \label{eqn:vbge}
\end{align}
Note that for $\beta > \underline{\beta} + \frac{\sigma^2}{2 \alpha}p$, we have from Equation~(\ref{eqn:vbge}) that $v_{\beta}(x) \rightarrow \infty$ as $x \rightarrow \infty$, proving that $\beta \in \mathcal{I}$. Thus, $\mathcal{I} \not= \emptyset$.

Next, we prove that $\mathcal{D} \not= \emptyset$. Suppose not, i.e., $\mathcal{D} = \emptyset$. Let $y^* = \arg \inf \{ \phi(y) : y \le p \}$ and note that $y^* \in (0, p]$. Also, let 
\begin{align*}
\hat{x} = \inf \{ x > 0 : p - v_{\beta}(x) \le y^* \}.
\end{align*}
Recall IVP($\beta$) (see Equations~(\ref{eqn:ivp1})-(\ref{eqn:ivp2})) and note that $v_{\beta}^{\prime}(0) > 0$ for $\beta > \underline{\beta}$:
\begin{align}
\frac{1}{2} \sigma^2 v_{\beta}^{\prime}(x) = \beta + \phi(p - v_{\beta}(x))
   -x \kappa (p - v_{\beta}(x)), \,\,\, x>0, \label{eqn:A20}
\end{align}
and consider the following three cases:

\noindent \emph{Case 1:} $\hat{x} = \infty$. In this case, we necessarily have $y^* < p$ and $p - v_{\beta}(x) \ge y^*$ for all $x > 0$. Then by convexity of $\phi$ and that $y^*$ is the minimizer, we deduce that
$\phi(p - v_{\beta}(x)) \le \phi(p)$ for $x > 0$.
Using this, (\ref{eqn:A20}), and that $p - v_{\beta}(x) \ge y^*$ for all $x > 0$, we conclude that
\begin{align*}
\frac{1}{2} \sigma^2 v_{\beta}^{\prime}(x) \le \beta + \phi(p)- x \kappa  y^*
\end{align*}
from which we see that $v_{\beta}^{\prime}(x) < 0$ for $x$ sufficiently large. Then, Lemma~\ref{lem:betainD} implies that $\beta \in \mathcal{D}$, a contradiction.

\noindent \emph{Case 2:} $\hat{x} = 0$. In this case, $y^* = p$, $\theta_M < 0$. Also note that $\phi(p) = - \underline{\beta}$. Note that for $\varepsilon > 0$ sufficiently small, we have that
\begin{align*}
\phi(p - v_{\beta}(x)) = \phi(p) - \theta_M v_{\beta}(x), \,\,\, x \in (0, \varepsilon).
\end{align*}
Substituting this into Equation~(\ref{eqn:A20}) gives: For $x \in (0, \varepsilon)$
\begin{align*}
\frac{1}{2} \sigma^2 v_{\beta}^{\prime}(x) = \beta + \phi(p) + ( | \theta_M| + \kappa x) v_{\beta}(x) -x \kappa p.
\end{align*}
Rearranging the terms and substituting $\phi(p) = -\underline{\beta}$ yields:
\begin{align*}
v_{\beta}^{\prime}(x) - \left( \frac{2 |\theta_M|}{\sigma^2} + \frac{2\kappa }{\sigma^2}x \right) v_{\beta}(x)
   = \frac{2}{\sigma^2} (\beta - \underline{\beta}) - \frac{2 \kappa p}{\sigma^2} x.
\end{align*}
Multiplying both sides by the integrating factor $g(x) = \exp \left\{ -\frac{\kappa }{\sigma^2} x^2 - \frac{2 | \theta_M|}{\sigma^2} x \right\}$ yields the following:
\begin{align*}
[g(x) v_{\beta}(x)]^{\prime} = \frac{2}{\sigma^2} (\beta - \underline{\beta}) g(x) - \frac{2 \kappa p}{\sigma^2}x g(x).
\end{align*}
Integrating both sides gives and using the boundary condition $v_{\beta}(0)=0$ give
\begin{align*}
v_{\beta}(x) = \frac{1}{g(x)} \int_0^x \frac{2}{\sigma^2} (\beta - \underline{\beta} g(s) ds
   -\frac{2 \kappa p}{\sigma^2 g(x)} \int_0^x s g(s) ds, \,\,\, x \in (0, \varepsilon).
\end{align*}
Clearly, it follows that $\lim_{\beta \downarrow \underline{\beta}} v_{\beta}(x) < 0$ for $x \in (0, \varepsilon)$.  Thus, there exists $\tilde{\beta} > \underline{\beta}$ sufficiently small and $x \in (0, \varepsilon)$ such that $v_{\beta}(x) < 0$. Thus, $\tilde{\beta} \in \mathcal{D}$ by Lemma~\ref{lem:betainD}, a contradiction.

\noindent \emph{Case 3:} $\hat{x} \in (0, \infty)$. In this case, $y^* < p$ (otherwise, i.e., if $y^* = p$, then $\hat{x} =0$). Also, note that $\theta_M > 0$. Note that
\begin{align}
p - v_{\beta}(x) \ge y^*, \,\,\, x \in (0, \hat{x}). \label{eqn:A21}
\end{align}
Then by convexity of $\phi$ and that $y^*$ minimizes $\phi$, we conclude that
\begin{align}
\phi(p - v_{\beta}(x)) \le \phi(p), \,\,\, x \in (0, \hat{x}). \label{eqn:A22}
\end{align}
Substituting~(\ref{eqn:A21})-(\ref{eqn:A22}) into IVP($\beta$), we get
\begin{align*}
\frac{1}{2} \sigma^2 v_{\beta}^{\prime}(x) \le \beta + \phi(p) - x\kappa y^* \le \beta + \phi(p),
   \,\,\, x \in (0, \hat{x}).
\end{align*}
Integrating this from 0 to $\hat{x}$ yields
\begin{align*}
\frac{1}{2} \sigma^2 v_{\beta}^{\prime}(x) \le \frac{1}{2} \sigma^2 (p - y^*)
   \le (\beta + \phi(p)) \hat{x},
\end{align*}
where the first inequality follows from~(\ref{eqn:A21}). In particular, we have that
\begin{align}
\hat{x}(\beta) \ge \frac{\sigma^2 (p - y^*)}{2( \beta + \phi(p))}. \label{eqn:A23}
\end{align}

For $\beta < \underline{\beta} + 1$, we write
\begin{align*}
\hat{x}(\beta) \ge \frac{\sigma^2 (p - y^*)}{2( \underline{\beta} + 1 + \phi(p))}. 
\end{align*}
Let $\beta = \underline{\beta} + \varepsilon$ for $\varepsilon \in (0, 1)$ sufficiently small and consider the IVP($\beta$) at $x = \hat{x}(\beta)$:
\begin{align*}
\frac{1}{2} \sigma^2 v_{\beta}^{\prime}(\hat{x}(p) ) =  \underline{\beta} + \varepsilon
   +\phi(y^*) - \hat{x} \kappa y^*.
\end{align*}
Substituting $\phi(y^*) = \underline{\beta}$ and (\ref{eqn:A23}) yields
\begin{align*}
\frac{1}{2} \sigma^2 v_{\beta}^{\prime}(\hat{x}(p) ) \le \varepsilon
   -\frac{\sigma^2(p - y^*)}{\underline{\beta} + 1 + \phi(p)} \kappa  y^*.
\end{align*}
Thus, for $\varepsilon < \min \left\{ 1, \frac{\sigma^2}{4} \frac{p - y^*}{\underline{\beta} + 1 + \phi(p)} \kappa y^* \right\}$, we have that $v_{\beta}^{\prime}( \hat{x} (\beta)) < 0$, and Lemma~\ref{lem:betainD} implies $\beta = \underline{\beta} + \varepsilon \in \mathcal{D}$, a contradiction.
\hfill $\blacksquare$

\paragraph{Proof of Lemma~\ref{lem:betaI}.}
Let $\beta > \underline{\beta}$. First, assume $\beta \in \mathcal{I}$ and $v_{\beta}$ is unbounded. Then, we trivially have that there exists $x_0$ such that $v_{\beta}(x_0) \ge p$. To prove the converse, assume there exists $x_0 > 0$ such that $v_{\beta} (x_0) \ge p$. Without loss of generality, we assume
\begin{align*}
x_0 = \inf \left\{ x > 0 : v_{\beta}(x) = p \right\}.
\end{align*}
Next, we prove that $v_{\beta}(x) \rightarrow \infty$ as $x \rightarrow \infty$. Then because $v_{\beta}$ increases strictly to its supremum by Lemma~\ref{lem:vinc}, we conclude that $\beta \in \mathcal{I}$. To this end, let
\begin{align*}
x_1 = \inf \left\{ x > x_0 : v_{\beta}(x) \le p \right\}.
\end{align*}
We argue that $x_1 = x_0$. Suppose not, i.e., $x_1 > x_0$. By continuity of $v_{\beta}$, we have $v_{\beta}(x_1) = p$. We also note that
\begin{align}
v_{\beta}(x) > p \,\,\, \mbox{for} \,\,\, x \in [x_0, x_1]. \label{eqn:A24}
\end{align}
Because $v_{\beta}(x_0) = v_{\beta}(x_1)$, the Mean Value theorem implies that there exists $\bar{x} \in ( x_0, x_1)$ such that $v_{\beta}^{\prime}(\bar{x})=0$. Consider IVP($\beta$) at $x=\bar{x}$:
\begin{align}
0 = \beta + \phi(p - v_{\beta}(\bar{x})) + \bar{x}\kappa  \left[ v_{\beta}(\bar{x}) 
   - p \right]. \label{eqn:A25}
\end{align}
Because $v_{\beta}(\bar{x}) \ge p$ by~(\ref{eqn:A24}), we conclude that
\begin{align*}
\phi(p - v_{\beta}(\bar{x})) = \theta_0 (p - v_{\beta}(\bar{x})).
\end{align*}

Thus, we have
\begin{align*}
\beta + \theta_0(p - v_{\beta}(\bar{x}) + \bar{x} \kappa \left[ v_{\beta}(\bar{x}) - p \right] >0,
\end{align*}
which contradicts~(\ref{eqn:A25}). Thus $x_1 = x_0$. Consequently, we deduce that
\begin{align}
v_{\beta}(x) > p \,\,\, \mbox{for} \,\,\, x > x_0. \label{eqn:A26}
\end{align}
In particular, we have that
\begin{align}
\phi(p - v_{\beta}(x)) = \theta_0(p - v_{\beta}(x))  \,\,\, \mbox{for} \,\,\, x > x_0. \label{eqn:A27}
\end{align}
Substituting (\ref{eqn:A26})-(\ref{eqn:A27}) into the IVP($\beta$), we write
\begin{align*}
\frac{1}{2} \sigma^2 v_{\beta}^{\prime} (x) 
   &= \beta + \phi(p - v_{\beta}(x))  -x\kappa  \left( p - v_{\beta}(x) \right) \\
&= \beta + | \theta_0| v_{\beta}(x)  - | \theta_0| p + x\kappa  \left( v_{\beta}(x) -p \right),
   \,\,\, x > x_0 \\
&\ge \beta > 0 \,\,\, \mbox{for} \,\,\, x> x_0.
\end{align*}
Integrating both sides (over $(x_0, x)$) gives
\begin{align}
v_{\beta}(x) &\ge \frac{\sigma^2}{2} \left[ \beta 
   (x - x_0) + v_{\beta}(x_0) \right], \,\,\, x > x_0 \\
&= \frac{\sigma^2}{2} \left[ \beta 
   (x - x_0) + p \right]
\end{align}
Hence, we conclude $v_{\beta}(x) \rightarrow \infty$ as $x \rightarrow \infty$, because the right-hand side tends to $+\infty$ as $x \rightarrow \infty$.
\hfill $\blacksquare$

\section{Further Details of the Numerical Study} \label{app:numerical-details}

\setcounter{equation}{0}
\renewcommand{\theequation}{\thesection\arabic{equation}}

\subsection{Volunteers} \label{sub:volunteer-details}

The following shows the calculations for the arrival rate [arrivals/year] of class $j$ volunteers in the $n^{th}$ system, where $n=56{,}000$:
   
\begin{itemize}
   \item \textbf{Corporate volunteers ($j=1$, repeat volunteers)}: There are two types of corporate volunteers: regulars who volunteer on average 3 times every two years (28\%) and sporadic volunteers who volunteer on average once every 4 years (72\%). Therefore, for class $j=1$, the average arrival rate for a volunteer is $r_1^n = 0.28 \,\,(1.5 \, times/year/volunteer) + 0.72 \,\,(0.25 \, times/year/volunteer) = 0.60 \, times/year/volunteer$ and the overall average arrival rate for the class is $r_1^n k_1^n = 0.60 \times 11{,}200= 6{,}720$ arrivals/year.

   \item \textbf{Individual volunteers ($j=2$, repeat volunteers)}  consist of 4 types of volunteers. In decreasing order of volunteering frequency: 
      0.5\% volunteer 52 times per year,
      17.5\% volunteer 12 times per year,
      32\% volunteer 2 times per year, and
      50\% volunteer once every 4 years.
      Therefore, for class $j=2$, the average arrival rate for a volunteer is $r_2^n = 0.005 \,\, (52 \, times/year/volunteer) + 0.175 \,\, (12 \, times/year/volunteer) + 0.32 \,\, (2 \, times/year/volunteer) + 0.5 \,\, (0.25 \, times/year/volunteer) = 3.125 \, times/year/volunteer$ and the overall average arrival rate for the class is $r_2^n k_2^n = 3.125 \times 16{,}800= 52{,}500$ arrivals/year.
      
   \item \textbf{Social group volunteers ($j=3$, one-time volunteers)}: Food Bank~A estimates that approximately 25\% of volunteer arrivals are social group volunteers. Therefore, we estimate that the arrival rate of class $j=3$ is $\lambda_3^n = 19{,}540$ arrivals/year.
         
   \item \textbf{Overall arrival rate:} Combining the arrival rates of all classes, the total annual arrival rate is $r_1^n k_1^n + r_2^n k_2^n + \lambda_3^n = 78{,}760$ arrivals/year. 
   \end{itemize}

\subsection{Engagement Activities} \label{sub:engage-activities}

The food bank can engage in $L=4$ different volunteer engagement activities.

\begin{itemize}
      \item \textbf{Orientation during volunteer shift ($l=1$)}:  A brief orientation at the beginning and closing remarks at the end of each volunteer shift explains the social impact that the volunteer work has on the population in need. This activity affects the arrival rate of volunteer classes $\mathcal{R}_{1} \cup \mathcal{S}_{1} = \{ 1,2,3\}$. Every time the food bank engages in this activity, on average it increases volunteer arrivals by 2 volunteers. Typical frequency of this activity is once per volunteer shift (i.e., 312 times per year). If the food bank engages in this activity at the beginning of each volunteer shift, the estimated impact is an increase of $2 \,vol/activity \times 312 \,activities/year = 624$ class $j=1,2,3$ volunteers per year, which is a 0.7923\% (=624/78{,}760) increase in total yearly arrivals. Assuming a proportional increase in arrivals, this gives $\hat{r}_{11}^n= 0.7923\% \times r_{1}^n = 0.00475$ additional arrivals/year/volunteer (for a total of $0.00475 \times 11{,}200 = 53$ additional arrivals per year), $\hat{r}_{21}^n= 0.7923\% \times r_{2}^n = 0.02476$ additional arrivals/year/volunteer (for a total of $0.02476 \times 16{,}800 = 416$ additional arrivals per year), and $\hat{\lambda}_{31}^n = 0.7923\% \times \lambda_3^n =155$ additional arrivals/year.
Typically, a junior staff would do the orientation but occasionally, a seasoned volunteer may be able to perform this activity. The orientation activity takes approximately 10 minutes: $\$20/hour \times 10/60 \sim \$3$. Therefore, the annual cost of this activity is \$936. This implies that $F_1^n = \$3 * 312 = \$936$/year. There is no per volunteer cost for the activity ($C_1^n = 0$).
   
      \item \textbf{Electronic communication ($l=2$)}: The food bank sends targeted emails and other forms of electronic communication to volunteers to notify them of volunteer activities. This activity affects the arrival rate of volunteer classes $\mathcal{R}_{2} \cup \mathcal{S}_{2} = \{ 1,2,3\}$. Every time the food bank engages in this activity, on average it increases volunteer arrivals by 15 volunteers. Typical frequency of this activity is once per week (i.e., 52 times per year). If the food bank engages in this activity weekly, the estimated impact is an increase of $15 \,vol/activity \times 52 \,activities/year =780$ class $j=1,2,3$ volunteers per year, which is a 0.9904\% (=780/78{,}760) increase in total yearly arrivals. Assuming a proportional increase in arrivals, this gives $\hat{r}_{12}^n= 0.9904\% \times r_{1}^n = 0.00594$ additional arrivals/year/volunteer (for a total of $0.00594 \times 11{,}200 = 66$ additional arrivals per year), $\hat{r}_{22}^n= 0.9904\% \times r_{2}^n = 0.03095$ additional arrivals/year/volunteer (for a total of $0.03095 \times 16{,}800 = 520$ additional arrivals per year), and $\hat{\lambda}_{32}^n = 0.9904\% \times \lambda_3^n = 194$ additional arrivals/year.
Typically, a senior staff member would compose the e-communication (e.g., a newsletter), taking approximately one hour. It costs approximately \$35 each time the food bank runs activity $l=2$, therefore, $F_2^n = \$35*52 = \$1820$/year. There is no per volunteer cost for the activity ($C_2^n = 0$).

      

      \item \textbf{Speaking engagement ($l=3$)}: Food bank staff can make presentations at organizations in the region to raise awareness. This activity affects the arrival rate of volunteer classes $\mathcal{S}_{3} = \{ 3\}$. Every time the food bank engages in this activity, on average it increases volunteer arrivals by 20 volunteers. Typical frequency of this activity is once per month (i.e., 12 times per year).  If the food bank engages in this activity once per month, the estimated impact is an increase of $\hat{\lambda}_{33}^n = 20 \,vol/activity \times 12 \,activities/year =240$ class $j=3$ volunteer arrivals per year.
For employee time (including travel), it costs approximately \$60 each time the food bank runs activity $l=3$, therefore, $F_3^n = \$60*12 = \$720$/year. (employees spend about one hour on travel and another 45-60 min for presentation, hourly rate \$35/hour). There is no per volunteer cost for the activity ($C_3^n = 0$).
   
      \item \textbf{Tabling at a fair ($l=4$)}: 
Food bank staff can set up an information table at fairs throughout the year. This activity affects the arrival rate of volunteer classes $\mathcal{R}_{4} = \{ 1,3\}$. Every time the food bank engages in this activity, on average it increases volunteer arrivals by 30 volunteers. Typical frequency of this activity is once per month (i.e., 12 times per year).  If the food bank engages in this activity once per month, the estimated impact is an increase of $30 \,vol/activity \times 12 \,activities/year =360$ class $j=1,3$ volunteers. 
Over classes $j=1,3$, this is a $1.3709\% (=360/(6720+19{,}540))$ increase in total yearly arrivals. Assuming a proportional increase in arrivals, this gives $\hat{r}_{14}^n= 1.3709\% \times r_{1}^n = 0.00823$ additional arrivals/year/volunteer (for a total of $0.00823 \times 11{,}200 = 92$ additional arrivals per year) and $\hat{\lambda}_{34}^n = 1.3709\% \times \lambda_3^n = 268$ additional arrivals per year.
For employee time (including travel), it costs approximately \$150 each time the food bank runs activity $l=4$, therefore, $F_4^n = \$150*12 = \$1800$/year. There is no per volunteer cost for the activity ($C_4^n = 0$).

      \end{itemize}

\subsection{Other parameters in the $n^{th}$ system} \label{sub:other-params}

\begin{itemize}

  \item $p$: Cost of throughput loss. If a volunteer slot is not filled, the lost value is the forgone value of the meals the volunteer would have made. Data from Food Bank~A indicates that a volunteer makes approximately 135 meals in a shift. The cost per meal is \$0.35 per meal. Therefore, we estimate that the lost value for a unit of idleness (i.e., unfilled volunteer slot) is $p = 135$ meals $\times$ \$0.35/meal $= \$47.25 \approx \$50$.


  
  \item $\gamma_j$: The abandonment rate of class $j$ volunteers. Food Bank~A indicated that the percentage of abandonment was very low. We perform sensitivity analysis using a range of values for $\gamma_j$.

\end{itemize}

\subsection{Parameters in the limit system} \label{sub:limit-params}

\begin{itemize}
	
   \item $\mu_j$: From (\ref{eqn:lambdajn}) and $\mu_1^n =78{,}000$, we have  $\mu_1 = \mu_2 = \mu_3 = \mu_1^n/n = 78{,}000/56{,}000 = 1.393$ volunteers/year.
   
   \item $\alpha_j$: $\alpha_1 = 1.37010$, $\alpha_2 = 7.13595$ and $\alpha_3 = 1.59356$. 
   
   \item $r_j$ and $\lambda_j$: From Equations~(\ref{eqn:rjn}) and (\ref{eqn:lambdajn}), $r_1= 0.60579$ ($r_1 \hat{k}_1 = 0.12116$), $r_2= 3.15515$ ($r_2 \hat{k}_2=0.94655$), and $\lambda_3= 0.35230$.

   \item $\hat{r}_{jl}$, $\hat{\lambda}_{jl}$, $\eta_{l}$: 
   
   \begin{itemize}
   
      \item \textbf{Orientation ($l=1$)}: 
      From Equation~(\ref{eqn:etajln-repeat})and~(\ref{eqn:etajln-onetime}), we have:
      \begin{align*}
         \hat{r}_{11} &=\sqrt{n} \, \hat{r}_{11}^n = \sqrt{56{,}000} \times 0.00475,\\
         \hat{r}_{21} &=\sqrt{n} \, \hat{r}_{21}^n = \sqrt{56{,}000} \times 0.02476, \\
         \hat{\lambda}_{31} &= \frac{\hat{\lambda}_{31}^n }{ \sqrt{n}} = \frac{155}{ \sqrt{56{,}000} } .
      \end{align*}
      From Equation~(\ref{eqn:preserve-eta}), we have:
      \begin{align*}
         \eta_1 = \frac{\hat{k}_1 \hat{r}_{11}}{\mu_1}  + \frac{\hat{k}_2 \hat{r}_{21}}{\mu_2}  + \frac{\hat{\lambda}_{31}}{\mu_3}
         = \frac{\frac{ 11{,}200 \times  0.00475}{\sqrt{56{,}000}} }{1.393} + \frac{\frac{ 16{,}800 \times 0.02476}{\sqrt{56{,}000}}}{1.393}
            + \frac{\frac{155}{ \sqrt{56{,}000} }}{1.393} = 1.89315.
      \end{align*}

      \item \textbf{E-communication ($l=2$)}: 
      From Equation~(\ref{eqn:etajln-repeat}) and~(\ref{eqn:etajln-onetime}), we have:
      \begin{align*}
         \hat{r}_{12} &=\sqrt{n} \, \hat{r}_{12}^n = \sqrt{56{,}000} \times 0.00594,\\
         \hat{r}_{22} &=\sqrt{n} \, \hat{r}_{22}^n = \sqrt{56{,}000} \times 0.03095, \\
         \hat{\lambda}_{32} &= \frac{\hat{\lambda}_{32}^n }{ \sqrt{n}} = \frac{194}{ \sqrt{56{,}000} }.
      \end{align*}
      From Equation~(\ref{eqn:preserve-eta}), we have:
      \begin{align*}
         \eta_2 = \frac{\hat{k}_1 \hat{r}_{12}}{\mu_1}  + \frac{\hat{k}_2 \hat{r}_{22}}{\mu_2}  + \frac{\hat{\lambda}_{32}}{\mu_3}
         = \frac{\frac{11{,}200 \times 0.00594}{\sqrt{56{,}000}} }{1.393} + \frac{\frac{16{,}800\times 0.03095}{\sqrt{56{,}000}}}{1.393}
            + \frac{\frac{194}{ \sqrt{56{,}000}}}{1.393} = 2.36643.
      \end{align*}
         
      \item \textbf{Speaking ($l=3$)}: 
      From Equation~(\ref{eqn:etajln-onetime}), we have 
      \begin{align*}
         \hat{\lambda}_{33} =  \frac{\hat{\lambda}_{33}^n }{ \sqrt{n}} = \frac{240}{ \sqrt{56{,}000} }.
      \end{align*}
      From Equation~(\ref{eqn:preserve-eta}), we have:
      \begin{align*}
         \eta_3 = \frac{\hat{\lambda}_{32}}{\mu_3}
         = \frac{\frac{240}{ \sqrt{56{,}000} }}{1.393} = 0.72813.
      \end{align*}
          
      \item \textbf{Tabling ($l=4$)}: 
      From Equation~(\ref{eqn:etajln-repeat}) and~(\ref{eqn:etajln-onetime}), we have:
      \begin{align*}
         \hat{r}_{14} &=\sqrt{n} \, \hat{r}_{14}^n = \sqrt{56{,}000} \times 0.00823,\\
				 \hat{\lambda}_{34} &= \frac{\hat{\lambda}_{32}^n }{ \sqrt{n}} = \frac{268}{ \sqrt{56{,}000} }.
      \end{align*}
      From Equation~(\ref{eqn:preserve-eta}), we have:
      \begin{align*}
         \eta_2 = \frac{\hat{k}_1 \hat{r}_{14}}{\mu_1}  + \frac{\hat{k}_2 \hat{r}_{24}}{\mu_2} 
         = \frac{\frac{11{,}200 \times 0.00823}{\sqrt{56{,}000}} }{1.393} + \frac{\frac{268}{\sqrt{56{,}000}} }{1.393}
         = 1.09220.
      \end{align*}

   \end{itemize}

   \item Fixed cost, $F_l$: From Equation~(\ref{eqn:costF}), we have
   
      \begin{itemize}
         
         \item \textbf{Orientation ($l=1$)}: $F_1 = F_1^n/\sqrt{n} = 936/\sqrt{56{,}000} = \$3.96$/year.
         
         \item \textbf{E-communication ($l=2$)}: $F_2 = F_2^n/\sqrt{n} = 1820/\sqrt{56{,}000} = \$7.69$/year.

         \item \textbf{Speaking ($l=3$)}: $F_3 = F_3^n/\sqrt{n} = 720/\sqrt{56{,}000} = \$3.04$/year.
         
         \item \textbf{Tabling ($l=4$)}: $F_4 = F_4^n/\sqrt{n} = 1800/\sqrt{56{,}000} = \$7.61$/year.
             
      \end{itemize}

   \item $p = p^n/n$: Penalty (cost) rate of throughput loss. We assume that $\hat{c}_1 < \hat{c}_2 < \dots < \hat{c}_L < p$ (see (\ref{eqn:cjmlm})).
      
   \item From~(\ref{eqn:handk}), $\kappa = \sum_{j=1}^J (r_j + \gamma_j ) x_j + \sum_{j = J+1}^{J+\tilde{J}} \gamma_j x_j = 2.12367$ for $\gamma_j=0.01$. 
   
   \item Variance of Brownian motion $X(t)$: $\sigma_w^2 = \sum_{j=1}^{J} 2 r_j \hat{k}_j/\mu_j^2 + \sum_{j=J+1}^{J+\tilde{J}} 2 \lambda_j / \mu_j^2 = 1.436$. 

\end{itemize}

\subsection{Optimal Policy Thresholds} \label{sub:policythresh}


From Proposition~\ref{prop:valuefunc}, for $x \in [\tau_l, \tau_{l-1})$ and $l=1, \ldots, L+1$, we have $\tau_l = v^{-1}(p-\hat{c}_l)$ and
\begin{align}
 v(x) &=  p  - \hat{c}_l \exp \left\{ \frac{- 2 \theta_{l-1} (x-\tau_l) + \kappa (x^2-\tau_l^2)}{\sigma^2} \right\}  \nonumber \\ 
& +  \frac{2(\beta  - c(\theta_{l-1}))\sqrt{\pi}\exp \left\{ \frac{(\kappa x - \theta_{l-1})^2}{\kappa\sigma^2} \right\}}{\sigma \sqrt{\kappa}}   
 \left[ \Phi \left( \frac{\kappa x - \theta_{l-1}}{\sigma \sqrt{\kappa/2}} \right) - \Phi \left( \frac{\kappa \tau_l- \theta_{l-1}}{\sigma \sqrt{\kappa/2}} \right) \right]. \label{eqn:v-ul}
\end{align}
Also recall that we defined $\hat{c}_{L+1}=p$, $\tau_{L+1} = 0$, and $\tau_0 = \infty$. Figure~\ref{fig:bellmansoln} shows the solution to the Bellman equation (and relationship between $\tau_l$, $p$, and $\hat{c}_l$).

\begin{figure}[htbp]
\centering
\includegraphics[scale=0.4]{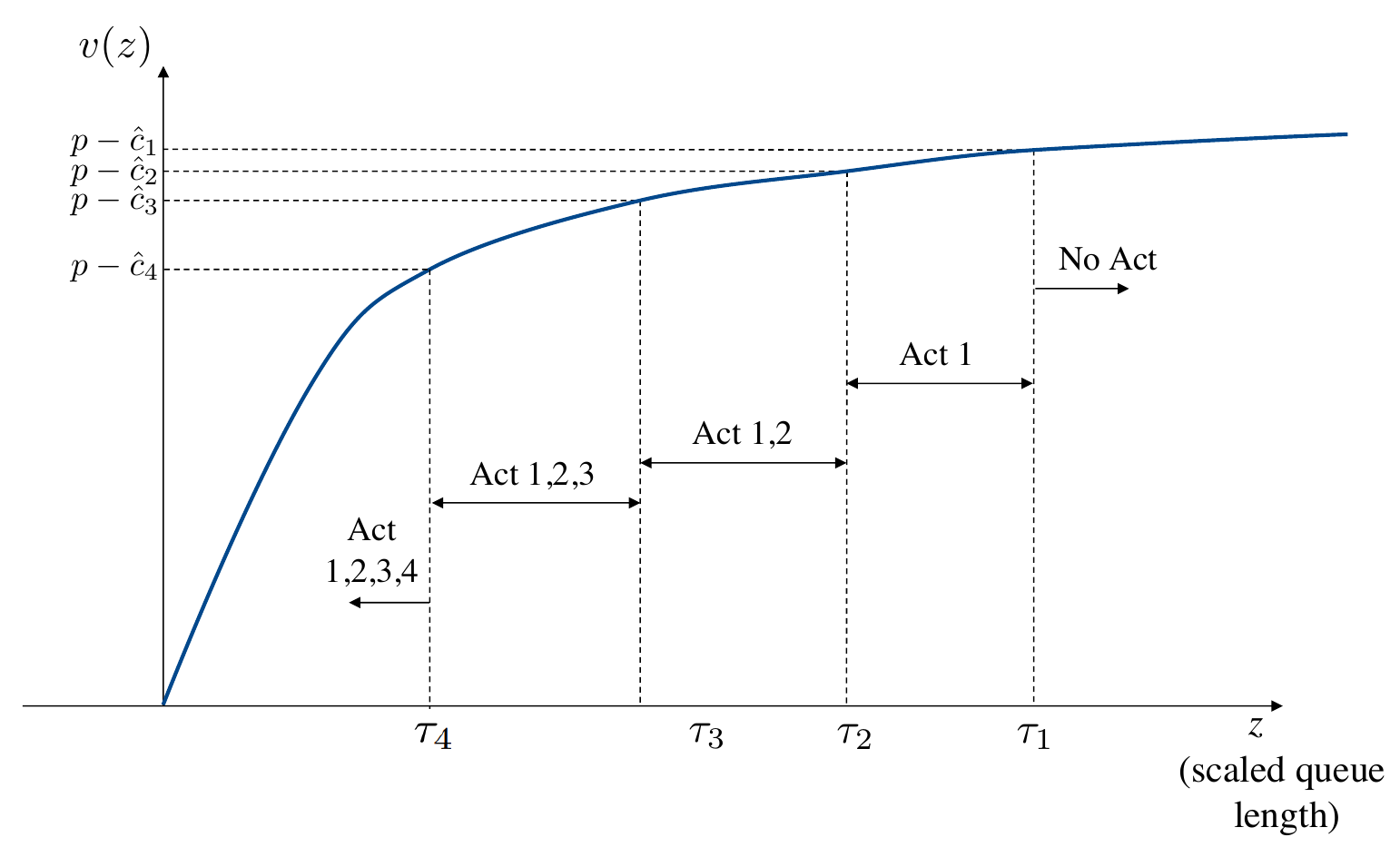}
\caption{The solution to the Bellman equation. }
\label{fig:bellmansoln}
\end{figure}

In the numerical study (Section~\ref{sec:numerical}) there are $L=4$ activities, $\tau_5 = 0$ and $\tau_4 < \ldots < \tau_{1}$. 
We begin with finding the optimal threshold $\tau_4$. Substituting $l=5$ in Equation (\ref{eqn:v-ul}) yields the following implicit equation for $x \in [0, \tau_4)$,
\begin{align}
v(x) &=  p  -  p \exp \left\{ \frac{- 2 \theta_4 (x) + \kappa (x^2)}{\sigma^2} \right\}  \nonumber \\ 
& +  \frac{2(\beta  - c(\theta_4))\sqrt{\pi}}{\sigma \sqrt{\kappa}} 
\exp \left\{ \frac{(\kappa x - \theta_4)^2}{\kappa\sigma^2} \right\} \left[ \Phi \left( \frac{\kappa x - \theta_4}{\sigma \sqrt{\kappa/2}} \right) - \Phi \left( \frac{- \theta_4}{\sigma \sqrt{\kappa/2}} \right) \right], \label{eqn:um1}
\end{align}
where $\Phi$ denotes the cumulative distribution function for a standard normal random variable. We solve for $\tau_4$ in $v(\tau_4) = p-\hat{c}_4$, see Figure~\ref{fig:bellmansoln}.

Similarly, substituting $l=4, \ldots, 1$ in Equation (\ref{eqn:v-ul}) yields Equations (\ref{eqn:um3}), (\ref{eqn:um4}), and (\ref{eqn:um5}), respectively, and using $v(\tau_{l}) = p-\hat{c}_{l}$ as the starting point in the range $[\tau_{l}, \tau_{l-1})$ we solve for $\tau_{l-1}$ in $v(\tau_{l-1}) = p-\hat{c}_{l-1}$ to find the optimal threshold $\tau_{l-1}$.

%

%
For $x \in [\tau_4, \tau_3)$, we have that
\begin{align}
v(x) &=  p  - \hat{c}_4 \exp \left\{ \frac{- 2 \theta_{3} (x-\tau_4) + \kappa (x^2-\tau_4^2)}{\sigma^2} \right\}  \nonumber \\ 
& +  \frac{2(\beta  - c(\theta_{3}))\sqrt{\pi}}{\sigma \sqrt{\kappa}}   
\exp \left\{ \frac{(\kappa x - \theta_{3})^2}{\kappa\sigma^2} \right\} \left[ \Phi \left( \frac{\kappa x - \theta_{3}}{\sigma \sqrt{\kappa/2}} \right) - \Phi \left( \frac{\kappa \tau_4- \theta_{3}}{\sigma \sqrt{\kappa/2}} \right) \right]. \label{eqn:um3}
\end{align}
Using this, we solve for $\tau_{3}$ in $v(\tau_{3}) = p-\hat{c}_{3}$.

For $x \in [\tau_3, \tau_2)$, we have that
\begin{align}
v(x) &=  p  - \hat{c}_3 \exp \left\{ \frac{- 2 \theta_2 (x-\tau_3) + \kappa (x^2-\tau_3^2)}{\sigma^2} \right\}  \nonumber \\ 
& +  \frac{2(\beta  - c(\theta_2))\sqrt{\pi}}{\sigma \sqrt{\kappa}}   
\exp \left\{ \frac{(\kappa x - \theta_2)^2}{\kappa\sigma^2} \right\} \left[ \Phi \left( \frac{\kappa x - \theta_2}{\sigma \sqrt{\kappa/2}} \right) - \Phi \left( \frac{\kappa \tau_3- \theta_2}{\sigma \sqrt{\kappa/2}} \right) \right]. \label{eqn:um4}
\end{align}
Using this, we solve for $\tau_{2}$ in $v(\tau_{2}) = p-\hat{c}_{2}$.

For $x \in [\tau_2, \tau_1)$, we have that
\begin{align}
v(x) &=  p  - \hat{c}_2 \exp \left\{ \frac{- 2 \theta_1 (x-\tau_2) + \kappa (x^2-\tau_2^2)}{\sigma^2} \right\}  \nonumber \\ 
& +  \frac{2(\beta  - c(\theta_1))\sqrt{\pi}}{\sigma \sqrt{\kappa}}   
\exp \left\{ \frac{(\kappa x - \theta_1)^2}{\kappa\sigma^2} \right\} \left[ \Phi \left( \frac{\kappa x - \theta_1}{\sigma \sqrt{\kappa/2}} \right) - \Phi \left( \frac{\kappa \tau_2- \theta_1}{\sigma \sqrt{\kappa/2}} \right) \right]. \label{eqn:um5}
\end{align}
Using this, we solve for $\tau_{1}$ in $v(\tau_{1}) = p-\hat{c}_{1}$.

For completeness, we note that for $x \geq \tau_1$, we have that
\begin{align}
v(x) &=  p  - \hat{c}_1 \exp \left\{ \frac{- 2 \theta_0 (x-\tau_1) + \kappa (x^2-\tau_1^2)}{\sigma^2} \right\}  \nonumber \\ 
& +  \frac{2(\beta)\sqrt{\pi}}{\sigma \sqrt{\kappa}} \exp \left\{ \frac{(\kappa x - \theta_0)^2}{\kappa\sigma^2} \right\} \left[ \Phi \left( \frac{\kappa x - \theta_0}{\sigma \sqrt{\kappa/2}} \right) - \Phi \left( \frac{\kappa \tau_1- \theta_0}{\sigma \sqrt{\kappa/2}} \right) \right]. \nonumber
\end{align}

\subsection{Simulation results details} \label{sub:simresultsdetails}

Details of the simulation results are given in Tables~\ref{tab:simdetailsbase}, \ref{tab:simdetailsrepeat}, and \ref{tab:simdetailsonetime}. Tables~\ref{tab:simdetailsrepeat} and~\ref{tab:simdetailsonetime} show results for simulations where all volunteer classes are repeat volunteers and one-time volunteers, respectively. Table~\ref{tab:simdetailsbase} contains more rows (more values of $\gamma_j$) because it is our base case, with both repeat and one-time volunteer classes.

\begin{landscape}
\begin{table} [tbp]
\centering
\begin{tabular}{ccccccccccc}
Policy  &  $\gamma_j$   &  Act. cost & Idle cost & Total cost & Idle  & Abandon  & Act1 usage &
   Act2 usage & Act3 usage & Act4 usage \\ \hline
Static (1,2)& 0.005 &  \$2756 & \$82 & \$2838 & 0.0021\% & 0.8543\% & 100\% & 100\% & 0\% &0\% \\ 
Dynamic & 0.005 &  \$1767 & \$203 & \$1970 & 0.0052\% & 0.6077\% & 95\% & 36\% & 16\% &6\% \\
Static (1,2,3)& 0.01 &  \$3476 & \$252 & \$3728 & 0.0065\% & 1.3805\% & 100\% & 100\% & 100\% &0\% \\
Dynamic& 0.01 &  \$2495 & \$629 & \$3123 & 0.0161\% & 1.0668\% & 99\% & 60\% & 31\% &14\% \\

Static (1,2,3)& 0.015 &  \$3476 & \$1595 & \$5071 & 0.0409\% & 1.6193\% & 100\% & 100\% & 100\% & 0\% \\
Dynamic& 0.015 &  \$3854 & \$782 & \$4636 & 0.0201\% & 1.6312\% & 100\% & 96\% & 69\% & 37\% \\ 

Static (1,2,3,4)& 0.02 &  \$5276 & \$1352 & \$6628 & 0.0347\% & 2.0865\% & 100\% & 100\% & 100\% &100\% \\
Dynamic& 0.02 &  \$4628 & \$1748 & \$6376 & 0.0448\% & 1.9667\% & 100\% & 100\% & 93\% &67\% \\ 

Static (1,2,3,4)& 0.025 &  \$5276 & \$3593 & \$8869 & 0.0921\% & 2.2403\% & 100\% & 100\% & 100\% & 100\% \\
Dynamic& 0.025 &  \$5023 & \$3822 & \$8854 & 0.0980\% & 2.1949\% & 100\% & 100\% & 99\% & 86\% \\ \hline
\end{tabular}
\caption{Simulation results details for base case ($j=1,2$ are repeat volunteers and $j=3$ are one-time volunteers). }
\label{tab:simdetailsbase}
\end{table}

\begin{table} [tbp]
\centering
\begin{tabular}{ccccccccccc}
Policy  &  $\gamma_j$   &  Act. cost & Idle cost & Total cost & Idle  & Abandon & Act1 usage &
   Act2 usage & Act3 usage & Act4 usage \\ \hline
Static (1,2)& 0.005 &  \$2756 & \$175 & \$2931 & 0.0045\% & 0.7909\% & 100\% & 100\% & 0\% &0\% \\ 
Dynamic & 0.005 &  \$1930 & \$301 & \$2231 & 0.0077\% & 0.5879\% & 96\% & 41\% & 19\% &8\% \\
Static (1,2,3)& 0.01 &  \$3476 & \$504 & \$3980 & 0.0129\% & 1.3056\% & 100\% & 100\% & 100\% &0\% \\
Dynamic& 0.01 &  \$2906 & \$574 & \$3480 & 0.0147\% & 1.1099\% & 100\% & 75\% & 40\% &18\% \\

Static (1,2,3,4)& 0.02 &  \$5276 & \$1822 & \$7098 & 0.0467\% & 2.0076\% & 100\% & 100\% & 100\% &100\% \\
Dynamic& 0.02 &  \$4744 & \$2166 & \$6910 & 0.0555\% & 1.9127\% & 100\% & 100\% & 94\% &73\% \\ \hline

\end{tabular}
\caption{Simulation results details for $j=1,2,3$ are repeat volunteers. }
\label{tab:simdetailsrepeat}
\end{table}

\begin{table} [tbp]
\centering
\begin{tabular}{ccccccccccc}
Policy  &  $\gamma_j$   &  Act. cost & Idle cost & Total cost & Idle  & Abandon & Act1 usage &
   Act2 usage & Act3 usage & Act4 usage \\ \hline
Static (none)& 0.005 &  \$0 & \$526 & \$526 & 0.0135\% & 0.9486\% & 0\% & 0\% & 0\% &0\% \\ 
Dynamic & 0.005 &  \$77 & \$0 & \$77 & 0.0000\% & 0.9934\% & 7\% & 1\% & 0\% &0\% \\
Static (1)& 0.01 &  \$936 & \$138 & \$1074 & 0.0035\% & 1.7165\% & 100\% & 0\% & 0\% &0\% \\
Dynamic& 0.01 &  \$534 & \$55 & \$589 & 0.0014\% & 1.3120\% & 38\% & 7\% & 4\% &1\% \\

Static (1,2)& 0.02 &  \$2756 & \$186 & \$2942 & 0.0048\% & 2.6708\% & 100\% & 100\% & 0\% &0\% \\
Dynamic& 0.02 &  \$1823 & \$379 & \$2202 & 0.0097\% & 2.1349\% & 98\% & 35\% & 17\% &8\% \\ \hline

\end{tabular}
\caption{Simulation results details for $j=1,2,3$ are one-time volunteers. }
\label{tab:simdetailsonetime}
\end{table}

\end{landscape}

\subsection{Simulation robustness checks} \label{sub:simrobust}

\subsubsection{Transition analysis} \label{sub:transitionsim}

We compare two simulations to investigate the performance of the system during the transition when the food bank switches from the best static policy to the dynamic policy. In each simulation, there is a warmup period of 20~years (Years~1-20), a period of 5 years (Period~1, starting in Year~21), and a period of 25~years (Period~2, starting in Year~26). In the first simulation (Sim~A), we use the dynamic policy from Section~\ref{sec:numerical} throughout, i.e., during the warmup period, Period~1, and Period~2. In the second simulation (Sim~B), we use the best static policy from Section~\ref{sec:numerical} during the warmup period and Period~1, but we switch to the dynamic policy in Period~2 (see Table~\ref{tab:sima-simb}). 
\begin{table} [htbp]
\centering
\begin{tabular}{c|c|c|c|c}
   & Years & Quarters, $q$ & Sim A policy & Sim B policy \\ \hline
      warmup & $1-20$ & $1-80$ & dynamic & static \\
      Period 1 & $21-25$ & $81-100$ & dynamic & static \\
      Period 2 & $26-50$ & $101-200$ & dynamic &dynamic \\ \hline
\end{tabular}
\caption{Structure of Sim A and Sim B. }
\label{tab:sima-simb}
\end{table}

We collect metrics on a quarterly basis to analyze changes in the transition period on a more granular timeframe. Quarter numbers start at $q=1$ at the beginning of Year~1 and increment by one each quarter until $q=200$ at the end of Year~50. Each simulation is replicated~50 times. We use common random numbers for Sim~A and Sim~B. That is, each replication uses a different seed, but Sim~A and Sim~B for the same replication number use the same seed.

Figure~\ref{fig:queuelength} shows the system state (i.e., queue length) by quarter for Period~1 and the first 5~years of Period~2 (the first 5~years of Period~2 are sufficient to show the behaviour during the transition). We see that that in Period~1, Sim~B's average queue length (in the range of $420-450$ volunteers) is consistently and noticeably higher than that of Sim~A (in the range of $386-405$ volunteers). Sim~B is using the static policy in Period~1, which overuses activities to engage volunteers, thereby resulting in more arrivals and longer queues. At the start of Period~2 ($q=101$), Sim~B switches to the dynamic policy and immediately, the average queue length decreases. In $q=101$, the difference in average queue length between Sim~B and Sim~A is 14 volunteers, but from quarter $q=102$ on, the average queue lengths of the two simulations are virtually identical. This suggests that it takes one quarter (or less) for the system state to reach steady state after we switch from the static policy to the dynamic policy. 

\begin{figure}[htbp]
\centering
\includegraphics[scale=0.5]{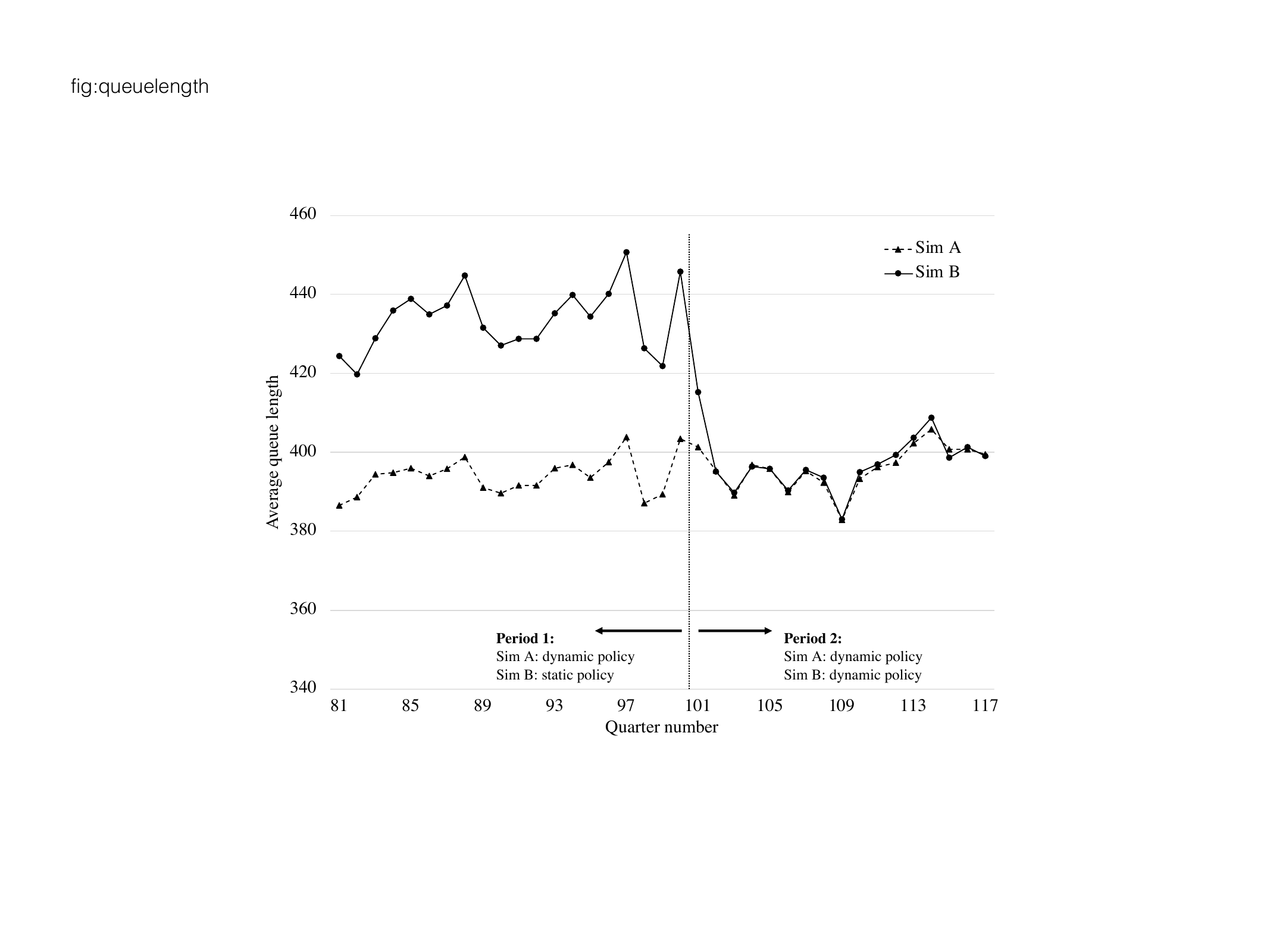}
\caption{Average queue length for Period 1 ($q=81, \dots, 100$) and the first 5 years of Period 2 ($q=101, \dots , 120$).  }
\label{fig:queuelength}
\end{figure}

The average abandonment percentage (Figure~\ref{fig:abandon}) tells a story that is consistent with the average queue length. Abandonments in Period~1 are higher for Sim~B because the queue length is higher. When Sim~B switches to the dynamic policy in Period~2 ($q=101$), there is one quarter of transition time and from $q=102$ onwards, the average abandonment percentages of Sim~A and Sim~B are virtually identical.

\begin{figure}[htbp]
\centering
\includegraphics[scale=0.5]{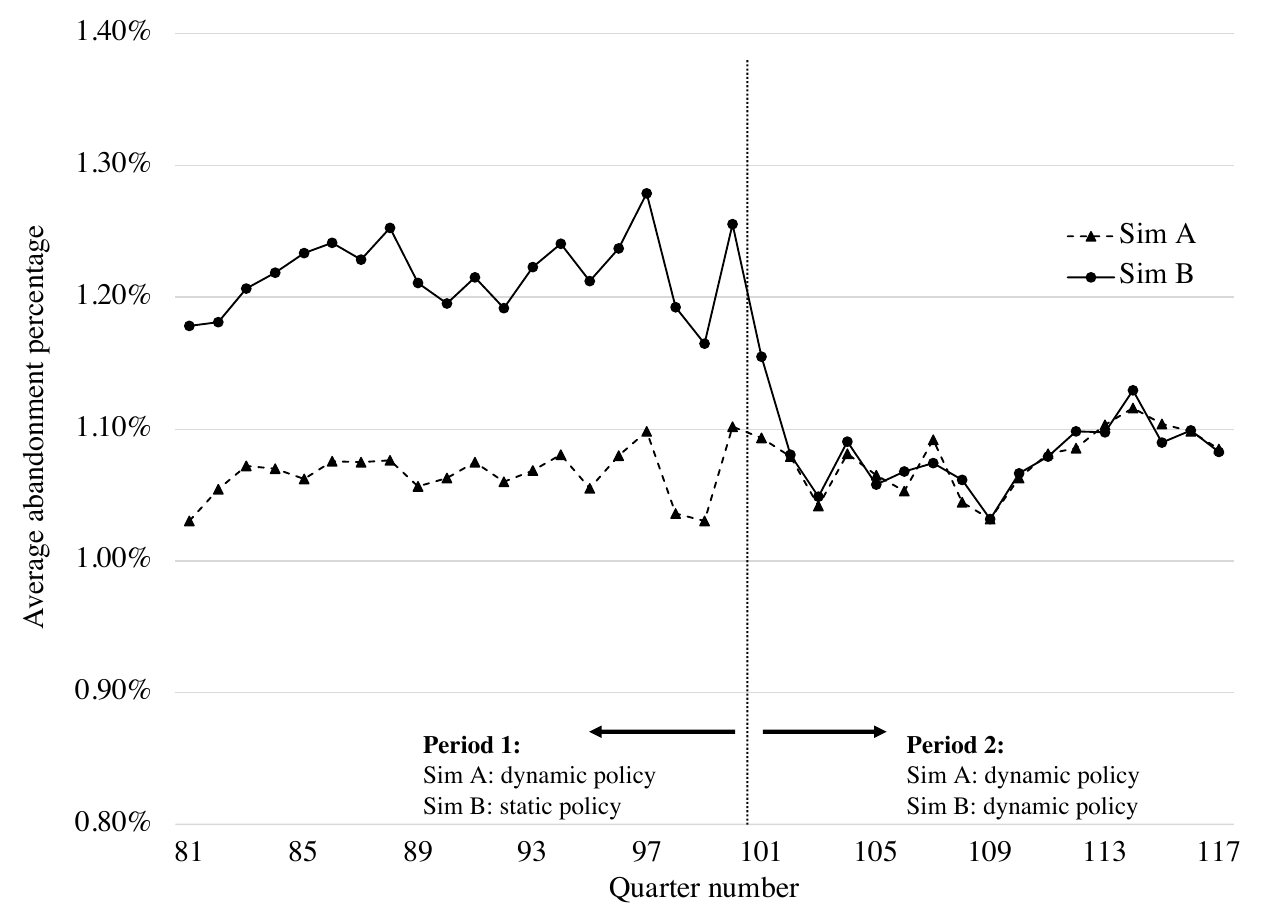}
\caption{Average abandonment percentage for Period 1 ($q=81, \dots, 100$) and the first 5 years of Period 2 ($q=101, \dots , 120$).  }
\label{fig:abandon}
\end{figure}

\begin{figure}[htbp]
\centering
\includegraphics[scale=0.5]{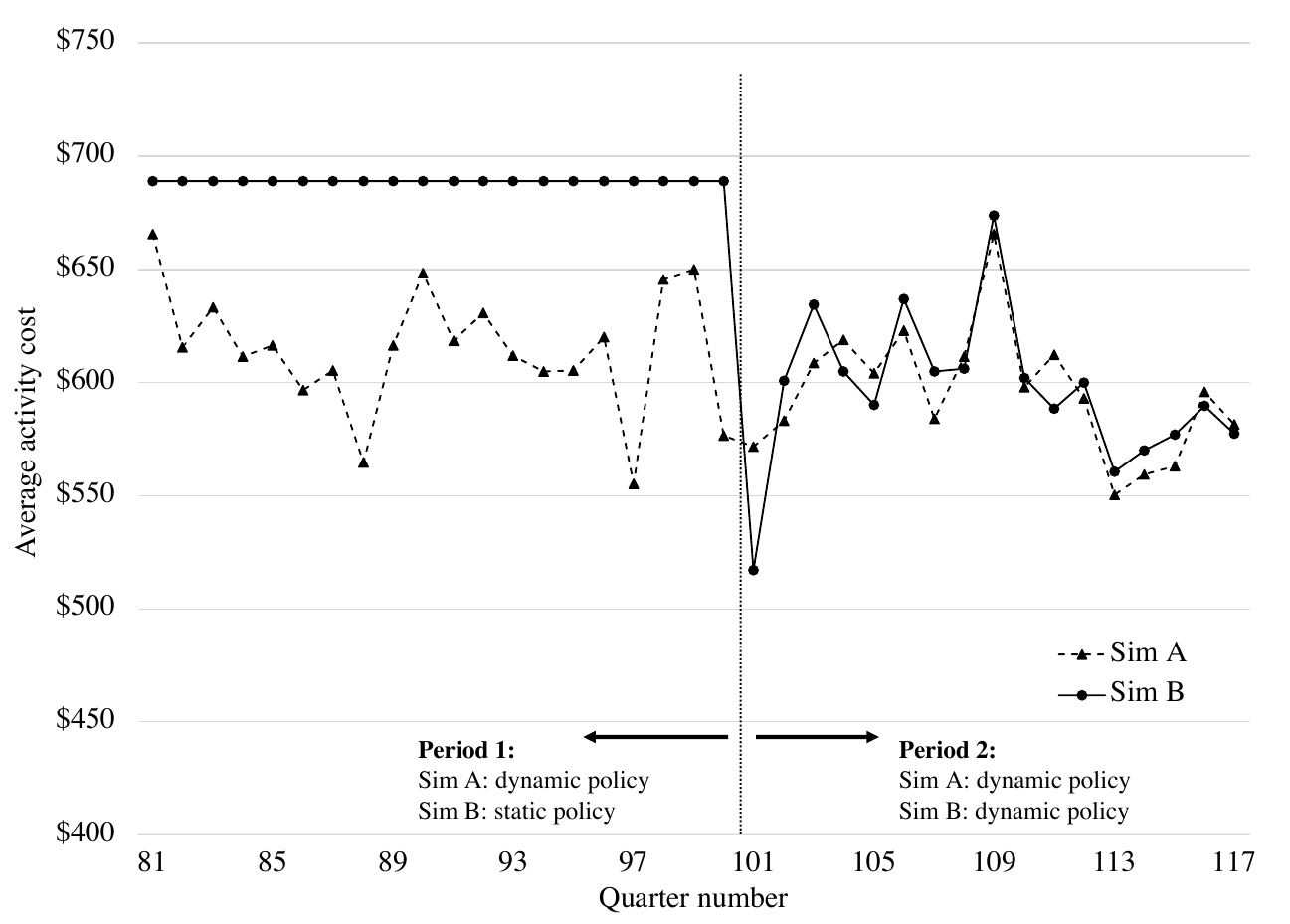}
\caption{Average activity cost for Period 1 ($q=81, \dots, 100$) and the first 5 years of Period 2 ($q=101, \dots , 120$).  }
\label{fig:actcost}
\end{figure}

In terms of cost, because the policy change is immediate, the change in activity cost is immediate. Figure~\ref{fig:actcost} shows the average activity cost per quarter for Period~1 and the first 5~years of Period~2.  In Period~1, we see that Sim~B has a constant activity cost that is consistently higher than that of Sim~A. This accounts for the longer queue in Period~1 -- Sim~A deploys more activities while using the static policy. At the start of Period~2 (Sim~B switches to the dynamic policy), Sim~B's activity cost drops immediately and thereafter closely tracks with the activity cost of Sim~A. The activity cost in $q=101$ is particularly low because the queue length carried over from the static policy in $q=100$ is high. Therefore, to draw down the queue to the optimal length, the usage of activities is particularly low. 

Figures~\ref{fig:queuelength}-\ref{fig:actcost} collectively suggest that the convergence to the new steady state under a policy change happens quickly. Additionally, hypothesis testing shows that average queue lengths, abandonment percentages, and activity costs for Sim~A and Sim~B converge in one quarter. This analysis is available from the authors upon request.

%
%

\subsubsection{Random number of volunteer slots per day} \label{sub:randomservice}

In Section~\ref{sec:numerical}, we set the number of volunteer slots per day to be constant at 250 slots per day. However, it's possible that there may be slow days or busier days when fewer or more meals need to be assembled (based on the client organizations served on a particular day). To capture this, we assume the number of volunteer slots on a given day is a random variable distributed according to the following discrete distribution:
\begin{align}
\mbox{number of volunteer slots per day} = \left\{
   \begin{array}{ll}
      a_1, & p=0.25 \\
      a_2, & p=0.5 \\
      a_3, & p=0.25
   \end{array}
   \right. \label{eqn:slotdist}
\end{align}
We simulate two scenarios with differing parameter values for distribution~(\ref{eqn:slotdist}); see Table~\ref{tab:slots}. Both scenarios preserve an average of 250 volunteer slots per day, but the variance of Scenario~2 is higher than Scenario~1. 

\begin{table} [htbp]
\centering
\begin{tabular}{c|c|c|c|c}
      & $a_1$ & $a_2$ & $a_3$ & \mbox{variance} \\ \hline
      Scenario 1 & $240$ & $250$ & 260 & 50 \\
      Scenario 2 & $230$ & $250$ & 270 & 200 \\ \hline
\end{tabular}
\caption{Two scenarios with differing parameter values for distribution~(\ref{eqn:slotdist}). }
\label{tab:slots}
\end{table}

The results for Scenario~1 are given in Table~\ref{tab:simrobust1}, Rows~3-4. We see that the total cost under the dynamic policy (\$3372) is 9.5\% lower than the total cost under the static policy (\$3726). Although the idle cost under the dynamic policy increases by $\$993-\$250 = \$742$, the dynamic policy uses activities more effectively, thus reducing the activity cost by $\$3476-\$2379 = \$1097$. The results for Scenario~2 are directionally the same (see Table~\ref{tab:simrobust1}, Rows 5-6). The total cost under the dynamic policy (\$4766) is 2.5\% lower than under the static policy (\$4886). The idle cost increases by $\$2099-\$1410 = \$689$, but the activity cost decreases by $\$3476-\$2667 = \$809$. The results in Scenarios~1 and~2 are consistent with the base case in Section~\ref{sec:numerical}, in that the dynamic policy performs better than the static policy (i.e., lower total cost). However, in Scenario~2, where the variance is higher, the improvement in performance is lower in magnitude than in Scenario~1. This result is consistent with our general intuition that higher uncertainty degrades system performance.

\subsubsection{Time to abandonment distributed according to Gamma distribution} \label{sub:abandongamma}

In our model, we assumed the time to abandon is exponentially distributed with mean $1/\gamma_j$, with $\gamma_j = 0.01$ in Section~\ref{sec:numerical}. We perform a robustness check here to see the impact of changing this assumption by using a Gamma$(\alpha, \beta)$ distribution. By using the Gamma distribution, we include our base case exponential distribution, which is equivalent to Gamma$(1, \gamma_j)$, and are able to study the impact of changing the variance without changing the mean. We simulate two scenarios with differing Gamma distribution parameters, as shown in Table~\ref{tab:gamma}. Both scenarios preserve the mean of $1/\gamma_j$ (to be consistent with the base case). However, Scenario~1 has a lower variance than the base case and Scenario~2 has a higher variance than the base case.

\begin{table} [htbp]
\centering
\begin{tabular}{c|c|c|c}
      & $\alpha$ & $\beta$ & variance \\ \hline
      Scenario 1 & 2 & $2 \gamma_j$ & $\frac{1}{2 \gamma_j^2}$ \\
      Scenario 2 & 0.9 & $0.9 \gamma_j$ & $\frac{1}{0.9 \gamma_j^2}$  \\ \hline
\end{tabular}
\caption{Two scenarios with differing Gamma distribution parameters. }
\label{tab:gamma}
\end{table}

The results for Scenario~1 are given in Table~\ref{tab:simrobust2}, Rows 3-4. We see that the total cost under the dynamic policy (\$853) is 69\% lower than under the static policy (\$2756). The static policy is very inefficient at compensating for the abandonments in this scenario. Although the dynamic policy increases idle cost by $\$38-\$0=\$38$, it has the flexibility to deploy activities in alignment with when abandonments occur, and thus it reduces the activity cost by $\$2756-\$815 = \$1941$. In Scenario~2, the dynamic policy also performs better than the static policy (see Table~\ref{tab:simrobust2}, Rows 5-6), but the improvement is lower in magnitude than Scenario~1. The total cost under the dynamic policy (\$4605) is 8.7\% lower than under the static policy (\$5043). In this scenario, the dynamic policy decreases idle cost by $\$1567-\$1188=\$379$ and decreases activity cost by $\$3476-\$3417 = \$59$. The results in Scenarios~1 and~2 are consistent with the base case in Section~\ref{sec:numerical}, in that the dynamic policy performs better than the static policy (i.e., lower total cost). However, in Scenario~2, where the variance is higher, the improvement in performance is lower in magnitude than in Scenario~1. This result is consistent with our general intuition that higher uncertainty degrades system performance.

\begin{landscape}
\begin{table} [tbp]
\centering
\begin{tabular}{lccccccccc}
  &  Act.  & Idle  & Total  &   &   & Act1  &
   Act2  & Act3  & Act4  \\ 
Simulation description &   cost &  cost &  cost & Idle  & Abandon  &  usage &
    usage &  usage &  usage \\\hline
Base case - static &  \$3476 & \$252 & \$3728 & 0.01\% & 1.38\% & 100\% & 100\% & 100\% &0\% \\
Base case - dyn &  \$2495 & \$629 & \$3123 & 0.02\% & 1.07\% & 99\% & 60\% & 31\% &14\% \\
Scenario 1, $\{a_1=240, a_2=250, a_3=260\} $ - static &  \$3476 & \$250 & \$3726 & 0.01\% & 1.44\% & 100\% & 100\% & 100\% &0\% \\
Scenario 1, $\{a_1=240, a_2=250, a_3=260\} $ - dyn &  \$2379 & \$993 & \$3372 & 0.03\% & 1.09\% & 98\% & 55\% & 28\% &14\% \\ 
Scenario 2, $\{a_1=230, a_2=250, a_3=270\}$ - static &  \$3476 & \$1410 & \$4886 & 0.04\% & 1.44\% & 100\% & 100\% & 100\% &0\% \\
Scenario 2, $\{a_1=230, a_2=250, a_3=270\}$ - dyn &  \$2667 & \$2099 & \$4766 & 0.05\% & 1.18\% & 99.5\% & 64\% & 35\% &17\% \\  \hline
\end{tabular}
\caption{Robustness check: Simulation results when the number of volunteer slots is distributed according to~(\ref{eqn:slotdist}). }
\label{tab:simrobust1}
\end{table}

\begin{table} [tbp]
\centering
\begin{tabular}{lccccccccc}
 &  Act.  & Idle  & Total  &   &   & Act1  &
   Act2  & Act3  & Act4  \\ 
Simulation description   &   cost &  cost &  cost & Idle  & Abandon  &  usage &
    usage &  usage &  usage \\\hline
Base case - static &  \$3476 & \$252 & \$3728 & 0.01\% & 1.38\% & 100\% & 100\% & 100\% &0\% \\
 Base case - dyn &  \$2495 & \$629 & \$3123 & 0.02\% & 1.07\% & 99\% & 60\% & 31\% &14\% \\
Scenario 1, Gamma$(2, 2\gamma_j)$ - static &  \$2756 & \$0 & \$2756 & 0.00\% & 0.17\% & 100\% & 100\% & 0\% &0\% \\
Scenario 1, Gamma$(2, 2\gamma_j)$  - dyn &  \$ 815& \$38 & \$853 & 0.00\% & 0.04\% &59\% & 11\% & 5\% & 2\% \\ 
Scenario 2, Gamma$(0.9, 0.9 \gamma_j)$  - static &  \$3476 & \$1567 & \$5043 & 0.04\% & 1.61\% & 100\% & 100\% & 100\% &0\% \\
Scenario 2, Gamma$(0.9, 0.9 \gamma_j)$ - dyn &  \$3417 & \$1188 & \$4605 & 0.03\% & 1.52\% & 100\% & 87\% & 54\% & 29\% \\  \hline
\end{tabular}
\caption{Robustness check: Simulation results when the time to abandon is distributed according to Gamma$(\alpha, \beta)$ distribution.}
\label{tab:simrobust2}
\end{table}

\end{landscape}

\end{document}